\documentclass[a4paper,USenglish,cleveref,autoref,thm-restate,nolineno]{socg-lipics-v2021} 

\newif\iflong
\newif\ifshort

\hideLIPIcs

\longtrue

\iflong
\else
\shorttrue
\fi

\usepackage{verbatim}
\usepackage[textsize=footnotesize]{todonotes}
\usepackage{xcolor}
\usepackage[textsize=footnotesize]{todonotes}

\usepackage{etex}
\usepackage{enumerate}
\usepackage{graphics}
\usepackage{xspace}
 
\usepackage{caption}
\usepackage{subcaption}
\usepackage{epsfig}
\usepackage{arcs}
\usepackage[T1]{fontenc}
\usepackage{alltt}
\usepackage{amssymb,amsfonts,epsf,url}
\usepackage{graphicx}
\usepackage{pstricks,pstricks-add}
\usepackage{pst-node}
\usepackage{pst-coil}
\usepackage{amsmath}
\usepackage{amsthm}
\usepackage{verbatim}
\usepackage[textsize=footnotesize]{todonotes}
\usepackage{xcolor}
\usepackage[textsize=footnotesize]{todonotes}
\usepackage{fancyhdr}
\usetikzlibrary{arrows}

\newcommand{\cmpm}{\textsc{\textup{CMP-M}}\xspace} 
\newcommand{\cmpl}{\textsc{\textup{CMP-L}}\xspace}

\theoremstyle{plain}
 
\RequirePackage{etex}

\newcommand{\bigoh}{\mathcal{O}}

\newtheorem{longtheorem}{Theorem}
\newtheorem{longlemma}[longtheorem]{Lemma}
\newtheorem{longdefinition}[longtheorem]{Definition}
\newtheorem{longobservation}[longtheorem]{Observation}
\newtheorem{ourclaim}{Claim}

\newlength{\alginputwidth}

\newcommand{\YES}{\textup{\textsf{YES}}}

\newcommand{\nat}{\mathbb{N}}

\newcommand{\Pol}{\mbox{\sf P}}

\newcommand{\NP}{\mbox{{\sf NP}}}
\newcommand{\APX}{\mbox{{\sf APX}}}
\newcommand{\FPT}{\mbox{{\sf FPT}}}

\newcommand{\W}{\mbox{{\sf W}}}

\newcommand{\R}{\mathcal{R}}
\newcommand{\SSS}{\mathcal{S}}

\newcommand{\LA}{\mathcal{LA}}
\newcommand{\SM}{\mathcal{SM}}

\newcommand{\BPSAT}{\textsc{4-Bounded Planar 3-SAT}}
 
\newcommand{\I}{\mathcal{I}}

\newcommand{\normalproblem}[3]{\noindent  
{\sc #1}
\\
{\bf Given:} #2\\
{\bf Question:} #3
 
\medskip
}

 \title{The Parameterized Complexity of Coordinated Motion Planning}  
\titlerunning{The Parameterized Complexity of Coordinated Motion Planning} 

\author{Eduard Eiben}{Department of Computer Science, Royal Holloway, University of London, Egham, UK}{eduard.eiben@gmail.com}{https://orcid.org/0000-0003-2628-3435}{}

\author{Robert Ganian}{Algorithms and Complexity Group, TU Wien, Vienna, Austria}{rganian@gmail.com}{https://orcid.org/0000-0002-7762-8045}{Project No. Y1329 of the Austrian Science Fund (FWF), Project No. ICT22-029 of the Vienna Science Foundation (WWTF)}

\author{Iyad Kanj}{School of Computing, DePaul University, Chicago, USA}{ikanj@cdm.depaul.edu}{0000-0003-1698-8829}{DePaul URC Grants 606601 and 350130}

\authorrunning{E.\ Eiben, R.\ Ganian, I.\ Kanj } 
\Copyright{Eduard Eiben, Robert Ganian, Iyad Kanj} 

\category{} 
\relatedversion{} 

\acknowledgements{}

\EventEditors{Erin W. Chambers and Joachim Gudmundsson} \EventNoEds{2} \EventLongTitle{39th International Symposium on Computational Geometry (SoCG 2023)} \EventShortTitle{SoCG 2023} \EventAcronym{SoCG} \EventYear{2023} \EventDate{June 12--15, 2023} \EventLocation{Dallas, Texas, USA} \EventLogo{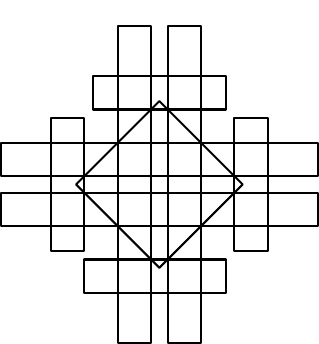} \SeriesVolume{258} \ArticleNo{XX}  
\begin{document}

\maketitle
 
\begin{abstract}
 In Coordinated Motion Planning (CMP), we are given a rectangular-grid on which $k$ robots occupy $k$ distinct starting gridpoints and need to reach $k$ distinct destination gridpoints. In each time step, any robot may move to a neighboring gridpoint or stay in its current gridpoint, provided that it does not collide with other robots. The goal is to compute a schedule for moving the $k$ robots to their destinations which minimizes a certain objective target---prominently the number of time steps in the schedule, i.e., the makespan, or the total length
 traveled by the robots.
 We refer to the problem arising from minimizing the former objective target as \cmpm and the latter as \cmpl. 
Both \cmpm and \cmpl are fundamental problems that were posed as the computational geometry challenge of SoCG 2021, and CMP also embodies the famous $(n^2-1)$-puzzle as a special case.

In this paper, we settle the parameterized complexity of \cmpm and \cmpl with respect to their two most fundamental parameters: the number of robots, and the objective target. We develop a new approach to establish the fixed-parameter tractability of both problems under the former parameterization that relies on novel structural insights into optimal solutions to the problem. When parameterized by the objective target, we show that \cmpl remains fixed-parameter tractable while \cmpm becomes para-\NP-hard. The latter result is noteworthy, not only because it improves the previously-known boundaries of intractability for the problem, but also because the underlying reduction allows us to establish---as a simpler case---the \NP-hardness of the classical Vertex Disjoint and Edge Disjoint Paths problems with constant path-lengths on grids. 

\end{abstract}

\section{Introduction}
\label{sec:intro}

Who among us has not struggled through solving the 15-puzzle? Given a small square board, tiled with 15 tiles numbered $1, \ldots, 15$, and a single hole in the board, the goal of the puzzle is to slide the tiles in order to reach the final configuration in which the tiles appear in (sorted) order; see Figure~\ref{fig:puzzle} for an illustration.
The 15-puzzle has been generalized to an $n \times n$ square-board, with tiles numbered $1, \ldots, n^2-1$.  Unsurprisingly, this generalization is called the $(n^2-1)$-puzzle. Whereas deciding whether a solution to an instance of the $(n^2-1)$-puzzle exists (i.e., whether it is possible to sort the tiles starting from an initial configuration) is in $\Pol$~\cite{spirakis}, determining whether there is a solution that requires at most $\ell \in \nat$ tile moves has been shown to be \NP-hard~\cite{nphard2,nphard1}. 

\begin{figure}[htbp]
\centering
	
	\includegraphics[width=0.3\linewidth]{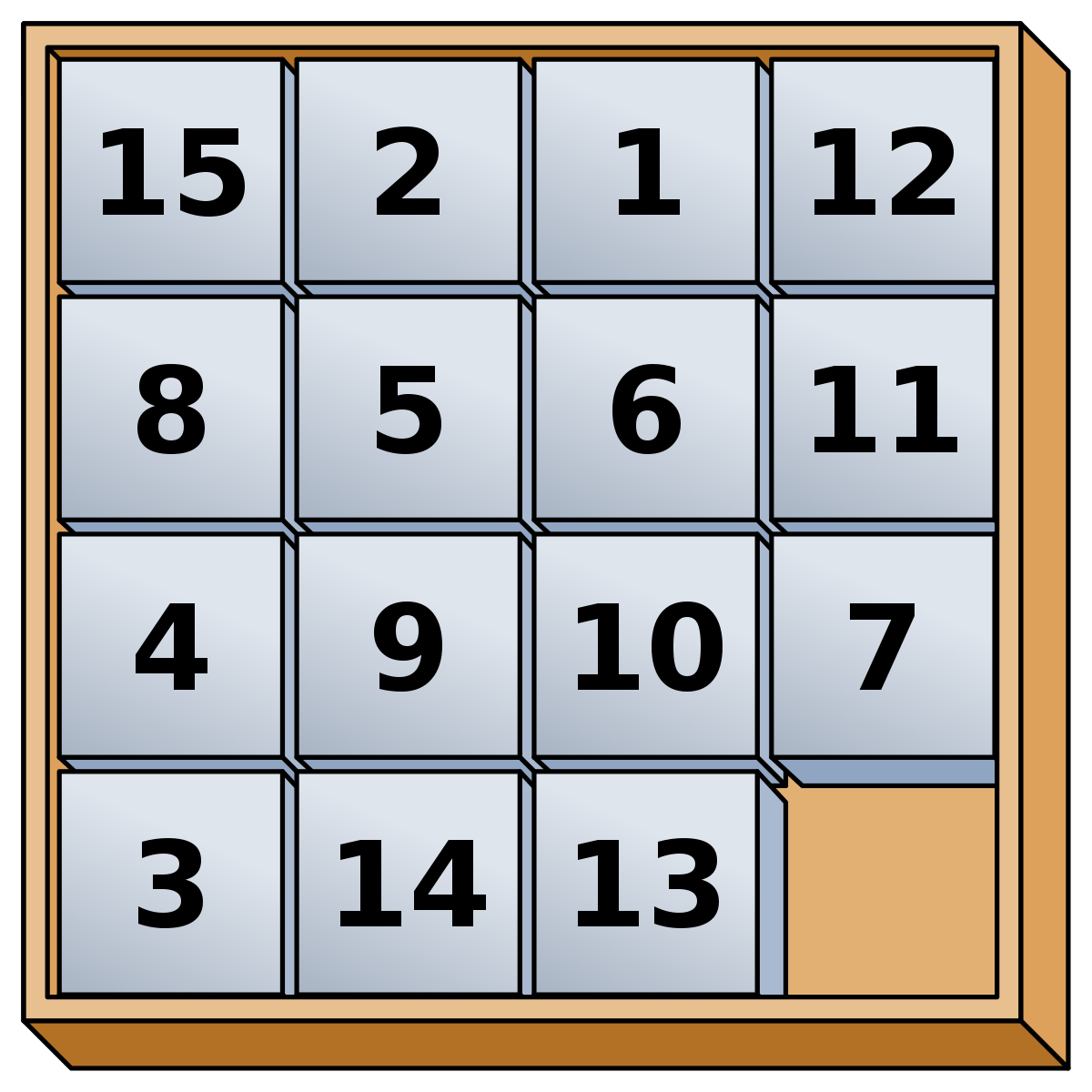}  
	\includegraphics[width=0.3\linewidth]{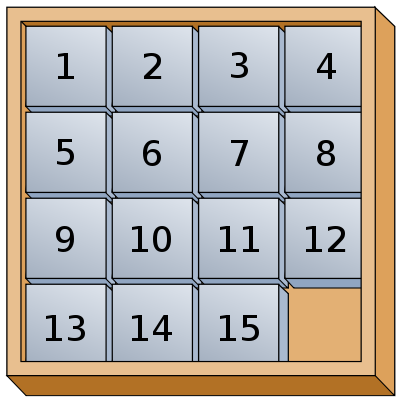} 
	
	\caption{The left figure shows an initial configuration of the 15-puzzle and the right figure shows the desirable final configuration. Source: \url{https://en.wikipedia.org/wiki/15\_puzzle}.}
	\label{fig:puzzle}
\end{figure}

Deciding whether an $(n^2-1)$-puzzle admits a solution is a special case of Coordinated Motion Planning (CMP), a prominent task originating from robotics which has been extensively studied in the fields of Computational Geometry and Artificial Intelligence (where it is often referred to as Multi-Agent Path Finding).
In CMP, we are given an $n\times m$ rectangular-grid on which $k$ robots occupy $k$ distinct starting gridpoints and need to reach $k$ distinct destination gridpoints. Robots may move simultaneously at each time step, and at each time step, a robot may move to a neighboring gridpoint or stay in its current gridpoint provided that (in either case) it does not collide with any other robots; two robots collide if they are occupying the same gridpoint at the end of a time step, or if they are traveling along the same grid-edge (in opposite directions) during the same time step.
We are also given an objective target, and the goal is to compute a schedule for moving the $k$ robots to their destination gridpoints which satisfies the specified target. The two objective targets we consider here are (1) the number of time steps used by the schedule (i.e., the makespan), and (2) the total length traveled by all the robots (also called the ``total energy'', e.g., in the SoCG 2021 Challenge~\cite{socg2021}); the former gives rise to a problem that we refer to as \cmpm, while we refer to the latter as \cmpl. An illustration is provided in Figure~\ref{fig:instance}.

\begin{figure}[htbp]
\begin{minipage}{.35\linewidth}
 \includegraphics[width=\linewidth]{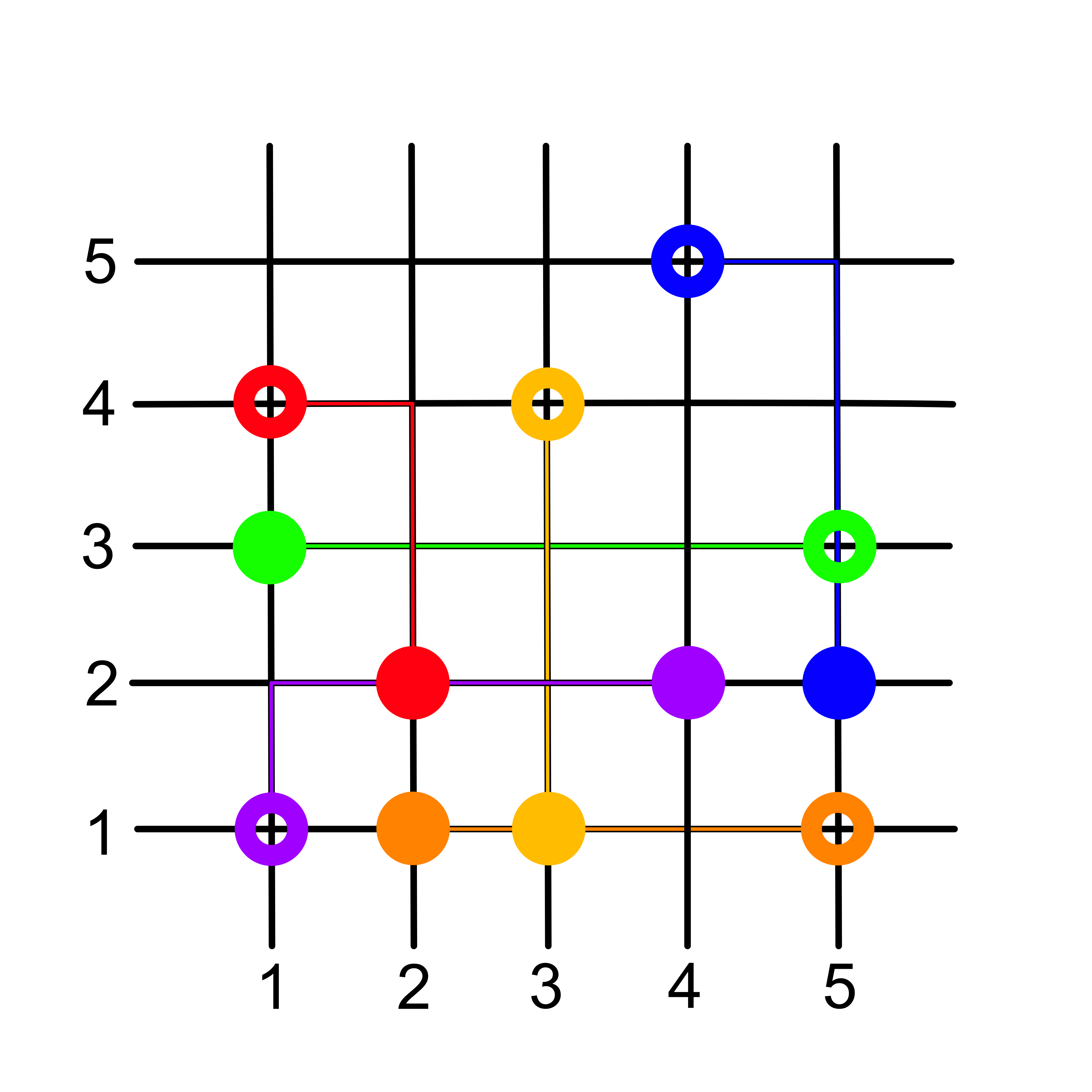} 
 \end{minipage}
\begin{minipage}{.5\linewidth}
\includegraphics[scale=0.95]{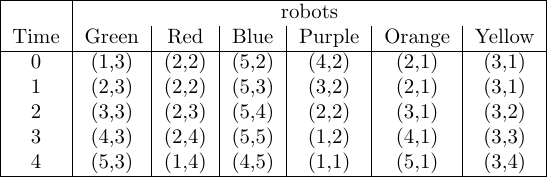}
\end{minipage}
\caption{An illustration (left) of an instance of \cmpm with six robots, indicated using distinct colors (blue, green, yellow, red, orange, purple), and a makespan $\ell=4$. The starting points are marked using a disk shape (filled circle) and destination points using an annular shape. A schedule indicating each of the robot's position at each of the four time steps is shown in the table (right).}
	\label{fig:instance}
\end{figure}

In this paper, we settle the parameterized complexity of \cmpm and \cmpl with respect to their two most fundamental parameters: the number $k$ of robots, and the objective target. In particular, we obtain fixed-parameter algorithms for both problems when parameterized by $k$ and for \cmpl when parameterized by the target, but show that \cmpm remains \NP-hard even for fixed values of the target.
 Given how extensively CMP has been studied in the literature (see the related work below), we consider it rather surprising that fundamental questions about the problem's complexity have remained unresolved. We believe that one aspect contributing to this gap in our knowledge was the fact that, even though the problems seem deceptively easy, it was far from obvious how to obtain exact and provably optimal algorithms in the parameterized setting. 
Furthermore, en route to the aforementioned intractability result, we establish the \NP-hardness of the classical \textsc{Vertex Disjoint Paths} and \textsc{Edge Disjoint Paths} problems on grids when restricted to bounded-length paths.

\subsection{Related Work}
\label{subsec:relatedwork}
CMP has been extensively studied by researchers in the fields of computational geometry, AI/Robotics, and theoretical computer science in general. In particular, \cmpm and \cmpl were posed as the Third Computational Geometry Challenge of SoCG 2021, which took place during the Computational Geometry Week in 2021~\cite{socg2021}. The CMP problem generalizes the $(n^2-1)$-puzzle, which was shown to be \NP-hard as early as 1990 by Ratner and Warmuth~\cite{nphard1}. A simpler \NP-hardness proof was given more recently by Demaine et al.~\cite{nphard2}.  Several recent papers studied the complexity of CMP with respect to optimizing various objective targets, such as: the makespan, the total length traveled, the maximum length traveled (over all robots), and the total arrival time~\cite{banfi,demaine,halprin,yu}.
The continuous geometric variants of CMP, in which the robots are modeled as geometric shapes (e.g., disks) in a Euclidean environment, have also been extensively studied~\cite{halprinunlabeled,demaine,survey,alagar,sharir}. Finally, we mention that there is a plethora of works in the AI and Robotics communities dedicated to variants of the CMP problem, for both the continuous and the discrete settings~\cite{heuristic3,heuristic2,heuristic1,heuristic5,heuristic4,lavalle,lavalle1}.

The fundamental vertex and edge disjoint paths problems have also been thoroughly studied, among others due to their connections to graph minors theory.
The complexity of both problems on grids was studied as early as in the 1970's motivated by its applications in VLSI design~\cite{frank,kramer,marx,suzuki}, with more recent results focusing on approximation~\cite{chuzhoy1,chuzhoy2}.

  \subsection{High-Level Overview of Our Results and Contributions}
\label{subsec:results}
As our first set of results, we show that \cmpm and \cmpl are fixed-parameter tractable (\FPT) parameterized by the number $k$ of robots, i.e., can be solved in time $f(k)\cdot n^{\bigoh(1)}$ for some computable function $f$ and input size $n$. 
Both results follow a two-step approach for solving each of these problems. 
In the first step, we obtain a structural result revealing that every \YES-instance of the problem has a \emph{canonical} solution in which the number of ``turns'' (i.e., changes in direction) made by any robot-route is upper bounded by a function of the parameter $k$; this structural result is important in its own right, and we believe that its applications extend beyond this paper. This first step of the proof is fairly involved and revolves around introducing the notion of ``slack'' to partition the robots into two types, and then exploiting this notion to reroute the robots so that their routes form a canonical solution.  In the second step, we show that it is possible to find such a canonical solution (or determine that none exists) via a combination of delicate branching and solving subinstances of Integer Linear Programming (ILP) in which the number of variables is upper bounded by a function of the parameter $k$; fixed-parameter tractability then follows since the latter can be solved in \FPT-time thanks to Lenstra's result~\cite{FrankTardos87,Lenstra83,Kannan87}.

Next, we consider the other natural parameterization of the problem: the objective target. For \cmpl, this means parameterizing by the total length traveled, and there we establish fixed-parameter tractability via exhaustive branching. The situation becomes much more intriguing for \cmpm, where we show that the problem remains \NP-hard even when the target makespan is a fixed constant. 
As a by-product of our reduction, we also establish the \NP-hardness of the classical Vertex and Edge Disjoint Paths problems on grids when restricted to bounded-length paths.
 
The contribution of our intractability results are twofold. First, the \NP-hardness of CMP with constant makespan is the first result showing its \NP-hardness in the case where one of the parameters is a fixed constant. As such, it refines and strengthens several existing \NP-hardness results for CMP~\cite{banfi,demaine,halprin}. It also answers the open questions in~\cite{halprin} about the complexity of the problem in restricted settings where the optimal path of each robot passes through a constant number of starting/destination points, or where the overlap between any two optimal paths is upper bounded by a constant, by directly implying their \NP-hardness. Second, the \NP-hardness results for the bounded-length vertex and edge disjoint paths problems on grids also refine and deepen several intractability results for these problems.  All previous \NP-hardness (and \APX-hardness) results for the vertex and edge disjoint paths problems on grids~\cite{banfi,chuzhoy2,demaine,nphard2,halprin,kramer,marx,nphard1} yield instances in which the path length is unbounded. Last but not least, we believe that the \NP-hardness results we derive are of independent interest, and have the potential of serving as a building block in \NP-hardness proofs for problems in geometric and topological settings, where it is very common to start from a natural problem whose restriction to instances embedded on a grid remains \NP-hard.

\section{Preliminaries and Problem Definition}
\label{sec:prelim}

We use standard terminology for graph theory~\cite{Diestel12} and assume basic familiarity with the parameterized complexity paradigm including, in particular, the notions of \emph{fixed-parameter tractability} and \emph{para-\NP-hardness}~\cite{CyganFKLMPPS15,DowneyFellows13}. For $n \in \nat$, we write $[n]$ for the set $\{1, \ldots, n\}$. 

Let $G$ be an $n \times m$ rectangular grid, where $n, m \in \nat$. Let $\{R_i \mid i \in [k]\}$, $k\in \mathbb{N}$, be a set of robots that will move on $G$. Each $R_i$, $i \in [k]$, is associated with a starting gridpoint $s_i$ and a destination gridpoint $t_i$ in $V(G)$, and hence can be specified as the pair $R_i=(s_i, t_i)$; we assume that all the $s_i$'s are pairwise distinct and that all the $t_i$'s are pairwise distinct, and we denote by $\R=\{(s_i, t_i) \mid i \in [k]\}$ the set of all robots.  At each time step, a robot may either stay at the gridpoint it is currently on, or move to an adjacent gridpoint, and robots may move simultaneously.  We reference the sequence of moves of the robots using a time frame $[0, t]$, $t \in \nat$, and where in time step $x \in [0, t]$ each robot remains stationary or moves.

Let a \emph{route} for $R_i$ be a tuple $W_i=(u_0, \ldots, u_{t})$ of vertices in $G$ such that (i) $u_0=s_i$ and $u_{t}=t_i$ and (ii) $\forall j \in [t]$, either $u_{j-1} =u_j$ or $u_{j-1}u_{j} \in E(G)$.
Intuitively, $W_i$ corresponds to a ``walk'' in $G$, with the exception that consecutive vertices in $W_i$ may be identical (representing waiting time steps), in which $R_i$ begins at its starting point at time step $0$, and is at its destination point at time step $t$. Two routes $W_i=(u_0, \ldots, u_{t})$ and $W_j=(v_0, \ldots, v_{t})$, where $i \neq j \in [k]$, are \emph{non-conflicting} if (i) $\forall r \in \{0, \ldots, t\}$, $u_r \neq v_r$, and (ii) $\nexists r \in \{0, \ldots, t-1\}$ such that $v_{r+1} =u_r$ and $u_{r+1} =v_r$. Otherwise, we say that $W_i$ and $W_j$ \emph{conflict}. Intuitively, two routes conflict if the corresponding robots are at the same vertex at the end of a time step, or go through the same edge (in opposite directions) during the same time step.

A \emph{schedule} $\SSS$ for $\R$ is a set of routes $W_i, i \in [k]$, during a time interval $[0, t]$, that are pairwise non-conflicting. The integer $t$ is called the \emph{makespan} of $\SSS$.  The (\emph{traveled}) \emph{length} of a route  (or its associated robot) within $\SSS$ is the number of time steps $j$ such that $u_j\neq u_{j+1}$, and the \emph{total traveled length} of a schedule is the sum of the lengths of its routes.

We are now ready to define the problems under consideration.

 \normalproblem{{\sc  Coordinated Motion Planning with Makespan Minimization} (\cmpm)}{An $n \times m$ rectangular grid $G$, where $n, m \in \nat$, and a set $\R=\{(s_i, t_i) \mid i \in [k]\}$ of pairs of gridpoints of $G$ where the $s_i$'s are distinct and the $t_i$'s are distinct; $k, \ell \in \nat$.}{Is there a schedule for $\R$ of makespan at most $\ell$?}
 
 The \textsc{Coordinated Motion Planning with Length Minimization} problem (\cmpl) is defined analogously but with the distinction being that, instead of $\ell$, we are given an integer $\lambda$ and are asked for a schedule of total traveled length at most $\lambda$. 
 For an instance $\mathcal{I}$ of \cmpm or \cmpl, we say that a schedule is \emph{valid} if it has makespan at most $\ell$ or has total traveled length at most $\lambda$, respectively.
We remark that even though both \cmpm and \cmpl are stated as decision problems, all the algorithms provided in this paper are constructive and can output a valid schedule (when it exists) as a witness.
 
 We will assume throughout the paper that $k \geq 2$; otherwise, both problems can be solved in linear time. Furthermore, we remark that the inputs can be specified in $\bigoh(k\cdot (\log n + \log m) + \log \ell)$ (or $+\log \lambda$) bits, and our fixed-parameter algorithms work seamlessly even if the inputs are provided in such concise manner. On the other hand, the lower-bound result establishes ``strong'' \NP-hardness of the problem (i.e., also applies to cases where the input contains a standard encoding of $G$ as a graph).
 
For two gridpoints $p=(x_p, y_p)$ and $q=(x_q, y_q)$, the Manhattan distance between $p$ and $q$, denoted $\Delta(p, q)$, is $\Delta(p, q)=|x_p-x_q| + |y_p-y_q|$. 
For two robots $R_i, R_j \in \R$ and a time step $x \in \nat$, denote by $\Delta_x(R_i, R_j)$ the Manhattan distance between the grid points at which $R_i$ and $R_j$ are located at time step $x$. The following notion will be used in several of our algorithms:
 
 \begin{definition}\rm
  \label{def:slack}
 Let $(G, \R, k, \bullet)$ be an instance of \cmpm or \cmpl and let $T=[t_1, t_2]$ for $t_1, t_2\in \mathbb{N}$. For a robot $R_i$ with corresponding route $W_i$, let $u_p$ and $u_q$ be the gridpoints in $W_i$ at time steps $t_1$ and $t_2$, respectively. Define the \emph{slack} of $R_i$ w.r.t.~$T$, denoted $\texttt{slack}_T(R_i)$, as $(t_2-t_1)-\Delta(u_p, u_q)$ (alternatively, $(q-p)-\Delta(u_p, u_q)$). 
 \end{definition}
 
Observe that the slack measures the amount of time (i.e., number of time steps) that robot $R_i$ ``wastes'' when going from $u_p$ to $u_q$ relative to the shortest time needed to get from $u_p$ to $u_q$.  
For a robot $R_i$ with route $W_i$, for convenience we write $\texttt{slack}_T(W_i)$ for $\texttt{slack}_T(R_i)$.
 When dealing with \cmpm, we write $\texttt{slack}(R_i)$ as shorthand for $\texttt{slack}_{[0, \ell]}(R_i)$, and when dealing with \cmpl, we write  $\texttt{slack}(R_i)$ as shorthand for $\texttt{slack}_{[0, \lambda]}(R_i)$.

 \section{CMP Parameterized by the Number of Robots}
 \label{sec:fpt}

In this section, we establish the fixed-parameter tractability of \cmpm and \cmpl parameterized by the number $k$ of robots. 
 
Both results follow the two-step approach outlined in Subsection~\ref{subsec:results}: showing the existence of a canonical solution, and then reducing the problem via branching to a tractable fragment of Integer Linear Programming. These two steps are described for \cmpm in Subsections~\ref{subsec:canonical} and~\ref{sub:ilp}, while Subsection~\ref{subsec:totaltime} shows how the same technique is used to establish the fixed-parameter tractability of \cmpl.

\subsection{Canonical Solutions for \cmpm}
\label{subsec:canonical}

We begin with a few definitions that formalize some intuitive notions such as ``turns''.

Let $W=(u_0, \ldots, u_{\ell})$, where $\ell > 2$, be a route in an $n \times m$ grid $G$, where $n, m \in \nat$. We say that $W$ makes a \emph{turn} at $u_i=(x_i, y_i)$, where $i \in \{1, \ldots, \ell-1\}$, if the two vectors  $\overrightarrow{u_{i-1}u_i}$ and 
$\overrightarrow{u_{i}u_{i+1}}$ have different orientations (i.e., either one is horizontal and the other is vertical, or they are parallel but have opposite directions). We write $\langle u_{i-1}, u_i, u_{i+1} \rangle$ for the turn at $u_i$. 
A turn $\langle u_{i-1}, u_i, u_{i+1} \rangle$ is a \emph{U-turn} if $\overrightarrow{u_{i-1}u_i}= - \overrightarrow{u_{i}u_{i+1}}$; otherwise, it is a \emph{non U-turn}. The \emph{number of turns} in $W$, denoted $\nu(W)$, is the number of vertices in $W$ at which it makes turns. A sequence $M=[u_i, \ldots, u_{j}] $ of consecutive turns 
is said to be \emph{monotone} if all the turns in each of the two alternating sequences $[u_i, u_{i+2}, u_{i+4}, \ldots]$ and $[u_{i+1}, u_{i+3}, u_{i+5}, \ldots]$, in which $M$ can be partitioned, have the same direction (see Figure~\ref{fig:monotone}).

\begin{figure}[htbp]

\centering
\begin{tikzpicture} [scale=0.8]
 
    \draw[help lines,dashed] (0,0) grid (13,8);
     \node at (1, 1)   (a){};
     
          \node at (2.7, 1.2)   (b) {$u_i$};
          \filldraw (3,1) circle (2pt);
          
          \node at (2.5, 2.2)   (b) {$u_{i+1}$};
          \filldraw (3,2) circle (2pt);
          
          \node at (4.5, 2.2)   (c) {$u_{i+2}$};
          \filldraw (5,2) circle (2pt);

          \node at (4.5, 4.2)   (d) {$u_{i+3}$};
          \filldraw (5,4) circle (2pt);

          \node at (6.5, 4.2)   (e) {$u_{i+4}$};
          \filldraw (7,4) circle (2pt);
          
                    \node at (7, 5)   (f) {};

          \node at (9.5, 6.3)   (g) {$u_{j-1}$};
          \filldraw (10,7) circle (2pt);

          \node at (9.7, 7.2)   (h) {$u_{j}$};
          \filldraw (10,6) circle (2pt);

                    \draw[line width =1mm] (1,1) -- (3,1);
                     \draw[line width =1mm] (3,1) -- (3,2);
                      \draw[line width =1mm] (3,2) -- (5,2);

                      \draw[line width =1mm] (5,2) -- (5,4);
                      \draw[line width =1mm] (5,4) -- (7,4);
                       \draw[dashed,line width =0.5mm] (7,4) -- (7,5);
                       \draw[dashed,line width =0.5mm] (9,6) -- (10,6);
                       \draw[line width =1mm] (10,6) -- (10,7);
                           \draw[line width =1mm] (10,7) -- (12,7);
 \end{tikzpicture}
\caption{Illustration of a monotone sequence of consecutive turns.}
\label{fig:monotone}
\end{figure}

 Let $T=[t_1, t_2] \subseteq [0, \ell]$.  We say that a route $W_i$ for $R_i$ has \emph{no slack} in $T$ if 
$\texttt{slack}_T(R_i) =0$; that is, robot $R_i$ does not ``waste'' any time and always progresses towards its destination during $T$. The following observation is straightforward:

\begin{observation}
\label{obs:monotone}
Let $W_i$ be a route for $R_i$ and $T\subseteq [0, \ell]$ be a time interval such that $\texttt{slack}_T(R_i) =0$. The sequence of turns that $W_i$ makes during $T$ is a monotone sequence (and in particular does not include any U-turns).
\end{observation}

Let $W_i=(s_i=u_0, \ldots, u_{t}=t_i)$ be a route for $R_i$ in a valid schedule $\SSS$ of a \YES-instance of \cmpm or \cmpl, and let $W=(u_q, u_{q+1}, \ldots, u_r)$ be the subroute of $W_i$ during a time interval $T\subseteq [0, t]$. We say that a route $W'=(v_q, \ldots, v_r)$ is \emph{equivalent to} $W$ if: (i) $v_q=u_q$ and $v_r=u_r$ (i.e., both routes have the same starting and ending points); (ii) $|W|=|W'|$; and (iii) replacing $W_i$ in $\SSS$ with the route $(s_i=u_0, \ldots, u_{q-1}, v_q, \ldots, v_r, u_{r+1}, \ldots, u_{t}=t_i)$ still yields a valid schedule of the instance. 

\begin{definition}\rm
\label{def:minimal}
 Let ${\cal I}=(G, \R, k, \bullet)$ be a \YES-instance of \cmpm or \cmpl. A valid schedule $\SSS$ for $(G, \R, k, \bullet)$ is \emph{minimal} if the sum of the number of turns made by all the routes in $\SSS$ is minimum over all valid schedules of ${\cal I}$.
 \end{definition}

The following lemma is the building block for the crucial Lemma~\ref{lem:boundedslack}, which will establish the existence of a canonical solution (for a \YES-instance) in which the number of turns made by ``small-slack'' robots is upper bounded by a function of the parameter. This is achieved by a careful application of a ``cell flattening'' operation depicted in Figure~\ref{fig:cellshort}.

\begin{figure}[htbp]
\centering
\begin{subfigure}[b]{.45\textwidth}
\begin{tikzpicture}
 
    \draw[help lines,dashed] (0,0) grid (7,5);
     \node at (1, 1)   (a){};
          \node at (2.7, 1.2)   (b) {$u_p$};
       
 \draw [fill=green] (3,1) rectangle (5,3);
     \node at (2.7, 3.2)   (c) {$u_r$};
          \node at (4.7, 3.2)   (d) {$u_s$};
                    \node at (5, 4)   (e){} ;

                    \draw[line width =1mm] (1,1) -- (3,1);
                     \draw[line width =1mm] (3,3) -- (3,1);
                      \draw[line width =1mm] (3,3) -- (5,3);

                      \draw[line width =1mm] (5,4) -- (5,3);

                         \node at (4, 2)   (d) {$C$};

    \filldraw (3,1) circle (2pt);
    \filldraw (3,3) circle (2pt);
    \filldraw (5,3) circle (2pt);
    
      \node at (0.6, 1.2)   (d) {$u_{p-1}$};
      \filldraw (1,1) circle (2pt);
\end{tikzpicture}
\end{subfigure}
\hspace*{1cm}
\begin{subfigure}[b]{.45\textwidth}
\begin{tikzpicture}
 
    \draw[help lines,dashed] (0,0) grid (7,5);

     \node at (1, 1)   (a){};
          \node at (2.7, 1.2)   (b) {$u_p$};
             \node at (4.7, 3.2)   (d) {$u_s$};
                    \node at (5, 4)   (e){} ;
                     
                    \draw[line width =1mm] (1,1) -- (5,1);
 
                      \draw[line width =1mm] (5,1) -- (5,4);
                        
                        \filldraw (3,1) circle (2pt);
                        \filldraw (5,3) circle (2pt);
                          \node at (0.6, 1.2)   (d) {$u_{p-1}$};
      \filldraw (1,1) circle (2pt);
      \filldraw (5,1) circle (2pt);
     
\end{tikzpicture}
\end{subfigure}
\caption{Illustration of a \emph{cell} in a route (left) and its \emph{flattening} (right).}
\label{fig:cellshort}
\end{figure}

More specifically, we show that if in a solution a robot has no slack during a time interval but its route makes a ``large'' number of turns, then there exists a ``cell'' corresponding to a turn in its route that can be flattened, resulting in another (valid) solution with fewer turns.  

\begin{lemma}
\label{lem:noslackinterval}
Let $\SSS$ be a minimal (valid) schedule for a \YES-instance of \cmpm. Let $W_i$ be a route in $\SSS$ and $T_i \subseteq [0, \ell]$ be a time interval during which $W_i$ has no slack. Then there is an equivalent route, $W'_i$, to $W_i$ such that the number of turns that $W'_i$ makes during $T_i$, $\nu_{T_{i}}(W'_{i})$, satisfies $\nu_{T_{i}}(W'_{i}) \leq  3k^k$. 
\end{lemma}

By carefully subdividing a time interval into roughly $\sigma(k)$ subintervals, for a function $\sigma(k)$ that upper bounds the slack of a robot, and applying Lemma~\ref{lem:noslackinterval} to each of these subintervals, we can extend the result in Lemma~\ref{lem:noslackinterval} to robots whose slack is upper bounded by $\sigma(k)$:

\begin{lemma}
\label{lem:boundedslack}
Let $(G, \R, k, \ell)$ be a \YES-instance of \cmpm, and let $T_i \subseteq [0, \ell]$. Then $(G, \R, k, \ell)$ has a minimal schedule such that, for each $R_i$, $i \in [k]$, satisfying $\texttt{slack}_{T_i}(W_i) \leq \sigma(k)$ for an arbitrary function $\sigma$, its route $W_i$ satisfies $\nu_{T_i}(W_i) \leq \tau(k)$, where $\tau(k)=3k^k(\sigma(k)+1)+ \sigma(k)$. 
\end{lemma}

Lemma~\ref{lem:boundedslack} already provides us with the property we need for ``small-slack'' robots: their number of turns can be upper-bounded by a function of the parameter. We still need to deal with the more complicated situation of ``large-slack'' robots. Our next course of action will be establishing the existence of a sufficiently large time interval during which the ``large-slack'' robots are far from the ``small-slack'' ones. We begin with an observation linking the slack of two robots that ``travel together''.

\begin{observation}
\label{obs:smalllargeslack}
Let $R, R' \in \R$ and let $T=[t_1, t_2] \subseteq [0, \ell]$. Let $u, u'$ be the gridpoints at which $R$ and $R'$ are located at time step $t_1$, respectively, and $v, v'$ those at which $R$ and $R'$ are located at time $t_2$, respectively. Suppose that
$\Delta(u, u') \leq d(k)$ and $\Delta(v, v') \leq d(k)$, for some function $d(k)$. Then $\texttt{slack}_T(R') \leq \texttt{slack}_T(R) + 2d(k)$.  
\end{observation}

Intuitively speaking, the above observation implies that a robot with a large slack in some time interval cannot be close to a robot with a small slack for the whole interval (otherwise, both robots would be moving at ``comparable speeds'', which would contradict that one of them has a small slack and the other a large-slack).
 
Next, we observe that either the slack of all the robots can be upper-bounded by a function $h$, or there is a sufficiently large multiplicative gap between the slack of some robots. This will allow us to partition the set of robots into those with small or large slack. For any function $h$, let $h^{(j)}=\underbrace{h \circ \cdots \circ h}_{j \ \mbox{times}}$ denote the composition of $h$ with itself $j$ times.

\begin{lemma}
\label{lem:slack}
Let $(G, \R, k, \ell)$ be an instance of \cmpm and let $T \subseteq [0, \ell]$. Let $h(k)$ be any computable function satisfying $h^{(p)} (k) \leq h^{(q)}(k)$ for $p \leq q \in [k]$.
Then either $\texttt{slack}_T(R_i) \leq h^{(k)}(k)$ for every $i \in [k]$, or there exists $j \in \nat$ with $2 \leq j \leq k$, such that $\R$ can be partitioned into $(\R_{S}, \R_{L})$ where $\R_{L} \neq \emptyset$, $\texttt{slack}_T(R) \leq h^{(j-1)} (k)$ for every $R \in \R_{S}$, and $\texttt{slack}_T(R') > h^{(j)}(k)$ for every $R' \in \R_{L}$.
\end{lemma}

The next definition yields a time interval with the property that small-slack robots are sufficiently far from large-slack ones during that interval. Such an interval will be useful, since within it we will be able to re-route the large-slack robots (which are somewhat flexible) to reduce the number of turns they make, while avoiding collision with small-slack robots. 

 \begin{definition}\rm
\label{def:intervalslack}
Let $\sigma(k), \gamma(k), d(k)$ be functions such that $\sigma(k) < \gamma(k)$. An interval $T=[t_1, t_2] \subseteq [0, \ell]$ is a $[\sigma, \gamma]$-\emph{good interval} w.r.t.~$d(k)$ if $\R$ can be partitioned into $\R_{S}$ and $\R_{L}$ such that: (i) every
$R \in \R_S$ satisfies $\texttt{slack}_T(R) \leq \sigma(k)$ and every $R' \in \R_L$ satisfies $\texttt{slack}_T(R') \geq \gamma(k)$; (ii) for every time step $t \in T$,  $\Delta_t(R, R') \geq d(k)$ for every $R \in \R_S$ and every $R' \in \R_L$; and (iii) there exists a robot $R_i \in \R_L$ such that $\nu_T(W_i) > 3k^k(\sigma(k)+1)+ \sigma(k)$. If the function $d(k)$ is specified or clear from the context, we will simply say that $T$ is a $[\sigma, \gamma]$-good interval (and thus omit writing ``w.r.t.~$d(k)$''). 
\end{definition}

The following key lemma asserts the existence of a good interval assuming the solution contains a robot that makes a large number of turns:
 
\begin{lemma}
\label{lem:timeinstant}
Let $(G, \R, k, \ell)$ be a \YES-instance of \cmpm and let $\SSS$ be a minimal schedule for $(G, \R, k, \ell)$. If there exists $R' \in \R$ with route $W'$ such that $\nu(W') > 3^{k^3+1} \cdot  (3k^k  \cdot (3^{13^{2k^2-2k}}\cdot k^{13^{2k^2-2k}}+1) +3^{13^{2k^2-2k}}\cdot k^{13^{2k^2-2k}})$, then there exists a $[\sigma, \gamma]$-good interval $T \subseteq [0, \ell]$ w.r.t.~a function $d(k)$ such that 
$k^{13^{k-1}} \leq \sigma(k) \leq 3^{13^{2k^2-2k}}\cdot k^{13^{2k^2-2k}}$, and $d(k)= \gamma(k) = \sigma^{13}(k)$.
\end{lemma}

Once we fix a good interval $T$, we can finally formalize/specify what it means for a robot to have small or large slack within $T$:

\begin{definition}\rm
\label{def:intervalslack2}
Let $T=[t_1, t_2] \subseteq [0, \ell]$ be a $[\sigma, \gamma]$-good interval with respect to some function $d(k)$, where $\sigma(k) < \gamma(k)$ are two functions, and let $R_i \in \R$.  We say that $R_i$ is a \emph{$T$-large slack robot} if 
$\texttt{slack}_T(R_i) \geq \gamma(k)$; otherwise, $\texttt{slack}_T(R_i) \leq \sigma(k)$ and we say that $R_i$ is a \emph{$T$-small slack robot}.
\end{definition}

At this point, we are finally ready to prove Lemma~\ref{lem:reroutingoverall}, which is the core tool that establishes the existence of a solution with a bounded number of turns (w.r.t.~the parameter), even in the presence of large-slack robots: for each solution with too many turns, we can produce a different one with strictly less turns. 
Note that if one simply replaces the routes of large-slack robots so as to reduce their number of turns, then the new routes may bring the large-slack robots much closer to the small-sack robots and hence may lead to collisions. Therefore, the desired rerouting scheme needs to be carefully designed, and it exploits the properties of a good interval: property (i) is used to reorganize and properly reroute these robots, while property (ii) is used to avoid collisions.

\begin{lemma}
\label{lem:reroutingoverall}
Let $(G, \R, k, \ell)$ be a \YES-instance of \cmpm and let $\SSS$ be a minimal schedule for $(G, \R, k, \ell)$. Let $T=[t_1, t_2] \subseteq [0, \ell]$ be a $[\sigma, \gamma]$-good interval with respect to $d(k)$, where $k^{13^{k-1}} \leq \sigma(k) \leq 3^{13^{2k^2-2k}}\cdot k^{13^{2k^2-2k}}$, and $d(k) =\gamma(k) = \sigma^{13}(k)$. For every $T$-large-slack robot $R_i$, there is a route $W'_i$ that is equivalent to $W_i$ and such that $\nu_T(W'_i)$ is at most $3k^3$ and $W'_i$ is identical to $W_i$ in $[0, \ell]\setminus T$. 
 \end{lemma}

\iflong

We say that a robot $R_i \in \R$ \emph{uses} gridline $L$ if there exist two consecutive points in $W_i$ that belong to $L$. We say that a gridpoint $u$ is \emph{to the left} of grid point $v$ if the $x$-coordinate of $u$ is smaller or equal to the $x$-coordinate of $v$.

The following observation is used in proving the existence of a canonical solution in the special case where exactly one of the two grid dimensions is upper bounded by a function of the parameter:

\begin{longlemma}
\label{lem:boundeddimensionordering}
Let $(G, \R, k, \ell)$ be a \YES-instance of \cmpm such that the vertical dimension of $G$ is upper bounded by a function $p(k) > 1$, and let $\SSS$ be a solution for this instance. Let $T=[t_1, t_2] \subseteq [0, \ell]$, and let $R_i$ be a robot such that $\texttt{slack}_T(W_i) \leq \sigma_i(k)$, where $W_i$ is $R_i$'s route in $\SSS$. Let $R_j \in \R$ and suppose that, for each $t \in T$, $\Delta_t(R_i, R_j) \geq d(k)$ for some function $d(k) > \sigma_i(k) + 2p(k)$. Let $p_i, q_i$ be the starting and ending points of $R_i$ at time steps $t_1$ and $t_2$, respectively, and $p_j, q_j$ those for $R_j$. Then 

\begin{itemize}

\item[(i)] the relative order of $p_i, p_j$ is the same as that of $q_i, q_j$; that is, $p_i$ is to the left of $p_j$ if and only if $q_i$ is to the left of $q_j$. Moreover, $\Delta(q_i, q_j) \geq d(k)$. 

\item[(ii)] For any route $W'^{T}_{j}$ for $R_j$ in $T$ between $p_j$ and $q_j$ such that $|W'^{T}_{j}| -\Delta(p_j, q_j) \leq \delta_j(k)$, for some function $\delta_j(k)$, and for any $t \in T$, if at $t$ $R_i$ is at grid point $u_i$ in $W_i$ and $R_j$ is at $u_j$ in $W'_{j}$, then 
$\Delta(u_i, u_j) \geq d(k) -2p(k) - \sigma_i(k) - \delta_j(k)$.
\end{itemize}
\end{longlemma}

\begin{proof}
For (i), we show the statement for the case where $p_i$ is to the left of $p_j$; the other case is symmetric. For any $t \in T$, let $x_{i}^{t}$ and $x_{j}^{t}$ denote the $x$-coordinates of the gridpoints at which $R_i$ and $R_j$ are located at time step $t$, respectively. We proceed by contradiction and assume that $q_i$ is not to the left of $q_j$. For $t=t_1$, since $p_i$ is to the left of $p_j$ and $\Delta(p_i, p_j) \geq d(k)$, we have $x_{j}^{t} - x_{i}^{t} > 0$.  Since $q_j$ is to the left of $q_i$, at time step $t=t_2$ we have $x_{j}^{t} - x_{j}^{t} \leq 0$.  Since in each time step the difference $x_{j}^{t} - x_{i}^{t}$ can decrease by at most 2, there must exist a time step $t$ at which $x_{j}^{t} - x_{j}^{t} \leq 2$, which implies that $\Delta_t(R_i, R_j) \leq p(k)+2$ contradicting the hypothesis. It follows that $q_i$ is to the left of $q_j$, and since $\Delta(q_i, q_j) =\Delta_{t_2}(R_i, R_j)$, we have $\Delta(q_i, q_j) \geq d(k)$.

To prove (ii), let $W'^{T}_{j}$ be a route for $R_j$ in $T$ between $p_j$ and $q_j$ such that $|W'^{T}_{j}| -\Delta(p_j, q_j) \leq \delta_j(k)$. Without loss of generality, assume that $p_i$ is to the left of $p_j$ as the proof is analogous for the other case. Then it follows from part (i) (proved above) that $q_i$ is to the left of $q_j$ and that $\Delta(q_i, q_j) \geq d(k)$. Since $\Delta(p_i, p_j) \geq d(k)$ and the grid's vertical dimension is at most $p(k)$, the horizontal distance, $d_H$, between 
$p_i$ and $p_j$ satisfies $d_H \geq d(k) -p(k)$. For any $t \in T$, let $x_{i}^{t}$ and $x_{j}^{t}$ denote the $x$-coordinates of the gridpoints at which $R_i$ and $R_j$ are located at time step $t$, respectively.

Suppose first that $q_i$ is to the left of $p_i$. For any $t \in T$, $x_{i}^{t}$ satisfies $x_{i}^{t} \leq x_{i} + \sigma_i(k) +p(k) -(t-t_1)$, where $x_i$ is the $x$-coordinate of $p_i$. Regardless of how the route $W'^{T}_{j}$ for $R_j$ is, for any $t \in T$, the $x$-coordinate of the point of $W'^{T}_{j}$ at time $t$, $x'^{t}_{j}$, satisfies: $x'^{t}_{j} \geq x_j -(t-t_1) \geq x_i + d_H -(t-t_1)$, where $x_j$ is the $x$-coordinate of $p_j$. It follows that $x'^{t}_{j} \geq x_{i}^{t} -\sigma_i(k) -p(k)+(t-t_1)+d_H-(t-t_1)$ and hence, $x'^{t}_{j} -x_{i}^{t} \geq d(k) -2p(k) -\sigma_i(k)$.  Observing that, for any $t \in T$, 
$\Delta(u_i, u_j)$ is lower bounded by $x'^{t}_{j} -x_{i}^{t}$, it follows that, for any $t \in T$ we have $\Delta(u_i, u_j) \geq d(k) -2p(k) - \sigma_i(k) \geq  d(k) -2p(k) - \sigma_i(k) - \delta_j(k)$. 

Suppose now that $q_i$ is to the right of $p_i$. Let $t_j \in T$ be the time step at which $R_j$ reaches $q_j$. For any $t \in [t_1, t_j]$, $x_{i}^{t}$ satisfies $x_{i}^{t} \leq x_{i} + (t-t_1)$, where $x_i$ is the $x$-coordinate of $p_i$. Regardless of how the route $W'^{T}_{j}$ for $R_j$ is, for any $t \in [t_1, t_j]$, the $x$-coordinate of the point of $W'^{T}_{j}$ at time $t$, $x'^{t}_{j}$, satisfies: $x'^{t}_{j} \geq x_j +(t-t_1) -p(k) -\delta_j(k)$, where $x_j$ is the $x$-coordinate of $p_j$. It follows that $x'^{t}_{j} -x_{i}^{t} \geq d(k) -2p(k) -\delta_j(k)$.  Observing that, for any $t \in [t_1, t_j]$, $\Delta(u_i, u_j)$ is lower bounded by $x'^{t}_{j} -x_{i}^{t}$, it follows that, for any $t \in [t_1, t_j]$ we have $\Delta(u_i, u_j) \geq d(k) -2p(k) - \delta_j(k) \geq  d(k) -2p(k) - \sigma_i(k) - \delta_j(k)$. Now for any $t \in [t_j+1, t_2]$, $R_j$ is located at $q_j$. Since $q_i$ is to the left of $q_j$,  $\texttt{slack}_T(R_i) \leq \sigma_i(k)$, and $\Delta(q_i, q_j) \geq d(k)$, it follows that for any $t \in [t_j+1, t_2]$, $\Delta(u_i, u_j) \geq d(k) - \sigma_i(k)$. Therefore, for any $t \in T$, we have $\Delta(u_i, u_j) \geq d(k) -2p(k) - \sigma_i(k) - \delta_j(k)$.   
\end{proof}

The following two lemmas prove the existence of a canonical solution. Lemma~\ref{lem:rerouting} handles the special case where exactly one of the two grid dimensions is upper bounded by a function of the parameter, whereas Lemma~\ref{lem:rerouting1} treats the case where both dimensions are unbounded. (The simple case, in which both dimensions are bounded, is \FPT{}, as explained in the proof of Theorem~\ref{thm:fpt}.)

\begin{longlemma}
\label{lem:rerouting}
Let $(G, \R, k, \ell)$ be a \YES-instance of \cmpm and let $\SSS$ be a minimal schedule for $(G, \R, k, \ell)$. Let $T=[t_1, t_2] \subseteq [0, \ell]$ be a $[\sigma, \gamma]$-good interval with respect to $d(k)$, where $k^{13^{k-1}} \leq \sigma(k) \leq 3^{13^{2k^2-2k}}\cdot k^{13^{2k^2-2k}}$, and $d(k) =\gamma(k) = \sigma^{13}(k)$. Assume that exactly one dimension of the grid $G$ is at least $p(k)= 2k \cdot (3k^k(\sigma(k)+1) + \sigma(k))+4k$. Then for every $T$-large-slack robot $R_i$, there is a route $W'_i$ that is equivalent to $W_i$ and such that $\nu_T(W'_i)$ is at most $3k^3$ and $W'_i$ is identical to $W_i$ in $[0, \ell]\setminus T$. 
 \end{longlemma}

\begin{proof}
For $R_i \in \R$, denote by $W_{i}^{T}$ the restriction of $W_i$ to $T$.  Let $\LA$ be the set of $T$-large slack robots and $\SM$ be that of $T$-small slack robots, as defined in Definition~\ref{def:intervalslack2}. By Lemma~\ref{lem:noslackinterval}, we may assume that for every $R_i \in \SM$, we have $\nu_T(W_{i}) \leq \tau(k)=3k^k(\sigma(k)+1) + \sigma(k)$. From the hypothesis of the current lemma, we may assume in what follows that exactly one of the (two) grid dimensions, say w.l.o.g.~the horizontal dimension, is larger than $p(k)$, whereas the vertical dimension is upper bounded by $p(k)$.  

Recall that, at any time step $t \in T$, for every $R_i \in \SM$ and for every $R_j \in \LA$, we have $\Delta_t(R_i, R_j) \geq d(k)$. It follows from part (i) of Lemma~\ref{lem:boundeddimensionordering} that, for every $R_i \in \SM$ and for every $R_j \in \LA$, the starting point $p_i$ of $R_i$ at time $t_1$ is to the left of the starting point $p_j$ of $R_j$ if and only if the destination point $q_i$ of $R_i$ at $t_2$ is to left of the destination point $q_j$ of $R_j$; moreover, $\Delta(q_j, q_j) \geq d(k)$. Now we define the route $W'_i$ for each $R_i \in \LA$. 

If there is only one horizontal line in the grid, then since  $(G, \R, k, \ell)$ is a \YES-instance, the order of the starting points of the robots on this horizontal line must match that of their destination points, and we define each $W'_i$ to be the sequence of points on the horizontal line segment joining its starting point to its destination point. Clearly, 
$\nu_T(W'_i) =0$ for each $R_i \in \LA$, and no two routes for any two robots conflict. Moreover, each robot in $\R$ is routed along a shortest path between its starting and destination points, and the statement of the lemma obviously follows in this case.

Suppose now that there are at least two horizontal lines in $G$. We define the route $W'_i$ for each $R_i \in \LA$ in three phases. In phase (1), we simultaneously shift the robots in $\LA$ horizontally so that each occupies a distinct vertical line; this can be achieved in at most $k$ time steps (by shifting so that robots occupying the same vertical line are shifted properly by different lengths) and such that each robot in $\LA$ makes at most 1 turn (during the shifting). We arbitrarily designate two horizontal lines, $D_{left}$ and $D_{right}$, where $D_{left}$ will be used to route every $R_i \in \LA$ whose destination point is to the left of its current point, and $D_{right}$ will be used to route every $R_i \in \LA$ whose destination point is to the right of its current point. In phase (2), we route every $R_i \in \LA$ to its designated horizontal line in $\{D_{left}, D_{right}\}$ along its distinct vertical line; this can be done in at most $p(k)$ time steps (i.e., the vertical dimension of $G$) and such that each robot in $\LA$ makes at most 1 turn. In phase (3), we route each $R_i \in \LA$ along its designated horizontal line (either $D_{left}$ or $D_{right}$) to the vertical line containing its destination point; note that all the robots that have the same designated line move in the same direction. Once a robot reaches the vertical line containing its destination point, it is routed along that vertical line to its destination on it. Conflicts among the robots in $\LA$ may arise in phase (3), and we discuss how to resolve them next.  

We treat the case where there are exactly two horizontal lines in $G$, $D_{left}$ and $D_{right}$; the case where there are more horizontal lines is easier since there is more space to move/swap the robots.
We clarify first how, in phase (3), a robot $R$ that has reached the vertical line $L_R$ containing its destination point, and that needs to move vertically to its destination point (i.e., its destination point is on the other horizontal line), is routed along to its destination point. If $R$'s destination point is not occupied by any robot, then in the next time step, only $R$ moves to its destination point while all other robots wait; note that there can be at most $k$ such time steps during which a robot moves to its destination point. Otherwise, $R$'s destination point $p'$ is occupied by a robot $R'$, and in such case a conflict arises; we call such conflicts \emph{vertical-conflicts}. When a vertical-conflict arises, we prioritize it over the other type of conflicts (discussed below) and resolve them one at a time (in arbitrary order), while freezing the desired motion of all the robots that are not involved in the vertical-conflict. To resolve a vertical-conflict, let $p$ be the gridpoint at which $R$  currently is, and we distinguish two cases. If $p$ is the destination point of $R'$, then we swap $R$ and $R'$ by simultaneously shifting horizontally by one position all the robots located on the side of $L_R$ that contains at least $k$ vertical lines (thus creating an empty space on that side of $L_R$), and then rotating $R$ and $R'$ using that empty space so that they exchange positions in 3 time steps, and then simultaneously shifting back by one position all the robots that were shifted. Clearly, this swap can be achieved in at most 5 time steps, in which each robot makes at most 5 turns, and resolves the vertical-conflict.
Suppose now that $p$ is not the destination point of $R'$, and hence, the destination point of $R'$ must be to the left of $L_R$. Consider the sequence $S$ of robots (if any) that appear consecutively to the left of $R'$ on $D_L$. If no robot in this sequence has its destination point above its current point, then there must exist at least one empty space to the left of this sequence on $D_L$; moreover, the destination point of each robot in this sequence is to its left. To resolve the vertical-conflict in this case, we simultaneously shift $R'$ and the robots in $S$ by one position to the left while moving $R$ to its destination $p'$. This incurs one time step and at most one turn per robot. Suppose now that one robot in this sequence $S$ has its destination point above it. Let $R''$ be the closest such robot to $R'$, and let $p''$ be the gridpoint above $R''$. At least one side of $L_R$ (either left or right or both) contains at least $k$ vertical lines; without loss of generality, assume it is the right side and the treatment is analogous if it was the left side. We horizontally and simultaneously shift all robots to the right of $L_R$ on both $D_L$ and $D_R$ by one position to the right thus creating one unoccupied gridpoint on each of $D_L$ and $D_R$ to the right of $L_R$, move $R$ in two steps using the unoccupied gridpoints to the right of $p'$, and simultaneously shift any robots between $p''$ and $p$ one position to the right, $R''$ to $p''$ and all robots in $S$ other than $R''$ one position to the left. Finally, move $R$ to $p'$.  This sequence of moves resolves the vertical-conflict in at most 3 time steps and incurs at most 3 turns per robot. It follows that the at most $k$ vertical conflicts can be resolved in at most $5k$ time steps and by making at most 5 turns per robot.
   
Suppose now that there are no vertical-conflicts and that a conflict arises in phase (3) when a robot is moving into a point on its designated line that is already occupied by another robot. We will call such conflicts \emph{horizontal-conflicts}. 
When multiple horizontal conflicts arise during a time step, we prioritize them as follows. A horizontal-conflict on $D_{right}$ has higher priority than one on $D_{left}$, and for conflicts on the same line, we prioritize them from right (highest) to left (lowest) for $D_{right}$ and left to right for $D_{left}$. When multiple horizontal conflicts arise during a time step $t$, we only consider one conflict at a time, in order of priority, and freeze the motion of all the robots that are not involved in the conflict and its resolution, by keeping their current position before time step~$t$. 

Without loss of generality, we discuss how to resolve a horizontal-conflict $C$ on $D_{right}$. Suppose that $R$, which is currently located at grid point $p$ on $D_{right}$, wants to move during time step $t$ into a gridpoint $p'$ on $D_{right}$ that is currently occupied by a robot $R'$. If $p'$ is not the destination of $R'$, then the point, $p''$, to the right of $R'$ on $D_{right}$ must be unoccupied (and conflict $C$ would not have arisen); otherwise, a conflict of higher priority, which is that resulting from moving $R'$ from $p'$ to $p''$ should have been considered before $C$. Therefore, we may now assume that $p'$ is the destination of $R'$. Let $q'$ be the gridpoint on the same vertical line as $p'$ and $q$ that on the same vertical line as $R$. Note that since the destination point of $R$ must be to the right of $p$, there are gridlines to $p'$'s right. We resolve conflict $C$ as follows. We first move $R'$ vertically to the point $q'$ on $D_{left}$; if $q'$ is occupied, we simultaneously shift by one position each robot on $D_{left}$ starting from $q'$ and residing on the larger vertical side of $q'$ in the direction of the larger side. We simultaneously move $R'$ to $q'$ and $R$ to $p'$. In the next time step, we attempt moving only $R$ to the point $p''$ to the right of $p'$ on $D_{right}$. If no conflicts arise, we simultaneously move $R$ to $p''$, reverse any shifts on $D_{right}$ that we performed to resolve $C$, and move $R'$ back to $p'$ on $D_{right}$. If, however, a conflict arises when moving $R$ to $p''$, we recursively resolve the new conflict, and when we are done, reverse any shifting performed on $D_{right}$ when resolving $C$, and move $R'$ back to $p'$. Note that, since during these recursive calls, we are resolving conflicts that proceed in the same direction on $D_{right}$ (and shifts that follow the same direction on $D_{left}$), no cyclical dependencies can arise during the recursion. In the worst case, to resolve such a conflict we incur no more than $2k$ time steps and no more than $2k$ turns, and since there can be at most ${k \choose 2}$-many horizontal conflicts (where robots could be in reverse order), resolving all horizontal-conflicts incurs no more than $k^3$ time steps, and no more than $k^3$ turns per robot.

It is easy to see that, over the three phases of the process, each $W'_i$ satisfies that $\nu(W'_i) \leq 1 + 1 + 5k+ k^3  \leq 3k^3$ (assuming $k \geq 2$). Moreover, over all three phases of the process, the length of $W'_i$ deviates from that of a shortest path between the endpoints of $W'_i$ by $\delta_i(k)$, which is at most the number of additional time steps spent in the three phases to shift the robots so that they occupy distinct vertical lines, route them to their designated horizontal lines, and resolve their conflicts. It follows that $\delta_i(k) \leq k+p(k) + 5k +k^3 \leq p(k)+3k^3$. Since $d(k)=\sigma^{13}(k)$ and $\sigma(k) \geq k^{13^{k-1}}$, it is straightforward to verify that $d(k) > 3p(k) + \sigma(k) + 3k^3$ for $k \geq 2$, and it follows from part (ii) of Lemma~\ref{lem:boundeddimensionordering} that no route $W'_i$ of a robot $R_i\in \LA$ conflicts with a route of a robot in $\SM$. Moreover, since $\gamma(k) > \delta_i(k)$, it follows that each $R_i \in \LA$ arrives to its destination (where it stays) no later than time step $t_2$, and the statement of the lemma is proved in this case.
\end{proof}

\begin{lemma}
\label{lem:rerouting1}
Let $(G, \R, k, \ell)$ be a \YES-instance of \cmpm  and let $\SSS$ be a minimal schedule for $(G, \R, k, \ell)$. Let $T=[t_1, t_2] \subseteq [0, \ell]$ be a $[\sigma, \gamma]$-good interval with respect to $d(k)$, where $k^{13^{k-1}} \leq \sigma(k) \leq 3^{13^{2k^2-2k}}\cdot k^{13^{2k^2-2k}}$, and $d(k) =\gamma(k) = \sigma^{13}(k)$. Assume that both dimensions of the grid $G$ are at least $p(k)= 2k \cdot (3k^k(\sigma(k)+1) + \sigma(k))+4k$. Then for every $T$-large-slack robot $R_i$, there is a route $W'_i$ that is equivalent to $W_i$ and such that $\nu_T(W'_i)$ is at most $4$ and $W'_i$ is identical to $W_i$ in $[0, \ell]\setminus T$. 
 \end{lemma}

\begin{proof}
For $R_i \in \R$, denote by $W_{i}^{T}$ the restriction of $W_i$ to $T$.  Let $\LA$ be the set of $T$-large slack robots and $\SM$ be that of $T$-small slack robots as defined in Definition~\ref{def:intervalslack2}. By Lemma~\ref{lem:noslackinterval}, we may assume that, for every $R_i \in \SM$, we have $\nu_T(W_{i}) \leq \tau(k)=3k^k(\sigma(k)+1) + \sigma(k)$. From the hypothesis of the current lemma, we may assume in what follows that both grid dimensions are at least $p(k)$.

Since for each $R_i \in \SM$, we have $\nu(W_{i}^{T}) \leq \tau(k)$, the number of grid lines used by each 
$R_i \in \SM$ is at most $2\tau(k)$, and hence, the number of lines in the set ${\cal U}$ of grid lines used by all robots in $\SM$ is at most $2 \cdot k \cdot \tau(k)$. Let ${\cal F}$ be the set of grid lines not used by any robot in 
$\SM$ and notice that, since $p(k)=2k \cdot \tau(k) +4k$, both the number of horizontal gridlines and the number of vertical gridlines in ${\cal F}$ are at least $p(k) -|{\cal U}| \geq 4k$.  We will define another valid schedule $\SSS'$ in which every $R_i \in \LA$ performs a route $W'^{T}_{i}$ in $T$, while the routes performed by the robots in $\SM$ are unchanged.
For each robot $R_i$, denote by $p_i$ its location (i.e., grid point) at time $t_1$ and $q_i$ its location at time $t_2$. To define the new routes for the robots in $\LA$, we apply the following process, which consists of three phases.

In phase (1), we will relocate the robots in $\LA$ so that each robot $R_i \in \LA$ occupies a point $p'_i$ such that: (i) both the horizontal and vertical lines containing $p'_i$ are in ${\cal F}$; and for any two distinct $R_i, R_j \in \LA$, points $p'_i, p'_j$ belong to different horizontal and vertical lines. Since $|{\cal U}| \leq  2k \cdot \tau(k)$ and both the number of horizontal and vertical gridlines in $|{\cal F}|$ are at least $4k$, clearly this relocation can be achieved in at most $|{\cal U}|+2k \leq 4k \cdot \tau(k)$ many steps, and by performing routes such that each makes at most $2$ turns (shifting vertically and then shifting horizontally). Let $t'_1$ be the time step at which relocating the robots in $\LA$ is complete, and each $R_i$ has been relocated to $p'_i$.

For each $R_i \in \LA$, choose a point $q'_i$ satisfying: (i) both the horizontal and vertical lines containing $q'_i$ are in ${\cal F}$; (ii) $\Delta(q_i, q'_i) \leq d(k)/8k$ and $q_i$ and $q'_i$ belong to different horizontal and vertical lines; and (iii) for any two distinct $R_i, R_j \in \LA$, $q'_i$ and $q'_j$ and $q'_i$ and $q_j$ belong to different horizontal and vertical lines. Clearly, the $q'_i$'s exist since both the number of horizontal and vertical lines in $|{\cal F}|$ are at least $4k$ and $d(k)/(8k) > |U|+8k$.

In phase (2), we will define a route $W_{i}^{T'}$, where $T'=[t'_1, t'_2]$ with $t'_2$ to be defined later, for each $R_i \in \LA$, whose endpoints are $p'_i$ and $q'_i$. To define $W_{i}^{T'}$, consider the projection $x'_i$ of $q'_i$ on the horizontal line containing $p'_i$; 
$W_{i}^{T'}$ is simply defined as the route that consists of the horizontal segment joining $p'_i$ to $x'_i$ followed by the vertical segment $x'_iq'_i$. Clearly, each $W_{i}^{T'}$ makes at most one turn. 
Moreover, by the choice of the $p'_i$'s and the $q'_i$'s, any two routes $W_{i}^{T'}$ and $W_{j}^{T'}$ intersect at most twice. Furthermore, since the two gridlines used by $W_{i}^{T'}$ are in ${\cal F}$, each $W_{i}^{T'}$ has at most $\eta(k)=2k \cdot \tau(k)$ many intersections with the routes of $T$-small-slack robots; note that a $T$-small-slack robot may intersect a route $W_{i}^{T'}$ multiple times at the same point, but those would constitute different intersections and are accounted for in the upper bound $\eta(k)$ on the number of intersections.  Moreover, since each $T$-small slack robot has slack at most $\sigma(k)$, the waiting time for any $T$-small slack robot at any gridpoint is at most $\sigma(k)$. Now we are ready to define the valid schedule for each $R_i \in \LA$ during $[t'_1, t'_2]$. We will impose an arbitrary ordering on the robots in $\LA$. We will define the valid schedule of the robots in $\LA$ in decreasing order. More specifically, a robot of lower order will always wait for a robot of higher order, and hence, the definition of the route of a robot $R_i \in \LA$ will not be affected by the routes of the robots in $\LA$ of lower-order than $R_i$.  For the robot $R_i \in \LA$ with the highest order, $R_i$ will follow its 
pre-defined route $W^{T'}_{i}$ from $p'_i$ to $q'_i$. Suppose that $R_i$ has traversed the following sequence of points $\langle u_1=p'_i, \ldots, u_r\rangle$, and suppose that the move of $R_i$ during the next time step $t$ will cause it to collide with a robot $R \in \SM$, then $R_i$ will wait at its current point $u_r$ for $\sigma(k)$ many time steps, and if during any of these time steps $R_i$ collides with any robot in $\SM$, then $R_i$ would instead wait twice that number of steps (i.e., $2\sigma(k)$) at point $u_{r-1}$ and so on and so forth. Since $p'_i$ does not belong to any line in ${\cal U}$ (i.e., and hence does not belong to any route of a robot in $\SM$) there is a wait time of at most $\sigma(k)\cdot \eta'(k)$ time steps at a point in $\langle u_1=p'_i, \ldots, u_r\rangle$ that would resolve the initial collision, where $\eta'(k)$ is the number of intersections (counting point multiplicities) between $\langle u_1=p'_i, \ldots, u_r\rangle$ and the routes of the robots in $\SM$. Since the number of possible collisions between $W^{T'}_{i}$ and the routes of $T$-small slack robots is at most $\eta(k)$, it follows that there is a valid schedule for $R_i$ to traverse $W^{T'}_{i}$ with wait-time at most $\eta(k)\cdot\sigma(k)$ that does not collide with any $T$-small-slack robots. Suppose, inductively, that for each $R_i \in \LA$ whose order is $s$, there is a valid schedule for $R_i$ that traverses $W_{i}^{T'}$ with wait time at most $2^s \cdot \eta(k)\cdot\sigma(k)$. Recalling that any two routes $W_{i}^{T'}$ and $W_{j}^{T'}$ intersect at most twice, it is easy to verify that the total wait time for the robot with order $s+1$ in $\LA$ is at most $2^{s+1}\cdot \eta(k)\cdot\sigma(k)$. It follows that we can define a valid schedule $W_{i}^{T'}$, for each $R_i \in \LA$, that makes two turns in $T'$ and that has wait time at most $2^{k}\cdot \eta(k)\cdot\sigma(k)$. Let $t'_2$ be the time step at which the last robot $R_i \in \LA$ arrives to $q'_i$.  

Since $\gamma(k) = \sigma^{13}(k)$, it is easy to see that  $\gamma(k)  > 2^{k}\cdot \eta(k)\cdot\sigma(k)$, and each robot $R_i \in \LA$ will arrive to $q'_i$ with ample slack, namely with a slack of 
at least $\sigma^{13}(k) - 2^{k}\cdot \eta(k)\cdot\sigma(k)$, which can be easily verified to be larger than $\sigma^{13}(k)/4=  d(k)/4$, for all $k \geq 2$.

In phase (3), each robot $R_i \in \LA$ will wait at $q'_i$ until time step $t_2 - d(k)/4$. We will re-route each $R_i \in \LA$ from $q'_i$ to its final destination $q_i$ so that it arrives at $q_i$ at time precisely $t_2$. First, for each $R_i$, define the smallest subgrid $G_i$ that contains both $q'_i$ and $q_i$ and whose dimension is at most $\Delta(q_i, q'_i)$, and observe that the dimension of $G_i$ is upper bounded by $d(k)/8k$. Let ${\cal G} = \{G_i \mid R_i \in \LA \}$. If two subgrids $G_i$ and $G_j$ in 
${\cal G}$ overlap, remove them from ${\cal G}$ and replace them with the smallest rectangular subgrid containing both $G_i$ and $G_j$; repeat this process until no resulting subgrids in  ${\cal G}$ overlap.  Observe that at the end of this process,  $|{\cal G}| \leq k$ and for each $G_i \in {\cal G}$, its dimensions are upper bounded by
$k \cdot d(k)/8k \leq d(k)/8$. Therefore, for each $G_i \in {\cal G}$, no $R \in \SM$ is located in $G_i$ during the interval $(t_2-d(k)/4, t_2]$ as otherwise, its Manhattan distance from $q_i$ would be less than $d(k)/2$, and hence, its Manhattan distance from $R_i$ at $t_2$ would be less than $d(k)$, contradicting the assumption that $T$ is a good interval. Therefore, it suffices route all the robots $R_i \in \LA$ from $q'_i$ and $q_i$ during the interval  $(t_2-d(k)/4, t_2]$. Now it is straightforward to see that the robots in each $G_i$ can be routed so that we have a valid schedule taking each robot $R_i$ from $q'_i$ to $q_i$ in $(t_2-d(k)/4, t_2]$ and such that each robot makes at most 2 turns. We outline how this can be done. We order the destination points $q'_i$ of the robots in $G_i$ in lexicographic order (say from left to right, and bottom to top, with the largest being the rightmost and then topmost), and order the robots in the same order as that of their destination points. We route the robots sequentially, starting with the highest ordered robot in lexicographic order by routing it horizontally until it reaches the vertical line on which its destination point is located, followed by routing it vertically to its destination point. Once a robot arrives to its destination, we start routing the next robot in lexicographic order. Routing each robot takes at most $\Delta(q_i, q'_i)$ time steps, and makes at most $1$ turn. Therefore, all robots can be routed to their destination points in at most  $k \cdot \max\{\Delta(q_i, q'_i) \mid R_i \in \LA\} \leq d(k)/8$ time steps, and hence in the interval $(t_2-d(k)/4, t_2]$. 

Note that the number of turns made by each robot over the three phases of the process is at most 4.
\end{proof}
\fi

We now establish the canonical-solution result that forms the culmination of this section.

\begin{theorem}
\label{thm:boundedturns}
Let $(G, \R, k, \ell)$ be an instance of \cmpm such that at least one dimension of the grid $G$ is lower bounded by $2k \cdot (3k^k(3^{13^{2k^2-2k}}\cdot k^{13^{2k^2-2k}}+1) + 3^{13^{2k^2-2k}}\cdot k^{13^{2k^2-2k}})+4k$. If $(G, \R, k , \ell)$ is a \YES-instance, then it has a valid schedule $\SSS$ in which each route makes at most $\rho(k)=3^{k^3+1} \cdot  (3k^k (3^{13^{2k^2-2k}}\cdot k^{13^{2k^2-2k}}+1) +3^{13^{2k^2-2k}}\cdot k^{13^{2k^2-2k}})$ turns. 
\end{theorem}

\begin{proof}
Suppose that $(G, \R, k , \ell)$ is a \YES-instance of \cmpm. We proceed by contradiction. Let $\SSS$ be a minimal schedule for $(G, \R, k , \ell)$ and assume that $S$ has a route $W_i$ for $R_i$ that makes more than $\rho(k)$ turns. By Lemma~\ref{lem:timeinstant}, there exists a $[\sigma,\gamma]$-good interval $T \subseteq [0, \ell]$ such that $\nu_T(W_i) > 3k^k(\sigma(k)+1)+ \sigma(k)$, where $\sigma$ and $\gamma$ are the function specified in Lemma~\ref{lem:timeinstant}. 
\iflong By Lemma~\ref{lem:rerouting} and Lemma~\ref{lem:rerouting1}, \fi
\ifshort By Lemma~\ref{lem:reroutingoverall}, \fi
there is an equivalent route $W'_i$ to $W_i$ that agrees with $W_i$ outside of $T$ and such that $\nu_T(W'_i) \leq 3k^3 < 3k^k(\sigma(k)+1)+ \sigma(k)$, which contradicts the minimality of $\SSS$.
\end{proof}
 
\subsection{Finding Canonical Solutions}
\label{sub:ilp}
Having established the existence of canonical solutions with a bounded number of turns, we can proceed to describe the proof of the \FPT{} result. In the proof, we 
identify a ``combinatorial snapshot'' of a solution whose size is upper-bounded by a function of the parameter $k$. We then branch over all possible combinatorial snapshots and, for each such snapshot, we reduce the problem of determining whether there exists a corresponding solution to an instance of Integer Linear Programming in which the number of variables is upper-bounded by a function of the parameter, which can be solved in \FPT-time by existing algorithms~\cite{FrankTardos87,Lenstra83,Kannan87}.

\ifshort
In particular, the aforementioned combinatorial snapshot will be a tuple $(G_\emph{snap},$ $\R_\emph{snap}, $ $\W_\emph{snap}, \iota)$ where $G_\emph{snap}$ is a bounded-size subgrid, $\R_\emph{snap}$ is a tuple of $k$ pairs of starting and ending vertices in $G_\emph{snap}$, $\W_\emph{snap}$ specifies a set of routes connecting the individual starting and ending vertices, and $\iota$ contains information about the order in which vertices are visited by the routes in $\W_\emph{snap}$. For each snapshot, we construct an ILP instance with variables that capture (1) the amount of ``expansion'' necessary to go from the snapshot to the full input grid, and (2) the amount of waiting a robot performs at certain ``critical'' junctions in the route. Constraints are then used to ensure that each robot arrives in time, that the routes correspond to the information in $\iota$ and do not lead to conflicts, and finally that the amount of expansion needed matches the size of the input grid.
\fi
 
 \begin{theorem} 
 \label{thm:fpt} \cmpm is \FPT{} parameterized by the number of robots.
 \end{theorem}

 \iflong
 \begin{proof}
 Let ${\cal I}=(G, \R, k, \ell)$ be an instance of \cmpm. If both dimensions of $G$ are upper bounded by $p(k)= 2k \cdot (3k^k(3^{13^{2k^2-2k}}\cdot k^{13^{2k^2-2k}}+1) + 3^{13^{2k^2-2k}}\cdot k^{13^{2k^2-2k}})+4k$, then it follows from the results in~\cite{demaine} that there is an approximate solution to the instance in which the length of the route of each robot is upper bounded by a linear function of the distance between the starting and destination point of the robot, and hence is $\bigoh(p(k))$. Therefore, we can assume that $\ell=\bigoh(p(k))$, and we can enumerate all possible routes for the robots in \FPT-time to decide if there is a valid schedule.

 We may now assume that at least one dimension of $G$ is lower bounded by $p(k)$. By Theorem~\ref{thm:boundedturns}, ${\cal I}$ is a \YES-instance if and only if it admits a solution $\SSS$ in which each route makes at most $\rho(k)=3^{k^3+1} \cdot (3k^k (3^{13^{2k^2-2k}}\cdot k^{13^{2k^2-2k}}+1) +3^{13^{2k^2-2k}}\cdot k^{13^{2k^2-2k}})$ turns. Therefore, it suffices to provide a fixed-parameter procedure to verify whether $\cal I$ admits such a solution. We will do so by applying a two-step procedure: first, we apply exhaustive branching to determine the combinatorial ``snapshot'' of a hypothetical solution $\SSS$, and then we use this to construct a set of instances of Integer Linear Programming (ILP) which will either allow us to construct a solution with the assumed combinatorial snapshot, or rule out the existence of any such solution. Crucially, the number of variables used in the ILP will be upper-bounded by a function of $k$, making it possible to solve it by the classical result of Lenstra~\cite{Lenstra83,Kannan87,FrankTardos87}.
 
Let us now formalize the notion of ``snapshot'' mentioned above: a \emph{snapshot} is a tuple $(G_\emph{snap},$ $\R_\emph{snap}, $ $\W_\emph{snap}, \iota)$ where 
\begin{itemize}
\item $G_\emph{snap}$ is a grid with at most $k^2\cdot (\rho(k)+2)$ many rows and columns (which we will sometimes refer to as coordinates),
\item $\R_\emph{snap}=\{(s'_i,t'_i)~|~s'_i,t'_i\in V(G_\emph{snap}) \wedge i\in[k]\}$ is a tuple of $k$ pairs of starting and ending vertices in $G_\emph{snap}$,
\item $\W_\emph{snap}$ is a tuple of $k$ routes, where for each $i\in [k]$ the $i$-th route starts at $s'_i$, ends at $t'_i$, does at most $\rho(k)$ turns and does not wait, and
\item $\iota$ is a mapping from $V(G_\emph{snap})$ to tuples from $\{[k]^j~|~j\in [k\cdot \rho(k)]\}$ with the following property: a vertex $v$ is visited by route $w_i\in \W_\emph{snap}$ precisely $p$ times if and only if the tuple $\iota(v)$ contains precisely $p$ many entries $i$. 
\end{itemize}

Let $\texttt{Snap}$ be the set of all possible snapshots. Observe that since the first three elements of the snapshot can be completely identified by specifying the grid size and the coordinates of all starting and ending vertices as well as of all the turns, we can bound the number of snapshots by $|\texttt{Snap}|\leq \rho(k)^{\bigoh(k\cdot \rho(k))}$. 

Next, we need to establish the connection between snapshots and solutions for $\mathcal{I}$. First, let us consider an arbitrary hypothetical solution $\SSS'$ where each route makes at most $\rho(k)$ turns.  We begin by marking each $x$-coordinate containing at least one turn, at least one starting vertex or at least one ending vertex as \emph{important}, and proceed analogously for $y$-coordinates; this results in at most $k\cdot (\rho(k)+2)$ coordinates being marked as important. 

A vertex is \emph{important} if it lies at the intersection of an important $x$-coordinate and an important $y$-coordinate. Observe that at each time step, each robot is either at an important vertex, or has precisely one unimportant coordinate; in the latter case, the robot has a clearly defined direction and can either wait or travel in that direction until it reaches an important vertex. 

We say that a vertex is a \emph{rest vertex} if it can be reached from an important vertex by a route of length at most $k$ without turns. The reason we define these is that we will show that one can restrict the waiting steps of the robots to only these two kinds of vertices. In particular, let us call a solution \emph{organized} if robots only wait at important and rest vertices.

\begin{ourclaim}
\label{claim:organized}
If $\mathcal{I}$ admits a solution $\SSS'$, then it also admits an organized solution $\SSS$.
\end{ourclaim}

\begin{proof}[Proof of the Claim]
Assume that $\SSS'$ is not organized, and choose an arbitrary robot $R$ that waits at some time step on some vertex $v$ that is neither important nor rest; let $i$ be the last time step during which $R$ waits on $v$. Recall that $v$ must have a clearly defined direction; without loss of generality, let us consider the direction to be ``\emph{right}'' (the arguments are completely symmetric for the three other directions). Hence, while at time step $i+1$ $R$ still remains on $v$, at time step $i+2$ $R$ moves to the neighbor of $v$ on the right.

Let $\R'$ be the set of robots with the following property: at time step $i+1$, each robot in $\R'$ can be reached by going right from $R$ by a path of vertices in $G$ all of which are occupied by robots. Note that in most cases $\R'$ will contain only $R$, but if, e.g., the $3$ vertices to the right of $R$ are occupied by other robots, $\R'$ will contain these three as well. Crucially, no vertex in $\R'$ can be located on an important vertex at time step $i+1$, and the vertex to the right of the rightmost robot in $\R'$ at time step $i+1$ cannot be important either ($\R'$ could, however, include some robots placed on rest vertices). Furthermore, since $R$ moves to the right at time step $i+1$ and no robots in $\R'$ lie on important vertices, all robots in $\R'$ move to the right as well. Similarly, since $R$ was waiting at time step $i$, all robots in $\R'$ had to wait at time step $i$ as well (otherwise they could not be immediately to the right of $R$ without turning).

Now, consider a new solution $\SSS^*$ that is obtained from $\SSS'$ by swapping the order of the wait and move commands for $\R'$, i.e., at time step $i$ all robots in $\R'$ move to the right, while at time step $i+1$ all robots in $\R'$ wait. Since no robot in $\R'$ lies at an important vertex, this local swap cannot cause any collisions in $\SSS^*$, and hence $\SSS^*$ is also a solution.

Notice that repeating the procedure described above always results in a new solution with the same makespan, and each application of the procedure postpones the times of some wait commands. Hence, after a finite number of iterative calls, we obtain a solution $\SSS_\emph{organized}$ where no robot waits on vertices that are neither important nor rest.
\end{proof}

Claim~\ref{claim:organized} will allow us to restrict our attention to organized hypothetical solutions of $\mathcal{I}$, meaning that robots can only ``move straight'' on non-rest vertices (they can neither wait nor make turns there). Now, let us return to the connection between snapshots and solutions. Let a horizontal (or vertical) coordinate be \emph{active} if it contains at least one vertex that is rest or important. Let a \emph{row contraction} at a horizontal coordinate $a$ be the operation of contracting every edge between $a$ and horizontal coordinate $a+1$; \emph{column contractions} are defined analogously for a vertical coordinate $a$. A row or column contraction at coordinate $a$ is \emph{admissible} if no vertex at coordinate $a+1$ contains a rest or important vertex. 

To obtain the snapshot of $\SSS$, we perform the following operations. First, we delete all vertices in rows whose horizontal coordinates are lower than the smallest horizontal coordinate containing at least one important or rest vertex, and store the number of rows deleted in this manner as $w_\downarrow(0)$. Next, for each horizontal coordinate $a$ such that $a$ is the $i$-th smallest horizontal coordinate, we exhaustively iterate admissible row contractions and we store the total number of such contractions carried out for $a$ as $w_\downarrow(i)$. We then proceed analogously for vertical coordinates to obtain the values $w_\rightarrow(0)$ and $w_\rightarrow(i)$. This procedure results in a new grid $G_\emph{snap}$ with $k$ terminal pairs stored in $\R_\emph{snap}$, a set $\W_\emph{snap}$ of routes connecting the terminal pairs obtained from $\SSS$. Finally, observe that $V(G_\emph{snap})\subseteq V(G)$, and hence, for each $v\in V(G_\emph{snap})$ there is a precise order in which the individual routes in $\SSS$ intersect $v$, and this is what is used to define $\iota$. Formally:
\begin{itemize}
\item $\iota(v)=\emptyset$ if $v$ is not used by any route in $\SSS$; otherwise
\item $\iota(v)[i]$ is the index of the $i$-th route that visits $v$.\footnote{We consider a visit to be a situation where a robot was not on $v$ at some time step $j-1$, and then enters $v$ at time step $j$. This means that waiting on $v$ does not amount to repeated visits, but a single robot may leave and then return to $v$ multiple times.}

 \end{itemize}

The tuple $(G_\emph{snap},\R_\emph{snap},\W_\emph{snap},\iota)$ is then the snapshot \emph{corresponding} to $\SSS$ witnessed by the functions $w_\downarrow$, $w_\rightarrow$. 
While this settles the construction of a snapshot from a solution (a step that will be required to argue correctness), what we actually need to do in the algorithm is the converse: given a snapshot $(G_\emph{snap},\R_\emph{snap},\W_\emph{snap},\iota)$, determine whether there is a solution corresponding to it. Here, it will be useful to notice that when constructing a snapshot from an organized solution $\SSS$, we lost two types of information about $\SSS$: the amount of contraction that took place (captured by the functions $w_\downarrow$, $w_\rightarrow$), and the amount of waiting that took place. However, since $\SSS$ is organized, no robot could have waited on any vertex outside of $V(G_\emph{snap})$. Hence, all the waiting times in $\SSS$ can be captured by a mapping which specifies how long robot $i$ waits at each individual visit of a vertex $v$. Since all visits of $v$ are captured by $\iota$, we can formalize this by defining, for each $v\in V(G_\emph{snap})$, the mapping $\emph{wait}_v: q\in [|\iota(v)|]\rightarrow \mathbb{N}$, with the following semantics: $\emph{wait}_v(q)=x$ if and only if the $q$-th robot visiting $v$, i.e., $\iota(v)[q]$, waits precisely $x$ time steps after arriving at $v$ for that visit.

The values of $w_\downarrow$, $w_\rightarrow$ and of the functions $\emph{wait}_v$ for all $v\in V(G_\emph{snap})$ cannot be directly inferred from a snapshot alone, and will serve as the variables for the constructed ILP. Note that the total number of variables capturing the values of $w_\downarrow$ and $w_\rightarrow$ is upper-bounded by $2\cdot (k^2\cdot (\rho(k)+2))$, while the total number of such variables for each of the at most $(k^2\cdot (\rho(k)+2)^2$ many $\emph{wait}_v$ functions is upper-bounded by $k\cdot \rho(k)$.

We now describe the constraints which will be used for the ILP. As a basic condition, we restrict all variables to be non-negative. To give our constraints, it will be useful to note that the time robot $i$ arrives at some vertex $v$ for the $j$-th time (over all routes)---hereinafter denoted $\emph{arrival}(i,v,j)$---can be completely captured using our snapshot and variables. 
To do so, one can follow the route $\R_\emph{snap}[i]$ from its starting point until it visits some vertex $v\in V(G_\emph{snap})$; if this is the $j$-th visit of $v$ by some route, then we define $\emph{arrival}(i,v,j)$ as a simple sum of the following variables:

\begin{enumerate}
\item for each $j'$-th visit of a vertex $v'$ (counted over all visits of $v'$ by all robots) by $i$, the variable $\emph{wait}_{v'}(j')$ (which semantically captures the number of time steps $i$ waited at $v$ at this visit), and
\item for each move along a horizontal edge between vertical coordinates $z$ and $z+1$ in $G_\emph{snap}$, the variable $w_\rightarrow(z)$ plus one (which semantically captures the number of time steps $i$ had to continuously travel in $G$ between these two rest and/or important vertices), and 
\item for each move along a vertical edge between horizontal coordinates $z$ and $z+1$ in $G_\emph{snap}$, the variable $w_\downarrow(z)$ plus one (as before).
\end{enumerate}

The first set of constraints are the \emph{timing} constraints, which are aimed to ensure that each robot arrives to its final destination in time. These simply state that for each robot $i$ with target vertex $t$ and where the route $i$ in $\R_\emph{snap}$ visits $t$ a total of $j$ times, $\emph{arrival}(i,v,j)\leq \ell$. The second set of constraints are \emph{size} constraints, which make sure that the $w_\downarrow$ and $w_\rightarrow$ functions reflect the initial size of $G$. Let $\emph{row}$ be the number of rows of $G_\emph{snap}$, and similarly for $\emph{col}$; then the size contraints simply require that (1) $\emph{row}+\sum_{i\in 0,\dots,\emph{row}}w_\downarrow(i)$ equals the number of rows in $G$, and similarly (2) $\emph{col}+\sum_{i\in 0,\dots,\emph{col}}w_\rightarrow(i)$ equals the number of columns in $G$.

The last set of constraints we need are the \emph{traffic constraints}, which ensure that the individual robots visit the vertices in $G_\emph{snap}$ in the order prescribed by $\iota$ and do not collide with each other. Essentially, these make sure that for each vertex $v$ that is consecutively visited by some robot $a$ and then $b$ (where $b$ could, in principle, be the same robot), robot $a$ leaves $v$ before $b$ arrives to $v$. These are fairly easy to formalize: for each vertex $v$ such that $\iota(v)(q)=i$ and $\iota(v)(q+1)=j$, we set $\emph{arrival}(i,v,q)+\emph{wait}_v(q) < \emph{arrival}(i,v,q+1)$.

We observe that these constraints are satisfiable for snapshots created from a solution $\SSS$; indeed, one can simply use the original values of $w_\downarrow$, $w_\rightarrow$ and $\emph{wait}_v$ obtained when constructing the snapshot, which are guaranteed to satisfy all of the above constraints by the virtue of $\SSS$ being a solution.

\begin{longobservation}
\label{obs:solsatisfies}
Let $\SSS$ be a solution with snapshot $(G_\emph{snap},\R_\emph{snap},\W_\emph{snap},\iota)$. Then there exist integer assignments to the variables $w_\downarrow$, $w_\rightarrow$ and $\emph{wait}_v$ satisfying all of the above constraints.
\end{longobservation}

Crucially, the converse is also true: if a ``well-formed'' snapshot gives rise to an ILP which is satisfiable, then $\mathcal{I}$ is a YES-instance. 

\begin{ourclaim}
\label{claim:snapshotsolution}
Let $(G_\emph{snap},\R_\emph{snap},\W_\emph{snap},\iota)$ be a snapshot and $\alpha$ be a variable assignment for $w_\downarrow$, $w_\rightarrow$ and $\emph{wait}_v$ satisfying all of the constraints. Then $\mathcal{I}$ admits a solution, and the solution can also be constructed from $\alpha$.
\end{ourclaim}

\begin{proof}[Proof of the Claim]
We construct the solution to $\mathcal{I}$ in the manner described in the paragraph defining the notion of correspondence between snapshots and solutions. In particular, we first use the values of $\alpha(w_\rightarrow(i))$ and $\alpha(w_\downarrow(i))$ (which satisfy the size constraints) to reverse the contractions and ``decompress'' $G_\emph{snap}$ into $G$. At this point, $\W_\emph{snap}$ traces a set of routes $\W$ in $G$ (one for each terminal pair), albeit without information about which robot waits where. We obtain this information from $\iota$ and $\alpha(\emph{wait}_{v'}(j'))$. Since $\alpha$ satisfies all the timing constraints, each robot arrives to their destination in time, and no two robots collide with each other due to the fact that $\alpha$ satisfies all traffic constraints.
\end{proof}

With Observation~\ref{obs:solsatisfies} and Claim~\ref{claim:snapshotsolution} in hand, we can move on to summarizing our algorithm and arguing its correctness. The algorithm begins by branching, in time at most $|\texttt{Snap}|$, over all possible snapshots of the input instance. This is carried out by brute force, where results which do not satisfy the conditions imposed on snapshots are discarded. For each snapshot, it constructs an ILP with the variables $w_\downarrow$, $w_\rightarrow$ and $\emph{wait}_v$ and the timing, size and traffic constraints described above. It then applies Lenstra's algorithm~\cite{Lenstra83,Kannan87,FrankTardos87} to determine whether the ILP instance has a solution. If the algorithm discovers a snapshot which admits a solution, it outputs ``YES'' and this is correct by Claim~\ref{claim:snapshotsolution}. On the other hand, if the algorithm determines that no snapshot gives rise to a solvable ILP instance, it can correctly output ``NO''; this is because by Observation~\ref{obs:solsatisfies}, the existence of a solution for $\mathcal{I}$ would necessitate the existence of a snapshot whose ILP is solvable, a contradiction.
 \end{proof}
 \fi

  \subsection{Minimizing the Total Traveled Length}
 \label{subsec:totaltime}
 In this subsection, we discuss how the strategy for establishing the fixed-parameter tractability of \cmpm parameterized by the number $k$ of robots can be used for \cmpl.
 \ifshort
The main difference between the two problems can be intuitively stated as follows: for \cmpm ``time matters'' but travel length could be lax, whereas for \cmpl ``travel length matters'' but time can be lax. 
The key tool we use to handle the complications arising in \cmpl when showing the existence of a canonical solution is a result that exhibits a schedule for any instance of \cmpl whose travel length is within a quadratic additive factor in $k$ from any length-optimal solution. 
Denote by $\texttt{dist}_{min}$ the sum of the Manhattan distances, over all the robots, between the starting point of the robot and its destination point. We have:

\begin{theorem}
\label{thm:boundeddistanceslack}
Let $\I=(G, \R, k, \lambda)$ be a \YES-instance of \cmpl. There is a schedule $\SSS$ for $\I$ satisfying that the total travel length of $\SSS$ is at most $\texttt{dist}_{min} + c(k)$, where $c(k)=\bigoh(k^2)$ is a computable function, and in which the number of turns made by each robot is $\bigoh(k)$.
\end{theorem}

The above theorem is then exploited for showing that if a robot makes a large number of turns, then we can find a time interval and a large rectangle of the grid such that, during that time interval, all the robots that are present in that rectangle 
behave ``nicely''. We formalize these notions in the following definitions:

 Let $M=[u_i, \ldots, u_j]$ be a monotone sequence of turns made by a robot $R \in \R$ during some time interval.   The \emph{rectangle} of $M$, denoted $\texttt{rectangle}(M)$, is the rectangle with diagonally-opposite vertices $u_i$ and $u_j$. We refer to Figure~\ref{fig:rectangle} for illustration.  
  
  \begin{figure}[htbp]

\centering
\begin{tikzpicture} [scale=0.8]
      \draw [fill=green] (3,1) rectangle (10,7);
    \draw[help lines,dashed] (0,0) grid (13,8);
     \node at (1, 1)   (a){};
     
          \node at (2.7, 1.2)   (b) {$u_i$};
          \filldraw (3,1) circle (2pt);
          
          \node at (2.5, 2.2)   (b) {$u_{i+1}$};
          \filldraw (3,2) circle (2pt);
          
          \node at (4.5, 2.2)   (c) {$u_{i+2}$};
          \filldraw (5,2) circle (2pt);

          \node at (4.5, 4.2)   (d) {$u_{i+3}$};
          \filldraw (5,4) circle (2pt);

          \node at (6.5, 4.2)   (e) {$u_{i+4}$};
          \filldraw (7,4) circle (2pt);
          
                    \node at (7, 5)   (f) {};

          \node at (9.5, 6.3)   (g) {$u_{j-1}$};
          \filldraw (10,7) circle (2pt);

          \node at (9.7, 7.2)   (h) {$u_{j}$};
          \filldraw (10,6) circle (2pt);

               \draw[line width =1mm] (1,1) -- (3,1);
                     \draw[line width =1mm] (3,1) -- (3,2);
                      \draw[line width =1mm] (3,2) -- (5,2);

                      \draw[line width =1mm] (5,2) -- (5,4);
                      \draw[line width =1mm] (5,4) -- (7,4);
                       \draw[dashed,line width =0.5mm] (7,4) -- (7,5);
                       \draw[dashed,line width =0.5mm] (9,6) -- (10,6);
                       \draw[line width =1mm] (10,6) -- (10,7);
                          \draw[line width =1mm] (10,7) -- (12,7);

 \end{tikzpicture}
\caption{$\texttt{rectangle}(M)$ (the green-shaded area) for a monotone sequence $M=[u_1, \ldots, u_j]$.}
\label{fig:rectangle}
\end{figure}

 \begin{definition}\rm
  \label{def:goodrectangle}
 Let $W$ be a subroute of a robot $R \in \R$ during some time interval $T$ such that the sequence  $M$ of turns in $W$ is monotone. Let $\sigma(k)$ be a function to be specified later. We say that $\texttt{rectangle}(M)$ is \emph{good} w.r.t.~$\sigma(k)$ and a time subinterval  $T' \subseteq T$ if:  (i) the set of robots present in $\texttt{rectangle}(M)$ is the same during each time step of $T'$; (ii) each robot $R_i$ present in $\texttt{rectangle}(M)$ during $T'$ satisfies $\texttt{slack}_{T'}(R_i) \geq \sigma(k)$; 
 (iii) each robot $R_i$ present in $\texttt{rectangle}(M)$ during $T'$ is traveling in the same direction as (the directions of the turns in) $M$; and (iv) each robot $R_i$ present in $\texttt{rectangle}(M)$ during $T'$ satisfies $\nu_{T'}(W_i) \geq \sigma(k)$.  
  \end{definition}

Next, we show that 
if a robot makes a large number of turns, then a good rectangle exists: 

  \begin{lemma}
    
  \label{lem:maintravelslack}
 Let $\I=(G, \R, k, \lambda)$ be a \YES-instance of \cmpl, let $\SSS$ be a valid schedule for $\I$, and assume that $\lambda < \texttt{dist}_{min}+c(k)$, where $c(k)=\bigoh(k^2)$ is the computable function in Theorem~\ref{thm:boundeddistanceslack}.  Let 
 $\sigma(k)=4k^2$ and $\tau(k)  =3k^k (\sigma(k)+1)+\sigma(k)$. Let $R$ be a robot such that the walk $W$ of $R$ during the time interval $T$ spanning $\SSS$ satisfies $\nu(W) =\Omega(\tau(k)^{2k+1})$. Then there exists a subwalk $W'$ for $R$ and a time interval $T' \subseteq T$ such that the sequence of turns $M'$ in $W'$ corresponding to $T'$ is monotone and $\texttt{rectangle}(M')$ is good w.r.t.~$\sigma(k)$ and~$T'$.
 \end{lemma}
 
 Using Theorem~\ref{thm:boundeddistanceslack} and Lemma~\ref{lem:maintravelslack}, we can prove that, given a good rectangle, we can reroute the robots that are present in that rectangle during a certain time interval so as to reduce the number of turns they make, which leads to the existence of a canonical solution:
 
  \begin{theorem}
\label{thm:boundedturnsdistance}
  If $\I=(G, \R, k, \lambda)$ is a \YES-instance of \cmpl, then $\I$ has a valid schedule $\SSS$ in which each route makes at most $\bigoh(\tau(k)^{2k+1})$ turns, where $\tau(k)  =3k^k (\sigma(k)+1)+\sigma(k)$, and $\sigma(k)=4k^2$.
\end{theorem}

  At this point, we can turn to the second step of our approach, notably checking whether an instance of \cmpl admits a solution in which the number of turns is upper-bounded by a function of the parameter. Luckily, here the proof of Theorem~\ref{thm:fpt} can be reused almost as-is, with only a single change in the ILP encoding at the end.

 \begin{theorem}
  \label{thm:fpt2} 
  \cmpl is \FPT{} parameterized by the number of robots.
 \end{theorem}
 \fi
\iflong

 We first recall some formal definitions related to the problem. 
Let $G$ be an $n \times m$ grid, $\R$ a set of $k$ robots in $G$, and $\SSS$ a valid schedule for $\R$. For a robot $R_i \in \R$ with starting point $s_i$ and destination point $t_i$, its \emph{traveled length} w.r.t.~$\SSS$ is the number $\lambda_i$ of time steps $t_i \geq 1$ satisfying that the gridpoint at which $R_i$ is at time step $t_{i-1}$ is different from that at which $R_i$ is at time step $t_i$; that is, $\lambda_i$ is the number of ``moves'' made by $R_i$ during $\SSS$.   
The \emph{total traveled length} of all the robots in $\R$ w.r.t.~schedule $\SSS$ is defined as $\lambda_{\SSS}= \sum_{i=1}^{k} \lambda_i$.
 
As our first step towards the fixed-parameter tractability of \cmpl parameterized by $k$, we once again show that every \YES-instance admits a canonical solution in which the total number of turns made by all the robots is upper-bounded by a function of $k$.
 To show this, we will not only use the notion of slack from Definition~\ref{def:slack}, but also a more refined distinction of it into: \emph{wait-slack} and \emph{travel-slack}. Informally speaking, the wait-slack of a robot during a time interval $T$ is the number of time steps in $T$ that the robot spends waiting (i.e., without moving), whereas the travel-slack of a robot is the number of time steps in $T$ during which the robot moves in a direction that is not a direction of its destination.

Whereas our main goal for \cmpl remains the same, which is establishing the existence of a canonical solution, the methods employed for \cmpm do not translate seamlessly to \cmpl. The main difference can be intuitively stated as follows: for \cmpm ``time matters'' but travel length could be lax, whereas for \cmpl ``travel length matters'' but time can be lax. More specifically, the existence of a good interval for an instance of \cmpm and the rerouting scheme designed to reroute the large-slack robots during that interval, cannot be adopted for \cmpl since the previously-designed rerouting scheme for \cmpm wastes travel length, which cannot be afforded in \cmpl. The key tool for working around this issue is a result that exhibits a solution for any instance of \cmpl whose travel length is within a quadratic additive factor in $k$ from any optimal solution; this allows us to work under the premise that the total travel-slack over all the robots cannot be large. This premise is then exploited for showing that if a robot makes a large number of turns, then we can find a time interval and a large region/rectangle of the grid such that, during that time interval, all the robots that are present in that rectangle behave ``nicely''; namely, all of them have large time-slack, all of them make a lot of turns, and all of them are traveling in the same direction. We then exploit these properties to show that, during that time interval, we can reroute all the robots present in the large rectangle so as to reduce their number of turns.

We now formalize the above notions.

\begin{longdefinition}\rm
 \label{def:distanceslack}
 Let $\I=(G, \R, k, \lambda)$ be an instance of \textsc{\cmpl}, $\SSS$ be a schedule for $\I$ with time interval $[0, t]$, and let $T=[t_1, t_2] \subseteq [0, t]$. For a robot $R_i \in \R$ with corresponding route $W_i=(u_0, \ldots, u_{t})$, define the \emph{wait-slack} of $R_i$ w.r.t.~$T$, denoted $\texttt{wait-slack}_T(R_i)$, as the number of indices $i \in [t_1+1, t_2]$ satisfying that $u_{i-1} \neq u_i$. Define the \emph{travel-slack} of $R_i$ w.r.t.~$T$, denoted $\texttt{travel-slack}_T(R_i)$, as 
 $\texttt{slack}_{T}(R_i) -$$\texttt{wait-slack}_T(R_i)$. We write $\texttt{wait-slack}(R_i)$ for $\texttt{wait-slack}_{[0, t]}(R_i)$ and $\texttt{travel-slack}(R_i)$ for $\texttt{travel-slack}_{[0, t]}(R_i)$. 
  \end{longdefinition}

Let $\I=(G, \R, k, \lambda)$ be a \YES-instance of \cmpl. Towards showing the structural result stating the existence of a canonical solution for $\I$, we will first show the existence of a schedule $\SSS$ for $\I$ in which the total travel-slack of all the robots is upper bounded by a function of $k$. More specifically, denote by $\texttt{dist}_{min}$ the sum of the Manhattan distances, over all the robots, between the starting point of the robot and its destination point; that is, $\texttt{dist}_{min} = \sum_{i=1}^{k} \Delta(s_i, t_i)$. We show that there exists a schedule $\SSS$ for $\I$ in which the total traveled length by all the robots is at most $\texttt{dist}_{min} +\bigoh(k^2)$.  

\begin{theorem}
\label{thm:boundeddistanceslack}
Let $\I=(G, \R, k, \lambda)$ be a \YES-instance of \cmpl. There is a schedule $\SSS$ for $\I$ satisfying that $\lambda_{\SSS} \leq \texttt{dist}_{min} + c(k)$, where $c(k)=\bigoh(k^2)$ is a computable function, and in which the number of turns made by each robot is $\bigoh(k)$.
\end{theorem}

\iflong
\begin{proof}
We will distinguish three cases: (i) both grid dimensions are more than $4k$; (ii) both grid dimensions are at most $4k$; and (iii) exactly one grid dimension is at most $4k$. 
In each of the three cases, we will describe a schedule $\SSS$ satisfying the statement of the lemma. Note that we only need to ensure that the total travel-slack of all the robots in the desired schedule $\SSS$ is at most $\bigoh(k^2)$ and that the number of turns made by each robot is $\bigoh(k)$. In particular, we are not concerned about the wait-slack of the robots, which could be arbitrarily large. 

In case (i), the schedule $\SSS$ is described as follows. We first shift the robots horizontally so that each robot is on a distinct vertical line $V_i$ such that $V_i$ does not contain the destination point of any robot. Since the horizontal dimension of $G$ is at least $4k$, it is easy to see that this shifting is feasible and can be done so that the total travel length of all the robots during this shifting is $\bigoh(k^2)$, and such that each robot incurs $\bigoh(k)$ many turns. Denote by $s'_i$ the current point (i.e., the new starting point) of $R_i$ on $V_i$ after this shifting. Next, we order the robots arbitrarily, and we will route them one at a time (i.e., sequentially). 

We pick a robot $R_i \in \R$ that has not been routed to its destination yet, and we choose a shortest-path route from its new starting point $s'_i$ to its destination $t_i$ as follows. First, we move $R_i$ along $V_i$ until it reaches the point $x_i$ on $V_i$ that is horizontally-collinear with $t_i$; note that this step incurs no collisions. Afterwards, we route $R_i$ along the horizontal segment $x_it_i$ until it reaches $t_i$; we will discuss shortly how to resolve any potential collisions during this step. Observe that, after the initial shifting, and since the robots that have been routed are already at their destinations, $t_i$ must be unoccupied when routing $R_i$. Clearly, the route from $s'_i$ to $t_i$ makes $\bigoh(1)$ many turns. We resolve any potential collisions that may happen when routing $R_i$ from $x_i$ to $t_i$ as follows. If $R_i$ is attempting to move horizontally from point $p_i$ to point $q_i$, where the latter is occupied by a robot, we first shift any robots on one side of the vertical line containing $q_i$ vertically (this side is chosen properly so that there is enough grid space to allow this shift), and then move $R_i$ to $q_i$. Since $q_i$ cannot be the final destination $t_i$ of $R_i$ (otherwise, it would be unoccupied), the next move of $R_i$ will be horizontal, thus leaving the vertical line containing $q_i$. After $R_i$ moves, we will reverse the shift (if any) performed to the robots on the vertical line containing $q_i$, thus restoring their positions on this vertical line. Clearly, each of these shifts incurs a total travel length at most twice the number of robots on the vertical line and incurs at most two turns per shifted robot.  It follows from the above that routing a robot $R_i$ incurs a travel length of $\lambda_i=\Delta(s_i, t_i) + \bigoh(k)$ for $R_i$, a travel length of $\bigoh(k)$ for any other robot, and a constant number of turns per robot. Therefore, routing all the robots incurs a total travel length of at most $\texttt{dist}_{min} + \bigoh(k^2)$ and $\bigoh(k)$ turns per robot. This shows the statement of the lemma for the first case.

For case (ii), the statement of the lemma follows from the result in~\cite{demaine}, which states that there is an approximate solution to the instance in which the length of the route of each robot is upper bounded by a linear function of the distance between the starting and destination point of the robot. Since in this case both grid dimensions are at most $4k$, for any robot $R_i$, we have $\Delta(s_i, t_i) \leq 8k$. It follows from the result in~\cite{demaine} that there is a solution $\SSS$ in which each robot travels a length that is $\bigoh(\Delta(s_i, t_i)) = \bigoh(k)$, and obviously makes $\bigoh(k)$ many turns. The statement of the lemma trivially follows in this case.

Finally, for case (iii), assume, without loss of generality, that the vertical dimension of the grid is at most $4k$. If there is only one horizontal line in the grid, then since $\I$ is a \YES-instance, it must be the case that the order of the starting points of the robots matches that of their destination points, and the statement trivially holds in this case. Therefore, we will assume that there are at least two horizontal lines in the grid. Denote by $H_1, \ldots, H_r$, where $r \leq 4k$, the horizontal lines in the grid, in a top-down order (i.e., $H_1$ is the topmost and $H_r$ is the bottommost). We first move all the robots to line $H_1$, starting from the robots on $H_2$ then $H_3$ until $H_r$. To move a robot $R_i$ on $H_j$ to $H_1$, we proceed along the vertical line containing $s_i$ until we reach $H_2$ ($s_i$ could be on $H_2$); note that no collision happens yet since all the robots on the lines $H_s$, where $s < j$ have been moved to $H_1$ at this point. Therefore, the only collision that could happen is when $R_i$ attempts to move to $H_1$. We resolve this potential collision by shifting all the robots on $H_1$ that fall on one side (the side is chosen properly so that there is enough horizontal grid space on that side to affect this shifting) of the vertical line containing $s_i$, thus creating an empty gridpoint at the intersection of $H_1$ and the vertical line containing $s_i$, which $R_i$ would now occupy. Observe that the length traveled by $R_i$, as well as by any other robot, during this step is $\bigoh(k)$. Moreover,  each robot incurs $\bigoh(1)$ many turns during the routing of $R_i$. After all the robots have moved to $H_1$, the next step is to move each robot whose destination point is on $H_1$ to its destination point. To do so, we will consider these robots in a left-to-right order of their destination points on $H_1$. To move a robot $R_i$ on $H_1$ to its destination point on $H_1$, we will use line $H_2$. If the destination point $t_i$ of $R_i$ on $H_1$ is not occupied, we will move $R_i$ down by 1 grid unit (along the same vertical line) to $H_2$, move it along $H_2$ to the gridpoint on $H_2$ just below $t_i$, and then move it up to $t_i$. If $t_i$ is occupied by some robot $R_j$, then we will move $R_j$ to an unoccupied gridpoint $u_i$ on $H_1$ that is within $k$ gridpoints from $t_i$ and that is not a destination point for any robot (which must exist since the horizontal dimension is more than $4k$), by first moving it down to $H_2$, routing it along $H_2$ to the point vertically below $u_i$, and then moving it up to $u_i$. We then route $R_i$ to $t_i$ as in the previous case where $t_i$ was unoccupied. Clearly, in this step each robot whose destination point is on $H_1$ travels a length of $\Delta(s_i, t_i) + \bigoh(k)$, and each other robot travels a length of $\bigoh(k)$. Moreover, each robot incurs $\bigoh(1)$ many turns. After this step, the robots whose destinations are on $H_1$ have been routed to their destinations. Finally, we will route the remaining robots on $H_1$, starting from those whose destinations are on $H_r$, going up to those whose destinations are on $H_2$, and for those robots whose destinations fall on $H_j$, where $j =r, \ldots, 2$, we route them in the right-left order in which their destination points appear on $H_j$. Each route will be along a vertical line until the robot arrives to the horizontal line containing its destination point, where it will proceed along this horizontal line until it arrives to its destination. Clearly, no collision could happen in this case (due to the imposed ordering), and each robot $R_i$ travels a length of $\Delta(s_i, t_i) + \bigoh(k)$ and makes $\bigoh(1)$ many turns. It follows from the above that, during the whole process, each robot $R_i$ travels a length of $\Delta(s_i, t_i) + \bigoh(k)$ and makes $\bigoh(k)$ many turns. Therefore, the total travel length of all the robots is $\texttt{dist}_{min} + \bigoh(k^2)$ and the number of turns made by each robot is $\bigoh(k)$.
\end{proof}
\fi

  For the rest of this subsection, we will denote by $c(k) = \bigoh(k^2)$ the computable function in the statement of Theorem~\ref{thm:boundeddistanceslack}.
\iflong
  \begin{longlemma}
  \label{lem:monotoneturns}
  Let $R \in \R$ and let $W$ be the walk of $R$ during a time interval $T$. Let $q(k)$ be a function. If $\texttt{travel-slack}_T(R) \leq q(k)$, then there is a time interval $T' \subseteq T$ such that the subwalk $W'$ of $W$ corresponding to $T'$ satisfies $\nu(W') = \Omega (\nu(W)/q(k))$ and the sequence of turns made by $R$ in $W'$ is monotone. 
  \end{longlemma}
  
  \begin{proof}
  Consider the sequence of turns in $W$. Assume, without loss of generality, that the destination of $R$ is Up-Right (including only Up or only Right) w.r.t.~its starting position. Any turn in $W$ such that one of its directions is not Up nor Right incurs $R$ a travel slack of at least 1. Hence, the number of such turns in $W$ is at most $q(k)$, and those turns subdivide $W$ into at most $q(k)+1$ subwalks such that the sequence of turns made in each of these subwalks is monotone. It follows that one of these subwalks, $W'$, corresponding to a time interval $T'$, satisfies $\nu(W') = \Omega (\nu(W)/q(k))$ and that the sequence of all the turns made by $R$ in $W'$ is monotone.
  \end{proof}
  \fi

  Let $M=[u_i, \ldots, u_j]$ be a monotone sequence of turns made by a robot $R \in \R$ during some time interval.   The \emph{rectangle} of $M$, denoted $\texttt{rectangle}(M)$, is the rectangle with diagonally-opposite vertices $u_i$ and $u_j$. We refer to Figure~\ref{fig:rectangle} for illustration.  
  
  \begin{figure}[htbp]

\centering
\begin{tikzpicture} [scale=0.8]
      \draw [fill=green] (3,1) rectangle (10,7);
    \draw[help lines,dashed] (0,0) grid (13,8);
     \node at (1, 1)   (a){};
     
          \node at (2.7, 1.2)   (b) {$u_i$};
          \filldraw (3,1) circle (2pt);
          
          \node at (2.5, 2.2)   (b) {$u_{i+1}$};
          \filldraw (3,2) circle (2pt);
          
          \node at (4.5, 2.2)   (c) {$u_{i+2}$};
          \filldraw (5,2) circle (2pt);

          \node at (4.5, 4.2)   (d) {$u_{i+3}$};
          \filldraw (5,4) circle (2pt);

          \node at (6.5, 4.2)   (e) {$u_{i+4}$};
          \filldraw (7,4) circle (2pt);
          
                    \node at (7, 5)   (f) {};

          \node at (9.5, 6.3)   (g) {$u_{j-1}$};
          \filldraw (10,7) circle (2pt);

          \node at (9.7, 7.2)   (h) {$u_{j}$};
          \filldraw (10,6) circle (2pt);

               \draw[line width =1mm] (1,1) -- (3,1);
                     \draw[line width =1mm] (3,1) -- (3,2);
                      \draw[line width =1mm] (3,2) -- (5,2);

                      \draw[line width =1mm] (5,2) -- (5,4);
                      \draw[line width =1mm] (5,4) -- (7,4);
                       \draw[dashed,line width =0.5mm] (7,4) -- (7,5);
                       \draw[dashed,line width =0.5mm] (9,6) -- (10,6);
                       \draw[line width =1mm] (10,6) -- (10,7);
                          \draw[line width =1mm] (10,7) -- (12,7);

 \end{tikzpicture}
\caption{Illustration of $\texttt{rectangle}(M)$ (green-shaded area) for a monotone sequence $M=[u_1, \ldots, u_j]$.}
\label{fig:rectangle}
\end{figure}

\ifshort
As it turns out, every minimal schedule with many turns must not only contain many rectangles, but also one with special properties that allow us to safely reroute the robots in a way which reduces the number of turns. 

that is ``good'', and which our proof will exploit, in the following sense.

As it turns out, every minimal schedule with many turns must not only contain many rectangles, but also one that is ``good'', and which our proof will exploit, in the following sense. 
\fi
  
  \iflong

    \begin{longlemma}
  \label{lem:staircasesmallslack}
 Let $\I=(G, \R, k, \lambda)$ be a \YES-instance of \cmpl and let $\SSS$ be a minimal schedule for $\I$.   Let $W$ be a walk for a robot $R \in \R$ during a time interval $T$ such that the sequence $M$ of turns in $W$ is monotone. Let $R_i \in \R$ be a robot such that $\texttt{slack}_T(R_i) \leq \sigma(k)$, for some function $\sigma(k)$. There is a subwalk $W'$ of $W$ during an interval $T' \subseteq T$ such that $\nu(W') = \Omega(\nu(W) /\tau(k))$, where $\tau(k) =3k^k (\sigma(k)+1)+\sigma(k)$, and $W_i$ does not intersect $\texttt{rectangle}(M')$, where $M'$ is the monotone subsequence of $M$ corresponding to the turns in $W'$.
  \end{longlemma}

 \begin{proof}
 Since slack$_T(R_i) \leq \sigma(k)$, by Lemma~\ref{lem:boundedslack}, we can assume that $\nu_T(W_i) \leq \tau(k)$. 
 Let $M=[u_1, \ldots, u_r]$ be the sequence of turns in $W$. We extract from $M$ ($\tau(k)+2$)-many monotone contiguous subsequences $S_1, \ldots, S_{\tau(k)+2}$, each of length at least $s= \Omega(\nu(W)/ \tau(k))$, as follows: $S_1=[u_1, \ldots, 
 u_s], S_2 = [u_{s+2}, \ldots, u_{2s+1}], \ldots, S_{i}=[u_{(i-1)s+i}, \ldots, u_{is+i-1}], \ldots, S_{\tau(k)+2}=[u_{(\tau(k)+1)s+\tau(k)+2}, \ldots, u_r]$. 
  
 Observe that, for any two different $S_i$ and $S_j$, where $i, j \in [\tau(k)+2]$, the projections of the horizontal (resp.~vertical) sides of their rectangles on the $x$-axis (resp.~$y$-axis) are pairwise non-overlapping. It follows that if $W_i$ intersects all the $S_i$'s then it would make at least $\tau(k)+1$ many turns, contradicting our assumption that  $\nu_T(W_i) \leq \tau(k)$. It follows that there exists a subwalk $W'$ of $W$ during a subinterval $T'  \subseteq T$, corresponding to a monotone (contiguous) subsequence $M'=S_i$ for some $i \in [\tau(k)+2]$, such that $\nu(W') = \Omega(\nu(W)/ \tau(k))$ and $W_i$ does not intersect $\texttt{rectangle}(M')$.   
 \end{proof}
  
  For a robot in $R \in \R$, we define its \emph{travel direction(s)} to be the direction(s) of its destination point w.r.t.~its starting point. 
  The proofs of the following lemmas follow similar arguments to those in the proof of Lemma~\ref{lem:staircasesmallslack} and are omitted:
  
   \begin{longlemma}
  \label{lem:staircasesmalltravelslack}
  Let $W$ be a walk for a robot $R \in \R$ during a time interval $T$ such that the sequence $M$ of turns in $W$ is monotone. Let $R_i \in \R$ such that $R_i$ is not traveling in the same direction as $M/W$ (i.e., at least one of the directions in which $R_i$ is traveling is not a direction of $M$) and such that $\texttt{travel-slack}_T(R_i) \leq q(k)$, for some function $q(k)$. Then there is a subwalk $W'$ of $W$ during an interval $T' \subseteq T$ such that $\nu(W') = \Omega(\nu(W)/q(k))$ and such that $W_i$ does not intersect $\texttt{rectangle}(M')$, where $M'$ is the monotone subsequence of $M$ corresponding to the turns in $W'$.
  \end{longlemma}
  
  \begin{longlemma}
  \label{lem:staircasesmallturns}
  Let $W$ be a walk for a robot $R \in \R$ during a time interval $T$ such that the sequence $M$ of turns in $W$ is monotone. Let $R_i \in \R$ such that $W_i$ makes at most $q(k)$ many turns in $\texttt{rectangle}(M)$. Then there is a subwalk $W'$ of $W$ during an interval $T' \subseteq T$ such that $\nu(W') =\Omega( \nu(W)/q(k))$ and such that $W_i$ does not intersect $\texttt{rectangle}(M')$, where $M'$ is the monotone subsequence of $M$ corresponding to the turns in $W'$.
  \end{longlemma}
  
   \begin{longlemma}
  \label{lem:entryexit}
 Let $\I=(G, \R, k, \lambda)$ be a \YES-instance of \cmpl and assume that $\lambda < \texttt{dist}_{min}+c(k)$, where $c(k)$ is the computable function in Theorem~\ref{thm:boundeddistanceslack}.  Let $W$ be a walk of a robot $R \in \R$ during a time interval $T$ such that the sequence $M$ of turns in $W$ is monotone. 
   Then there is a subwalk $W'$ of $W$ during a subinterval $T' \subseteq T$ corresponding to a subsequence $M'$ of $M$ such that $\nu(W') =\Omega(\nu(W) /(k \cdot c(k))$ and such that the set of robots $\R'$ contained in $\texttt{rectangle}(M)$ does not change during $T'$ (i.e., the set is the same at every time step in $T'$). 
  \end{longlemma}
  
  \begin{proof}
  Sine travel-\texttt{slack}($R_i$) $\leq  c(k)$ for each $R_i \in \R$, the number of times $R_i$ re-enters $\texttt{rectangle}(M)$ is $\bigoh(c(k))$, and hence, the total number of entries/exits to $\texttt{rectangle}(M)$ for all the robots during $T$ is $\bigoh(k \cdot c(k))$; this is true since each re-entry of a robot to $\texttt{rectangle}(M)$ costs the robot a travel-slack of at least 2. Therefore, we can partition $T$ into $\bigoh(k \cdot c(k))$ many contiguous intervals during each of which no robot enters/exits $\texttt{rectangle}(M)$. This partitioning in turn divides $T$ (and $M$) into $\bigoh(k \cdot c(k))$ many contiguous subintervals such that the set of robots in $\texttt{rectangle}(M)$ is the same during each of these intervals. By an averaging argument, there exists a time subinterval $T' \subseteq T$ corresponding to a subsequences $W'$, such that $\nu(W') = \nu_{T'}(W) = \Omega(\nu(W)/(k \cdot c(k)))$ and the statement follows. 
  \end{proof}
     \fi
  
  \begin{definition}\rm
  \label{def:goodrectangle}
 Let $W$ be the walk of a robot $R \in \R$ during some time interval $T$ such that the sequence  $M$ of turns in $W$ is monotone. Let $\sigma(k)$ be a function to be specified later. We say that $\texttt{rectangle}(M)$ is \emph{good} w.r.t.~$\sigma(k)$ and a time subinterval  $T' \subseteq T$ if:  (i) the set of robots present in $\texttt{rectangle}(M)$ is the same during each time step of $T'$; (ii) each robot $R_i$ present in $\texttt{rectangle}(M)$ during $T'$ satisfies $\texttt{slack}_{T'}(R_i) \geq \sigma(k)$; 
 (iii) each robot $R_i$ present in $\texttt{rectangle}(M)$ during $T'$ is traveling in the same direction as (the directions of the turns in) $M$; and (iv) each robot $R_i$ present in $\texttt{rectangle}(M)$ during $T'$ satisfies $\nu_{T'}(W_i) \geq \sigma(k)$.  
  \end{definition}

    \begin{lemma}
  \label{lem:maintravelslack}
 Let $\I=(G, \R, k, \lambda)$ be a \YES-instance of \cmpl, let $\SSS$ be a valid schedule for $\I$, and assume that $\lambda < \texttt{dist}_{min}+c(k)$, where $c(k)=\bigoh(k^2)$ is the computable function in Theorem~\ref{thm:boundeddistanceslack}.  Let 
 $\sigma(k)=4k^2$ and $\tau(k)  =3k^k (\sigma(k)+1)+\sigma(k)$. Let $R$ be a robot such that the walk $W$ of $R$ during the time interval $T$ spanning $\SSS$ satisfies $\nu(W) =\Omega(\tau(k)^{2k+1})$. Then there exists a subwalk $W'$ for $R$ and a time interval $T' \subseteq T$ such that the sequence of turns $M'$ in $W'$ corresponding to $T'$ is monotone and $\texttt{rectangle}(M')$ is good w.r.t.~$\sigma(k)$ and~$T'$.
 \end{lemma}
 
 \iflong
 \begin{proof}
 Since $\lambda < \texttt{dist}_{min}+c(k)$, we have travel-\texttt{slack}($R$) is at most $c(k)$. It follows from Lemma~\ref{lem:monotoneturns} that there is a time interval $T_1 \subseteq T$ such that the subwalk $W_1$ of $W$ corresponding to $T_1$ satisfies $\nu(W_1) =  \Omega(\nu(W) /c(k)) = \Omega( \tau(k)^{2k+1})$ and the sequence $M_1$ of turns that $R$ makes in $W_1$ is monotone. Assume, without loss of generality, that $R$ travels in an Up-Right direction during $T_1$.  We distinguish the following cases. \\
 
Case 1.  If a robot $R_i$ satisfying slack$_{T_1}(R_i) \leq \sigma(k)$ intersects $\texttt{rectangle}(M_1)$, then by Lemma~\ref{lem:staircasesmallslack}, there is a time interval $T_2$ such that the subwalk $W_2$ of $W_1$ corresponding to $T_2$ satisfies  
$\nu(W_2) =\Omega( \nu(W_1)/\tau(k))$ and such that $R_i$ does not intersect the monotone subsequence $M_2$ of turns corresponding to $W_2$. \\

Case 2. If a robot $R_i \in R$ is not traveling in the same direction as $M$ in $T_1$, then since  travel-\texttt{slack}($R_i$) is at most $c(k)$, by Lemma~\ref{lem:staircasesmalltravelslack}, there is a time interval $T_2$ such that the subwalk $W_2$ of $W_1$ corresponding to $T_2$ satisfies $\nu(W_2) =\Omega(\nu(W_1)/c(k))= \Omega(\nu(W_1)/\tau(k))$ and such that $R_i$ does not intersect the rectangle of the monotone subsequence  $M_2$ of turns corresponding to $W_2$. \\

Case 3. If for a robot $R_i \in R$, $W_i$ makes fewer than $\sigma(k)$ many turns in $\texttt{rectangle}(M_1)$ during $T_1$, then by Lemma~\ref{lem:staircasesmallturns}, there is a time interval $T_2$ such that the subwalk $W_2$ of $W_1$ corresponding to $T_2$ satisfies $\nu(W_2) =\Omega(\nu(W_1)/\sigma(k))= \Omega(\nu(W_1)/\tau(k))$ and such that $R_i$ does not intersect the rectangle of the monotone subsequence $M_2$ of turns corresponding to $W_2$. \\

Case 4. If a robot $R_i \in \R$ either enters/leaves $\texttt{rectangle}(M_1)$ during $T_1$, then by Lemma~\ref{lem:entryexit}, there is a time interval $T_2$ such that the subwalk $W_2$ of $W_1$ corresponding to $T_2$ satisfies
$\nu(W_2) = \Omega( \nu(W_1)/(k \cdot c(k)) =\Omega(\nu(W_1)/\tau(k))$ and such that the set of robots contained in $\texttt{rectangle}(M_1)$ does not change during $T_2$. \\

 If none of Cases 1-4 applies, then $\texttt{rectangle}(M_1)$ is good w.r.t.~$\sigma(k)$ and $T_1$, and the statement of the lemma follows. 
 Otherwise, we can extract from $M_1$ a sequence $M_2$ (according to the above cases 
 corresponding to a time interval $T_2$), which we test against the above cases. Note that each application of Cases 1-3 eliminates at least one of the robots from further consideration (since the subsequences defined are nested), and hence, Cases 1-3 cannot be applied more that $k-1$ times in total (since there are at most $k-1$ robots other than $R$). Moreover, if Case 4 applies, then either the process stops afterwards, or the application of Case 4 is directly followed by an application of one of the Cases 1-3, and hence, Case 4 applies at most $k-1$ times. It follows that this process must end 
after at most $2k-1$ iterations with a monotone sequence of turns $M'$ and an interval $T'$ such that $\texttt{rectangle}(M')$ is good w.r.t.~$\sigma(k)$ and $T'$. (Note that, since the process was applied at most $2k-1$ times, we have $\nu_{T'}(W') = \Omega(\tau^2(k))$, and by minimality of $\SSS$, it follows from Lemma~\ref{lem:boundedslack} that $\texttt{slack}_{T'}(R) \geq \sigma(k)$.)
 \end{proof}
 \fi

  Using Theorem~\ref{thm:boundeddistanceslack} and Lemma~\ref{lem:maintravelslack}, we can prove that, given a good rectangle, we can reroute the robots that are present in that rectangle during a certain time interval so as to reduce the number of turns they make (during that time interval). The following theorem puts it all together:

  \begin{theorem}
\label{thm:boundedturnsdistance}
  If $\I=(G, \R, k, \lambda)$ is a \YES-instance of \cmpl, then $\I$ has a valid schedule $\SSS$ in which each route makes at most $\bigoh(\tau(k)^{2k+1})$ turns, where $\tau(k)  =3k^k (\sigma(k)+1)+\sigma(k)$, and $\sigma(k)=4k^2$.
\end{theorem}

\iflong
\begin{proof} 
Assume that $\I$ is YES-instance of \cmpl and proceed by a contraction. Let $\texttt{dist}_{min} = \sum_{i=1}^{k} \Delta(s_i, t_i)$. By Theorem~\ref{thm:boundeddistanceslack}, we may assume that $\lambda \leq \texttt{dist}_{min} + c(k)$, where $c(k)=\bigoh(k^2)$ is the computable function in Theorem~\ref{thm:boundeddistanceslack}. It follows that, for each $R_i \in \R$ and for each valid schedule $\SSS$, we have $\texttt{travel-slack}(R_i) \leq c(k)$. Let $\SSS$ be a minimal solution for $\I$, let $T$ be the time interval spanning 
$\SSS$, and assume that there exists a robot $R$ that makes at least  $\bigoh(\tau(k)^{2k+1})$ turns in $T$.  Let $W$ be the route of $R$ during $T$.  Assume, without loss of generality, that the destination $t$ of $R$ is to the upper-right of its starting point $s$, and hence, except for the time during which $R$ deviates from its shortest route (i.e., incurs some travel slack), $R$ travels only Up or Right. 

By Lemma~\ref{lem:maintravelslack}, there exists a subwalk $W'$ for $R$ whose sequence of turns $M'$ is monotone and a time interval $T'=[t'_1, t'_2] \subseteq T$ such that $\texttt{rectangle}(M')$ is good w.r.t.~$\sigma(k)$ and $T'$. 
From properties (iii) and (iv) of a good interval and the monotonicity of $M'$, it follows that both dimensions of $\texttt{rectangle}(M')$ are at least $\sigma(k)$.

Let $\R'$ be the set of robots present in $\texttt{rectangle}(M')$ during $T'$. Since each robot in $\R'$ has time-slack at least $\sigma(k) \geq 4k^2$ during $T'$, each dimension of  $\texttt{rectangle}(M')$ is at least $\sigma(k)$, and all robots in $\R'$ are traveling in the same direction during $T'$,  it is not difficult to see that the robots in $\R'$ can be routed from their initial positions in $\texttt{rectangle}(M')$ at time instance $t'_1$, arriving to their destination positions in $\texttt{rectangle}(M')$ at time instance $t'_2$, without incurring any travel-slack and such that each robot makes a constant number of turns, thus contradicting the minimality of $\SSS$. To do so, we can start by repositioning the robots in $\R'$ so that each occupies a different horizontal line in $\texttt{rectangle}(M')$ while traveling along a shortest path to its destination at $t'_2$; it is not difficult to see that this can be achieved by properly shifting the robots to different horizontal lines in $\texttt{rectangle}(M')$ while moving them towards their destinations, and by incurring at most $k^2$ time-slack each. Once this has been achieved, we route each robot towards its destination along its distinct horizontal line until it arrives to a point in $\texttt{rectangle}(M')$ that is within a horizontal distance of $q(k)=k^2$ from its destination, at which it waits until time instance $t'_2-q(k)$; this is possible due to property (iv) of a good interval, which ensures/implies that each robot travels in $\texttt{rectangle}(M')$ a distance of at least $\sigma(k) =4k^2$ before it arrives to its destination. At instance $t'_2-q(k)$, we reposition the robots so that they are ordered/sorted in the same order as their destination points; again, this can be done while moving the robots towards their destination (and hence incurring no travel-slack) using at most $k^2$ time-slack per robot and a constant number of turns per robot. Finally, we route these ordered robots to their destination points at $t'_2$ where they stay until  time step $t'_2$. Since each robot, including $R$, incurs a constant number of turns in this rerouting, which contradicts the minimality of $\SSS$.
\end{proof}
\fi
  
  At this point, we can turn to the second step of our approach, notably checking whether an instance of \cmpl admits a solution in which the number of turns is upper-bounded by a function of the parameter. Luckily, here the proof of Theorem~\ref{thm:fpt} can be reused almost as-is, with only a single change in the ILP encoding at the end. 
  \iflong
  Notably, instead of the ``timing constraints'' which were used for \cmpm to ensure that the total makespan of the schedule is upper-bounded by $\ell$, we use ``length constraints'' that ensure that each route has length at most $\lambda$. Recalling the notation used in the proof of Theorem~\ref{thm:fpt}, these constraints will simply upper-bound the following sum by $\lambda$: (1) all the $w_\downarrow$ variables corresponding to the contracted vertical edges on its route, (2) all the $w_\rightarrow$ variables corresponding to the contracted horizontal edges on its route, and (3) the number of non-contracted edges traveled by its route.
  \fi
  
 \begin{theorem} 
 \label{thm:fpt2} \cmpl is \FPT{} parameterized by the number $k$ of robots.
 \end{theorem}
 
 \fi
\section{CMP Parameterized by the Objective Target} 
Having resolved the parameterization by the number $k$ of robots, we now turn our attention to the second fundamental measure in CMP problems, notably the objective target. Unlike the case where we parameterize by the number $k$ of robots, here the complexity of the problem strongly depends on the considered variant. 
We begin by establishing the fixed-parameter tractability of \cmpl\ parameterized by $\lambda$ via an exhaustive branching algorithm. The rest of this section then deals with the significantly more complicated task of establishing the intractability of \cmpl\  parameterized by $\ell$.

\begin{theorem}
\cmpl is \FPT{} parameterized by the objective target $\lambda$.
\end{theorem}

\iflong
\begin{proof}
Let us first observe that there are at most $\lambda$ robots $R\in \R$ such that $R=(s, t)$ for $s\neq t$, and for all of the remaining robots it holds that $s=t$. This is because the travel length for a robot $R=(s, t)$ is at least the distance between $s$ and $t$, and $\lambda$ is the total length traveled by all the robots. Now let $\SSS$ be a valid schedule for $\R$. We can assume that in every time step, at least one robot moves. Hence, we can assume that the makespan of $\SSS$ is at most $\lambda$ as well. Let $R\in \R$ be a robot. There are $5^\lambda$ possible routes that $R$ can make in $\lambda$ time steps. This is true since, in each time step, $R$ can either go up/down/left/right, or stay at its position. We are now ready to describe the exhaustive branching algorithm for \cmpl{} parameterized by $\lambda$.

If for each robot $(s, t)\in \R$ we have $s=t$, then we are done. 
Hence, we can start with an arbitrary robot, let us call it $R_1=(s_1,t_1)$, such that $s_1\neq t_1$. We can now enumerate all the $5^\lambda$ routes starting in $s_1$, and keep only the set $\W_1$ of routes that start in $s_1$ and end in $t_1$. Clearly, one of the routes in $\W_1$ is in $\SSS$. We now branch on each route in $\W_1$ as the route $W_1\in \W_1$ for $R_1$ in $\SSS$ of length $1\le \lambda_1\le \lambda$. 

To describe how we continue in our branching, let us now assume that we have already branched on routes $W_1, W_2, \ldots, W_i$ for robots $R_1, R_2,\ldots, R_i$, respectively, to be included in the schedule $\SSS$ and let $\SSS' = \{W_1, \ldots, W_i\}$. First observe that if the total traveled length by the routes in $\SSS'$ is more than $\lambda$, then these routes are never together in an optimal solution and we can cut this branch. On the other hand, if $\SSS'$ is already a valid schedule for $\R$, then we are done and we can stop our algorithm. There are two possible reasons why $\SSS'$ might not be a schedule yet. Either there is $R\in \R\setminus \{R_1,\ldots, R_i\}$ such that $R=(s,t)$ and $s\neq t$, or for some $R\in \R\setminus \{R_1,\ldots, R_i\}$ with $R=(s,t)$ and $s = t$ one of the routes in $\SSS'$ passes through $s$.
In either case, we let $R_{i+1} = R$, we enumerate all the $5^\lambda$ routes starting in $s$ and compute the subset $\W_{i+1}$ of the routes such that each $W_{i+1}\in \W_{i+1}$ starts in $s$, ends in $t$, and for every $j\in [i]$ the routes $W_j$ and $W_{i+1}$ are non-conflicting. We now try adding to $\SSS'$ each of the routes in $\W_{i+1}$,  and solve each of the at most $5^\lambda$ many branches separately. Note that if $\SSS'\subseteq \SSS$, then in either of the two cases, there exists $W_{i+1}\in \W_{i+1}\cap \SSS$ and in this case $\SSS'\cup \{W_{i+1}\}\subseteq \SSS$. Therefore, if indeed a valid schedule $\SSS$ exists, there is a path in our branching tree that leads to $\SSS$ and the algorithm is correct. 

To analyze the running time of the algorithm, let us observe the following. The depth of the branching/search tree in any possible run of the algorithm is at most $\lambda+1$, since in every node of the branching tree, we test if total traveled length is more than $\lambda$ and stop if that is the case, and since each route we add increases the total traveled length by at least $1$. Moreover, in each node of the search tree we branch into at most $5^{\lambda}$ many branches. Hence, the maximum number of nodes in the search tree is $\bigoh(5^{\lambda^2+\lambda})$. Moreover, in each node, we first check that the total traveled length so far is at most $\lambda$ and check if there is a robot for which we did not compute a route so far in this branch that is in conflict with the computed routes. Both of these checks can be done in time polynomial in $|\R|+\lambda$. Afterwards, we enumerate $5^\lambda$ routes with makespan at most $\lambda$, and check if they end in the correct gridpoint, and do not conflict with any of the previously-constructed routes in the current branch and if so, recursively call the algorithm. Hence, the time spent at a node is at most $\bigoh(5^{\lambda}\cdot\texttt{poly}(|\R|+\lambda))$, resulting in the total running time of 
$\bigoh(5^{\lambda^2+2\lambda}\cdot\texttt{poly}(|\R|+\lambda))$.
\end{proof}
\fi
  
 \subsection{Intractability of \cmpm with Small Makespans}
 \label{sec:span}

The aim of this subsection is to establish that \cmpm is \NP-hard even when the makespan $\ell$ is upper bounded by a constant. Before we proceed to show this \NP-hardness result for \cmpm, we will establish the \NP-hardness of \textsc{$d$-Bounded Length Vertex Disjoint Paths} on grids, as well as its edge variant \textsc{$d$-Bounded Length Edge Disjoint Paths}, which can be seen as a stepping stone for the para-\NP-hardness proof for \cmpm. In fact, the \NP-hardness result for these two classical disjoint paths problems on grids with constant path lengths is significant in its own right, as discussed earlier in the paper. 

All our reductions start from \BPSAT, a problem which is known to be \NP-complete~\cite{HasanMR22,Kratochvil94}. 
The \emph{incidence graph} of a CNF formula is the graph whose vertices are the variables and clauses of the formula, and in which two vertices are adjacent if and only if one is a variable, the other is a clause, and the variable-vertex occurs either as a positive or a negative literal in the clause-vertex. 
\ifshort
In \BPSAT, we are asked to evaluate a CNF formula whose incidence graph is planar and in which each clause contains exactly $3$ distinct literals and each variable occurs in at most $4$ clauses.
On the other hand, in the aforementioned \textsc{$d$-Bounded Length Vertex} (resp.~\textsc{Edge) Disjoint Paths} problems, we are given a graph with a set of vertex-pairs (called \emph{requests}), and are asked to determine if there is a set of vertex (resp.~edge) disjoint paths containing an $s$-$t$ path of length at most $d\in \mathbb{N}$ for every $(s,t)\in R$.

For all three reductions, consider an instance $\varphi$ of  \BPSAT{} and let $G_{\varphi}$ be its incidence graph. We start with an orthogonal drawing $\Omega$ of $G_{\varphi}$ in a polynomial-size grid. Our first goal is to show how to encode the satisfiability of $\varphi$ as an instance of  
\textsc{$d$-Bounded Length Vertex Disjoint Paths} on grids; the reduction for \textsc{$d$-Bounded Length Edge Disjoint Paths} is almost the same, and both can be seen as a stepping stone towards \cmpm. We encode variable assignment and clause satisfaction using bounded-length path requests that conform to the drawing $\Omega$. To model a variable-assignment, we create a variable gadget with a single request between two vertices, $s$ and $t$, on this gadget such that this request can be fulfilled by selecting one of the two $s$-$t$ paths in this gadget, each of length 27. Selecting one of the two paths corresponds to assigning the variable a truth value; an illustration is provided in Figure~\ref{fig:var_gadget}. 
We model clause-satisfaction by creating, for each clause, a clause-gadget, where a clause-gadget for a clause $C$ contains two vertices, $s_C$ and $t_C$, with a request between them that can be fulfilled in one of three ways, each corresponding to choosing a length-27 path between $s_C$ and $t_C$ in the gadget (see Figure~\ref{fig:clause_gadget}).

\begin{figure}[h]
\centering
\includegraphics[scale=0.8,page=4]{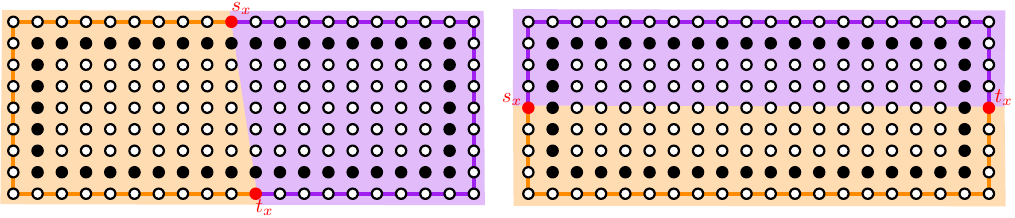}
\caption{Variable Gadget examples. In both cases, there is a request $(s_x, t_x)$. Each of the full black circles is a request $(v,v)$ forcing only two different paths of length at most $27$ between $s_x$ and $t_x$. Examples of a left-right variable gadget (left) and a top-bottom variable gadget (right). }\label{fig:var_gadget}
\end{figure}

\begin{figure}[h]
\centering
\includegraphics[page=3]{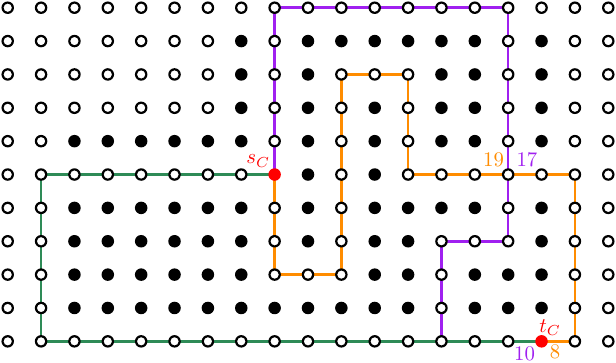}
\caption{Clause gadget example. There is a request $(s_C, t_C)$. Each of the full black circles is a request $(v,v)$. There are three possible ways to leave $s_C$. Choosing to go left forces us to take the green path of length $27$. The orange path going down reaches the intersection point with the purple path (going up) after $19$ steps on the orange path, but only $17$ on the purple. Hence, the purple can choose between going down and taking $10$ steps to reach $t_C$, or going right and taking 8 more steps, but the orange is forced to go right and reach $t_C$ in $8$ steps from the intersection point.}\label{fig:clause_gadget}
\end{figure}

To implement the above idea, we needed to overcome several issues. First, the position of a variable in the embedding could be very far from the position of the clauses that it is incident to, hence prohibiting us from using bounded-length requests to encode the variable-clause incidences. Second, due to planarity constraints, embedding the three paths corresponding to a clause-gadget such that each intersects a different variable gadget, is only possible if two of the clause-paths intersect, which could create shortcuts (i.e., paths that do not intersect the variable gadgets). Third, requests may use grid paths that are not part of the embedding. 

To handle the first issue,  instead of using a single variable-gadget per variable, we use a ``cycle'' of copies of variable gadgets such that a variable assignment in any gadget of this cycle forces the same variable assignment in all copies, thus ensuring assignment consistency. The clause gadget for $C$ is placed around the position of the vertex corresponding to $C$ in $\Omega$, whereas the cycle corresponding to a variable $x$ is placed around the edges of $\Omega$ joining the position of $x$ in $\Omega$ to that of $C$; see Figure~\ref{fig:var_cycle1}. To fit all the variable cycles around a clause gadget in the embedding, we use a connection gadget, which is a path of copies of variable gadgets propagating the same variable assignment as in the corresponding variable cycle.  

\begin{figure}[h]
\centering
\includegraphics[page=11]{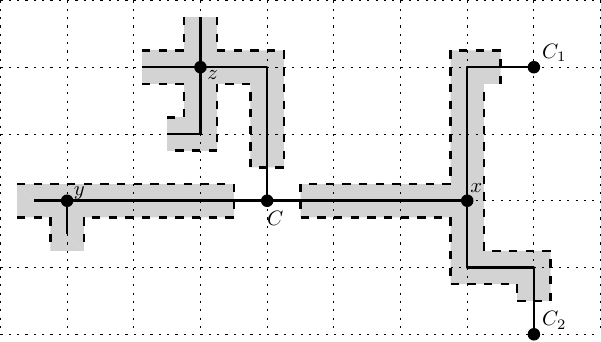}
\caption{Part of an orthogonal drawing of $G_{\varphi}$. Clause $C$ contains variables $x,y,z$. The variable $x$ is also in clauses $C_1$ and $C_2$. The dashed lines represent the variable cycles. \iflong The variable cycle for $x$ is fully drawn.\fi }\label{fig:var_cycle1}
\end{figure}

To model clause-satisfaction for a clause $C$, each of the three $s_C$-$t_C$ paths in the clause gadget of $C$ overlaps with a copy of a variable gadget corresponding to one of the variables whose literal occurs in $C$. If an assignment to variable $x$ whose literal occurs in $C$ does not satisfy $C$, then the path corresponding  to this assignment in the copies of the variable gadgets for $x$ intersects the $s_C$-$t_C$ path corresponding to $x$ in the clause-gadget of $C$, thus prohibiting the simultaneous choice of these clause-path and variable path.

To handle the second issue, we prevent any shortcuts from being taken by making each created shortcut longer than the prescribed upper bound on the path length (i.e., 27). 

Finally, to handle the third issue, when dealing with vertex disjoint paths we can artificially place an obstacle on a vertex $v$ in the grid to ``block'' that vertex (i.e., to prevent it from being used by any path other than $(v, v)$) by adding the request $(v,v)$,
thus forcing the set of possible paths between $s$ and $t$ for every request $(s,t)$, where $s \neq t$, to be chosen from the paths prescribed by the encoding of the instance of \BPSAT. A slight extension of this idea also works for edge disjoint paths. This allows us to establish:

\begin{theorem}
\textsc{$d$-Bounded Length Vertex Disjoint Paths} and \textsc{$d$-Bounded Length Edge Disjoint Paths} are \NP-hard even when restricted to instances where $d = 27$ and $G$ is a grid-graph.
\end{theorem}

These high-level ideas are then used to obtain the targeted \NP-hardness proof of \cmpm for a fixed makespan of $26$, by having an $(s, t)$ path request correspond to routing some robot from its starting gridpoint $s$ to its destination gridpoint $t$. However, the way we force robots to follow the prescribed paths here is completely different and presents the main difficulty when going from \textsc{$d$-Bounded Length Vertex Disjoint Paths} on grids to \cmpm; in particular, it is no longer possible to block certain points on the grid by creating ``dummy requests''. 
To ensure that the prescribed paths are followed, we block certain regions of the embedding by adding a large number of auxiliary non-stationary robots, and coordinating their motion so that they block the desired regions while still allowing the original robots to follow the set of paths prescribed by the encoding; this task turns out to be highly technical.  

We introduce a set of new gadgets whose role is to force the main robots in the reduction to follow the paths prescribed by the embedding. Those gadgets are dynamic, as opposed to the ``static blocker gridpoints'' used in the  \textsc{$d$-Bounded Length Vertex Disjoint Paths} on grids reduction. The two main new gadgets employed are a gadget simulating a ``stream'' of robots and a gadget simulating an ``arrow'' of robots. 

The stream gadget consists of a relatively large number of robots, all moving along the same line, such that each needs to move precisely the makespan many steps in the same direction, and hence cannot afford to waste a single time step.
The robots in the stream gadget will be used to either push the main robots in a certain direction, or to prevent them from taking shorter paths than the prescribed ones. See Figure~\ref{fig:mapf_var_chaining1} for an illustration. In the figure, the main red robot is pushed right by the green stream and forced to move right along the same horizontal line by the two blue streams sandwiching it.

 Figure~\ref{fig:mapf_var_gadget_example1} shows an example of an arrow gadget. In this gadget, there is an orange robot whose destination is 26 steps somewhere down and to the left. 
The gadget is again a ``stream'' of robots that force the orange robot to select one of the two directions towards its destination in the first step, and then to stick to this selection for a number of steps that depends on the size of the arrow. For example, in Figure~\ref{fig:mapf_var_gadget_example1}, there is a ``right arrow'' of green robots that all want to go 26 steps right. Since the orange robot has a slack of $0$, the right arrow forces it to either take the first 5 steps all to the left, or the first 7 steps all down.

 \begin{figure}[h]
\centering
\includegraphics[page=19, scale=2]{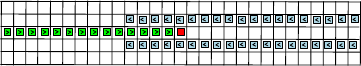}
\caption{Example of streams.}
\label{fig:mapf_var_chaining1}
\end{figure}

\begin{figure}[h]
\centering
\includegraphics[page=12, scale=2]{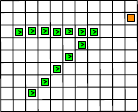}
\caption{An example of an arrow.} 
\label{fig:mapf_var_gadget_example1}
\end{figure} 

Other gadgets are needed to ensure that robots in the stream and arrow gadgets do not collide with anything. Using such enforcement gadgets, we can simulate the gadgets constructed in the reduction for \textsc{$d$-Bounded Length Vertex Disjoint Paths} on grids, thus encoding the instance of \BPSAT{} as an instance of \cmpm. 

\begin{theorem}
\label{thm:cmpmhard}
\cmpm\ is \NP-hard even when restricted to instances where $\ell=26$. 
\end{theorem}
\fi
 
 \iflong
 We define the \BPSAT{} problem and the {\sc $d$-Bounded Length Vertex Disjoint Paths} on grids problem formally:

 \normalproblem{\BPSAT}{  A CNF formula $\varphi$ with a set of clauses $C$ over a set of variables $X$ satisfying: 
  (i) every clause contains exactly $3$ distinct literals;
 (ii) every variable occurs in at most $4$ clauses; and
 (iii) the incidence graph of $\varphi$ is planar.
  }{Is $\varphi$ satisfiable?}

\normalproblem{$d$-Bounded Length Vertex Disjoint Paths}{
 A graph $G= (V,E)$ and a set $R$ of pairs of vertices called \emph{requests}. }{Is there a set $\mathcal{P}$ of vertex-disjoint paths, each of length at most $d\in \mathbb{N}$, such that there is an $s$-$t$ path in $\mathcal{P}$ for every $(s,t)\in R$?}

The proofs for establishing the intractability of both path problems on grids and \cmpm start exactly the same, but the proof for \cmpm{} becomes much more technically involved in later steps. We fix an instance of \BPSAT\ with a set $X$ of variables and a set $\mathcal{C}$ of clauses, and we let $G_{\varphi}$ be the incidence graph of $\varphi$. Note that every variable occurs in at most four clauses and every clause has three variables, and hence the maximum degree of $G_{\varphi}$ is at most 4. The first step of our proof is to embed the graph $G_{\varphi}$ in a grid.  

An \emph{orthogonal drawing} of a graph is an embedding in the plane such that the graph vertices are drawn as points and its edges are drawn as sequences of horizontal and vertical line (i.e., rectilinear) segments. An orthogonal drawing is \emph{plane} if the edges intersect only at their endpoints. The following theorem shows that every planar graph of maximum degree at most 4 admits a plane orthogonal drawing.

\begin{longtheorem}[\cite{TamassiaTollis89}]\label{thm:orthogonal_embedding}  	Given a planar graph $G$ of degree at most $4$ with $n>4$ vertices, there is a linear time algorithm that constructs a plane orthogonal drawing of $G$ on a grid with area $\bigoh(n^2)$. \end{longtheorem}
\begin{figure}[h]
\centering
\includegraphics[page=11]{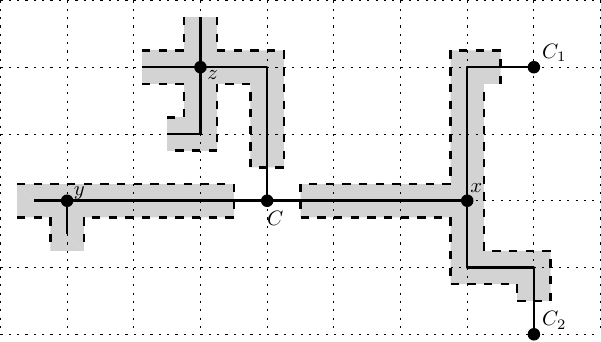}
\caption{Part of an orthogonal drawing of $G_{\varphi}$. Clause $C$ contains variable $x,y,z$. The variable $x$ is also in clauses $C_1$ and $C_2$. The dashed lines represent the variable cycles. The variable cycle for $x$ is fully drawn.  }\label{fig:var_cycle}
\end{figure} 

By Theorem~\ref{thm:orthogonal_embedding}, we can compute a plane orthogonal drawing $\Omega$ of $G_{\varphi}$ in a grid with $\bigoh(n^2)$ vertices. Next, we refine the grid underlying $\Omega$ by a factor of $1000$---more formally, we replace each cell in the grid underlying $\Omega$ with a $1000\times 1000$ subgrid. We note that the refinement by a factor of $1000$ was not chosen very accurately, and probably a smaller refinement would be sufficient. (We need a large-enough refinement to have ample space for constructing and connecting the gadgets in the reduction.) 

For all three problems, the instances obtained from the reductions will be composed of three parts: 

\begin{description}
\item[Clause Gadget.] For every clause $C$, there is a ``clause gadget'' that will be placed around the point $\Omega(C)$, where the vertex representing the clause $C$ is embedded.
\item[Variable Cycle.] For every variable $x$, there is a ``variable cycle''. The role of the variable cycle is to propagate the information about the choice/assignment made for the variable to the clause. This is necessary since the edges of the orthogonal drawing can be very long (i.e., span many edges of the grid) and can have bends. One can think about this as replacing the variable with a long cycle of implications, stipulating that the choices made along the cycles are consistent. For us, it will be a very long cycle consisting of different instances of a small \emph{``variable gadget''} connected in a way so that the choices in all the gadgets along the cycle have to be the same. The variable cycle wraps around the edges incident to $x$ in the drawing $\Omega$; see Figure~\ref{fig:var_cycle}. More precisely, the variable cycle for a variable $x$ is a rectilinear polygon that follows the drawing of the edges incident to $x$ at distance $250$ ($1/4$ of the width of an ``original'' cell of the drawing $\Omega$ before the refinement) and crosses the edges at distance $500$ ($1/2$ of the width of an ``original'' cell of the drawing $\Omega$ before the refinement) from each of the clause vertices that $x$ belongs to. This way, variable cycles are far apart, each straight-line segment (i.e., each edge of the polygon) of the variable cycle has length at least $250$, and each clause vertex is at distance $500$ from the variable cycle. In a concrete instance of each of the problems, this cycle will be constructed from a series of \emph{``variable gadgets''} that roughly follow the outline of this cycle. 
\item[Connection Gadget.] For every pair of a clause and a variable in the clause, there is a ``connection gadget'' that connects the clause gadget to the variable cycle. For both problems, this will be a path of ``variable gadgets'' roughly following the drawing of the edge between the clause and the variable until it intersect the variable cycle.  
\end{description}

This finishes the common part of the two proofs. In the next subsection, we will show the \NP-hardness of \textsc{$d$-Bounded Length Vertex Disjoint Paths} on grids, and in the last subsection, we will show the para-\NP-hardness of \cmpm.

\subsection{Bounded Length Disjoint Paths on Grids}
In this subsection, we show the \NP-hardness of \textsc{$d$-Bounded Length Vertex Disjoint Paths} on grids with constant path length (specifically, $27$), and extend this result to its edge variant. We first present the \NP-hardness proof for \textsc{$d$-Bounded Length Vertex Disjoint Paths} on grids. For \textsc{$d$-Bounded Length Vertex Disjoint Paths} on grids, the fact that the paths in the solution have to be vertex disjoint allows us to artificially place an obstacle on a vertex $v$ in the grid to ``block'' that vertex (i.e., to prevent it from being used by any path other than $(v, v)$) by adding the request $(v,v)$ to $R$, thus forcing the set of possible paths between $s$ and $t$ for every $(s,t) \in R$, where $s \neq t$, to be chosen from the paths prescribed by the encoding of the instance of \BPSAT. The choice of a path corresponds to either a variable-assignment in variable gadgets, or to the choice of the variable that satisfies a clause in our clause gadgets.  

The main result of this subsection is the following. 

 \begin{theorem}\label{thm:disjoint_paths_hard}
\textsc{$d$-Bounded Length Vertex Disjoint Paths} is \NP-hard even when restricted to instances where $d = 27$ and $G$ is a grid-graph.
 \end{theorem}
 
 We start by giving a low-rigor description of the \NP-hardness reduction from \BPSAT{} to \textsc{$d$-Bounded Length Vertex Disjoint Paths}.
  We will encode variable assignment and clause satisfaction using bounded-length path requests that conform to the drawing $\Omega$. To model a variable-assignment, we create a variable gadget with a single request between two vertices, $s$ and $t$, on this gadget such that this request can be fulfilled by selecting one of the two $s$-$t$ paths in this gadget, each of length 27. Selecting one of the two paths corresponds to assigning the variable a truth value. 
We model clause-satisfaction by creating, for each clause, a clause-gadget, where a clause-gadget for a clause $C$ contains two vertices, $s_C$ and $t_C$, with a request between them that can be fulfilled in one of three ways, each corresponding to choosing a length-27 path between $s_C$ and $t_C$ in the gadget. 

To implement the above idea, it was necessary to overcome several issues. First, the position of a variable in the embedding could be very far from the position of the clauses that it is incident to, hence prohibiting us from using bounded-length requests to encode the variable-clause incidences. Second, due to planarity constraints, embedding the three paths corresponding to a clause-gadget such that each intersects a different variable gadget, is only possible if two of the clause-paths intersect, which could create shortcuts (i.e., paths that do not intersect the variable gadgets). Third, requests may use grid paths that are not part of the embedding (i.e., are not intended to be used). 

To handle the first issue,  instead of using a single variable-gadget per variable, we use a ``cycle'' of copies of variable gadgets such that a variable assignment in any gadget of this cycle forces the same variable assignment in all copies, thus ensuring assignment consistency. The clause gadget for $C$ is placed around the position of the vertex corresponding to $C$ in $\Omega$, whereas the cycle corresponding to a variable $x$ is placed around the edges of $\Omega$ joining the position of $x$ in $\Omega$ to that of $C$. To fit all the variable cycles around a clause gadget in the embedding, we use a connection gadget, which is a path of copies of variable gadgets propagating the same variable assignment as in the corresponding variable cycle.  
 
To model clause-satisfaction for a clause $C$, each of the three $s_C$-$t_C$ paths in the clause gadget of $C$ overlaps with a copy of a variable gadget corresponding to one of the variables whose literal occurs in $C$. If an assignment to variable $x$ whose literal occurs in $C$ does not satisfy $C$, then the path corresponding  to this assignment in the copies of the variable gadgets for $x$ intersects the $s_C$-$t_C$ path corresponding to $x$ in the clause-gadget of $C$, thus prohibiting the simultaneous choice of these clause-path and variable path.

To handle the second issue, we prevent any shortcuts from being taken by making each created shortcut longer than the prescribed upper bound on the path length (i.e., 27). 

Finally, to handle the third issue, the fact that the paths in the solution have to be vertex disjoint allows us to artificially place an obstacle on a vertex $v$ in the grid to ``block'' that vertex (i.e., to prevent it from being used by any path other than $(v, v)$) by adding the request $(v,v)$, thus forcing the set of possible paths between $s$ and $t$ for every request $(s,t)$, where $s \neq t$, to be chosen from the paths prescribed by the encoding of the instance of \BPSAT. 
 
 We now proceed to give the formal proof.
There are two main building blocks in the \NP-hardness reduction: the \emph{variable gadget} and the \emph{clause gadget}. 

\subparagraph{Variable Gadget.} This gadget is rather flexible and consists of a rectangular $a\times b$ grid such that $2(a+b-2)=54$ (i.e., there are $54$ vertices on the outer layer of the gadget) and both $a$ and $b$ are at least four. We do exploit the 
flexibility of this gadget in the connection gadget, which connects the clause gadget to the variable cycle gadget;
see Figure~\ref{fig:var_gadget} for some examples of this variable gadget, with $a=20$ and $b=9$. There is a request $(s,t)$ such that both $s$ and $t$ are on the outer layer of the gadget and the lengths of the two paths that use only the vertices on the outer layer are the same and equal to $27$ in both cases. For the vertices inside the gadget, we have requests of the form $(v,v)$ (for each such vertex $v$) that block the $s$-$t$ path from using any vertex inside the gadget. Therefore, there are only two possible $s$-$t$ paths that are disjoint from all $(v,v)$ paths and have length at most $27$---the clockwise path and the counter-clockwise path following the outer layer. The vertices $s$ and $t$ will either be on the top and bottom of the gadget---in this case the leftmost column of the gadget is fully on one of these paths and the rightmost column on the other, or they can be on the left and right sides of the gadget, where the top row is fully on one path and the bottom row on the other. We will call the clockwise path the \emph{orange path} and the counter-clockwise path the \emph{purple} path, and we will use these colors to depict the gadgets in all our figures. 

Both the variable cycle and connection gadgets consist of many variable gadgets connected to each others in a way that a choice of one variable gadget propagates to the others. Each variable cycle is a long cycle of these gadgets. This forces that the only way to satisfy all requests of a variable cycle will be to either choose only the orange paths or only the purple paths for all gadgets on the cycle.
Similarly, the connection gadget is a path of these gadgets. 
 Since the height and width of the variable gadgets are very flexible and we subdivided the grid such that these cycles are very long and consist of very long straight-line stretches, we can always place the variable gadgets so that they conform to the outline of the variable cycles and the connection gadgets stipulated by the drawing $\Omega$.

\begin{figure}[h]
\centering
\includegraphics[page=4, width=\textwidth]{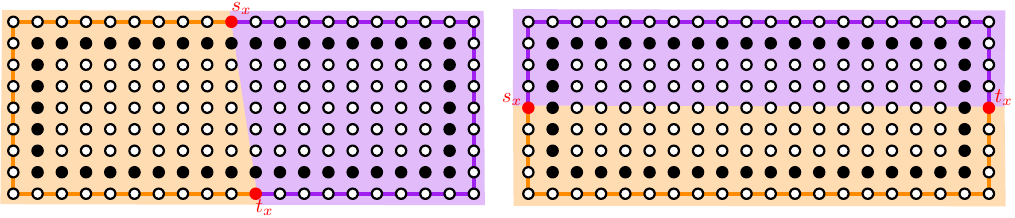}
\caption{Variable Gadget examples. In both cases, there is a request $(s_x, t_x)$. Each of the full black circles is a request $(v,v)$ forcing only two different paths of length at most $27$ between $s_x$ and $t_x$. On the left there is an example of left-right variable gadget and on the right a top-bottom variable gadget. }\label{fig:var_gadget}
\end{figure}

\subparagraph{Clause Gadget.} We have one gadget for every clause $C$ of $\varphi$. This gadget is very rigid and there is only one type for it, up to rotations by 90/180/270 degrees. The gadget is depicted in Figure~\ref{fig:clause_gadget}. We will describe only one of the four possible rotations of the gadget, which is the one depicted in the figure. The remaining rotations are analogous. There is a vertex $s_C$ at position $\Omega(C)$, where the drawing $\Omega$ places the vertex representing clause $C$ on the grid. There is a vertex $t_C$ at the position $\Omega(C)+(8,-5)$. There is a request $(s_C,t_C)$. The blocking requests of the form $(v,v)$ force only three possible ways for a path to leave from $s_C$:
\begin{description}
\item[Left.] Here we are forced to go seven steps left, then five steps down and 15 steps to the right. This is the only path of length at most $27$ that starts by going left from $s_C$.
\item[Down.] Here we are forced to go three steps down, two to the right, six up, two right, three down, and three right. This brings us to 19 steps and leaves us with only 8 steps. While the point going down is not blocked, it leads to a path that needs an additional ten steps to reach $t_C$, which would bring the total to $29$. Hence, we have to continue two more steps to the right, then five down and finally one left. The total length of the path is then $27$. 
\item[Up.] Here we are forced to go five steps up. Then we can go seven steps right, and five steps down, reaching the intersection point that we reached on the ``Down'' path after 19 steps. However, we arrive to this intersection point after only 17 steps, and we can continue down on the path taking 10 additional steps and arriving after $27$ steps to $t_C$ (purple path), or we can go right and reach $t_C$ in $25$ steps (orange path). Since both options arrive to $t_C$ within the prescribed makespan, and taking the purple path from the intersection point down to $t_C$ does not intersect any variable gadgets, it is safe to assume that in this case the robot follows the longer purple path.
Note that there are a few more paths of length 25 or 27 that go 5 steps up. This is not an issue since, as we will see after introducing the connection gadgets, all these paths intersect the variable gadget that the purple path intersects. 
\end{description}

\begin{figure}[h]
\centering
\includegraphics[page=3]{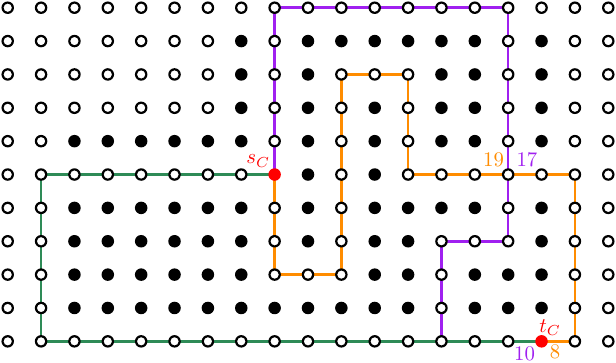}
\caption{Clause gadget example. There is a request $(s_C, t_C)$. Each of the full black circles is a request $(v,v)$. There are three possible ways to leave $s_C$. Choosing to go left forces us to take the green path of length $27$. The orange path going down reaches the intersection point with the purple path (going up) after $19$ steps on the orange path, but only $17$ on the purple. Hence, the purple can choose between going down and taking $10$ steps to reach $t_C$, or going right and taking 8 more steps, but the orange is forced to go right and reach $t_C$ in $8$ steps from the intersection point.}\label{fig:clause_gadget}
\end{figure}

After describing these two building blocks, we are now ready to describe the full reduction. 

We start by placing, for each clause $C$ in the instance, the clause gadget at position $\Omega(C)$. The rotation of the gadget depends on which of the four grid-edges incident to $\Omega(C)$ does not contain a drawing of any of the three edges incident to $C$. For example, the rotation of the clause gadget in Figure~\ref{fig:clause_gadget} is for the case where there is no edge going down from $C$ in the drawing $\Omega$. If the edge going left were missing, we would rotate the gadget 90 degrees clockwise around $s_C$. Afterwards, we place the variable gadgets in a cycle along the variable cycle such that two consecutive variable gadgets always intersect on one of the paths for the gadget.  See Figure~\ref{fig:var_cycle_example} for an illustration of a (very short) variable cycle. It is rather straightforward to verify that there are only two ways to selects all the $s_x$-$t_x$ paths of length at most $27$ for all the variable gadgets on the variable 
cycle: in the figure we either take all ``purple'' paths or all ``orange'' paths. We will associate the setting of a variable to True with the choice of taking all the orange paths, and to False with that of taking all the purple paths. 
  \begin{figure}[h]
 \centering
 \includegraphics[page=12, width=\textwidth]{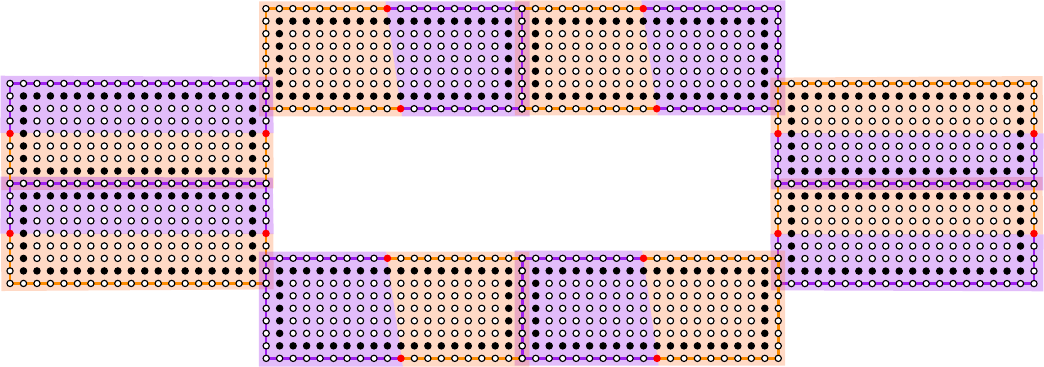}
 \caption{An very small example of a variable cycle. }\label{fig:var_cycle_example}
 \end{figure}

Finally, we describe how to place the connection gadget. This is again just a path of variable gadgets going between a clause gadget and a variable cycle. Note that consecutive variable gadgets do not need to be perfectly aligned; they just need to intersect. Moreover, the width/height of a variable cycle is variable (between $4$ and $25$). Hence, we can easily shift the gadgets inside the path to perfectly connect to the the gadgets as needed. The connections between a clause gadget and the connection gadgets is depicted in Figure~\ref{fig:clause_connection}. Note that all $s_C$-$t_C$ paths of length at most $27$ 
of the clause gadget intersect one of the variable gadgets at the end of a connection gadget. Moreover, an $s_C$-$t_C$ path $P$ intersects the variable gadget on an orange path if and only if the negation of the variable is in the clause. This means that only if the purple path of the variable gadget is chosen (signifying a variable being assigned False), we can use the path $P$ to reach from $s_C$ to $t_C$. Otherwise, $P$ intersects the variable gadget in the purple path, i.e., if the orange path in the variable gadget is selected, $P$ can be selected as an $s_C$-$t_C$ path. 
There are also two ways to connect the connection gadget and the variable-cycle depending on whether the variable appears positively or negatively. See Figure~\ref{fig:cycle_connection} for an illustration of how these two possibilities are achieved for the case where the clause is to the left of a variable cycle. The remaining cases are symmetric. The basic idea is that we shift one of the variable gadgets away of the clause gadget; this allows us to connect a variable gadget to either the orange path of the gadget immediately after the shifted gadget or the purple path of the gadget immediately before the shifted gadget without interfering with any other gadgets.

\begin{figure}[h]
\centering
\includegraphics[page=5,, width=\textwidth]{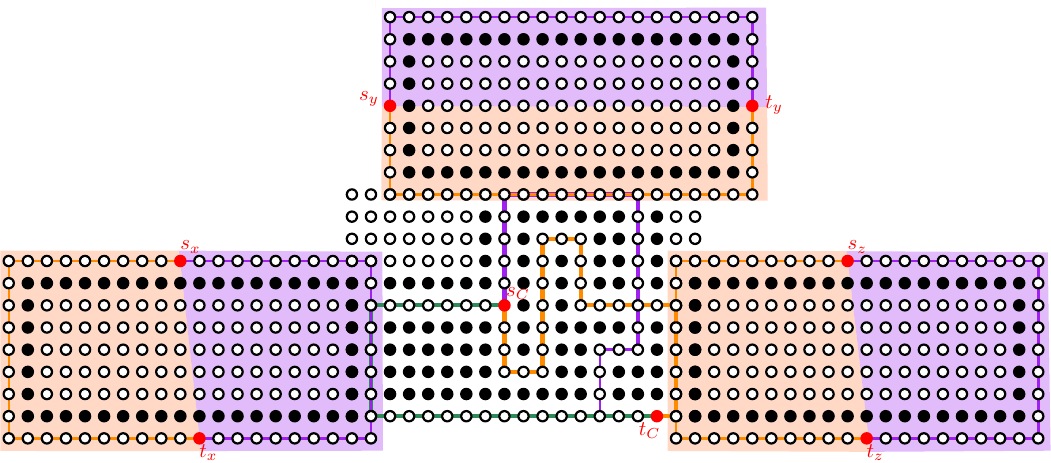}
\caption{Connecting a clause gadget for clause $C= (x\vee \neg y \vee \neg z)$ to respective clause-variable connection gadgets. }\label{fig:clause_connection}
\end{figure}

  \begin{figure}[h]
  \centering
  \begin{subfigure}[b]{0.48\textwidth}
   \centering
    \includegraphics[page=10,scale=.36]{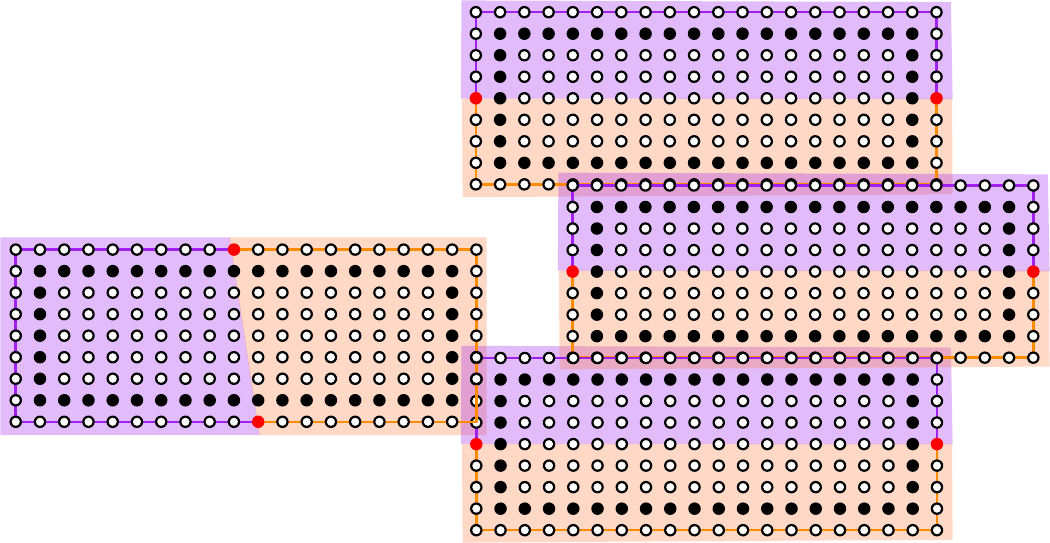}
    \subcaption{The variable $x$ appearing positively in $C$.}
   \end{subfigure}
   \begin{subfigure}[b]{0.48\textwidth}
    \includegraphics[page=9,scale=.36]{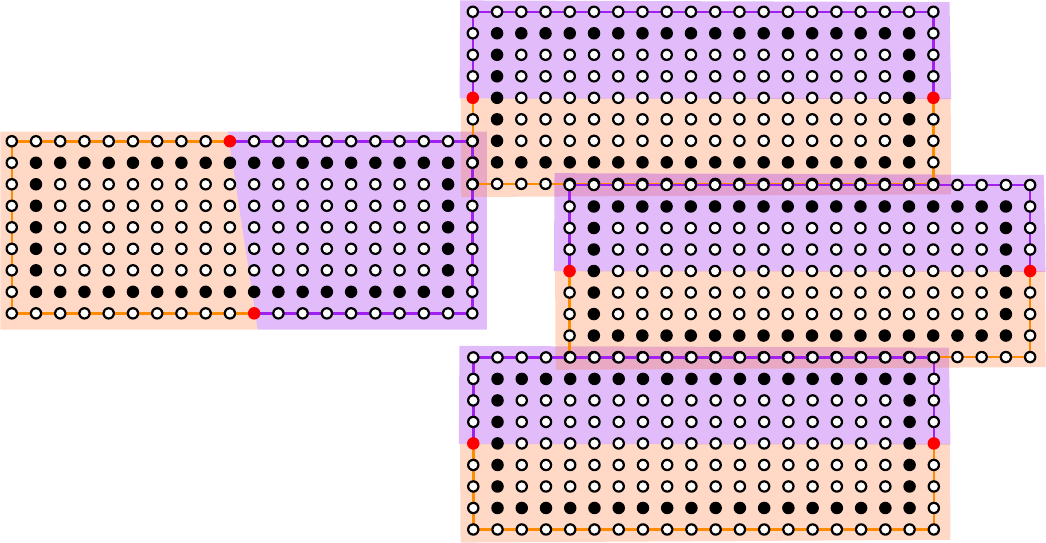}
    \subcaption{The variable $x$ appearing negatively in $C$.}
   \end{subfigure}
    
  \caption{Connecting to a ``variable cycle''. In this case the clause gadget is to the left of the variable cycle. The other three possibilities are symmetric.}\label{fig:cycle_connection}   \end{figure}

This finishes the construction and we are ready to prove Theorem~\ref{thm:disjoint_paths_hard}. 

\begin{proof}[Proof of Theorem~\ref{thm:disjoint_paths_hard}]
We start from an instance $\varphi$ of \BPSAT\ with a set of variables $X$, a set of clauses $\mathcal{C}$, and the incidence graph $G_{\varphi}$. We follow the construction described above and construct an instance $(G, R, 27)$ of \textsc{$d$-Bounded Length Vertex Disjoint Paths} on grids. Note that this construction takes polynomial time, and the number of vertices of $G$ is $\bigoh((n+m)^2)=\bigoh(n^2)$, where $n= |X|$ and $m=|\mathcal{C}|$. We now prove that $\varphi$ is satisfiable if and only if there is a set $\mathcal{P}$ of vertex disjoint paths, each of length at most $27$, such that for every $(s,t)\in R$, there is an $s$-$t$ path in $\mathcal{P}$. 

Let $\alpha: X \rightarrow \{\text{True, False}\}$ be a satisfying assignment for $\alpha$. For all requests $(v,v)\in R$, we add the single vertex path $(v)$ to $\mathcal{P}$. Each variable gadget in the instance is either on a variable cycle for some variable $x$, or on a connection gadget between a variable $x$ and a clause $C$ containing $x$ or its negation. Hence, each variable gadget is associated with a single variable. If the variable gadget is associated with variable $x$ such that $\alpha(x)=$True, then we choose  the orange path on this variable gadget; otherwise, we choose the purple path. It is straightforward to see---from the construction of the gadgets---that all of these paths are vertex disjoint. Moreover, if a variable $x$ appears positively in a clause $C$, then there is an $s_C$-$t_C$ path of length 27 that intersects only the single variable gadget associated with variable $x$ on the purple path. Hence, if $\alpha(x)=$True and $x$ satisfies $C$, then this path is disjoint from all the other chosen paths. Similarly, if a variable $x$ appears negatively in a clause $C$, then there is an $s_C$-$t_C$ path of length 27 that intersects only the single variable gadget associated with variable $x$ on the orange path. Hence, if $\alpha(x)=$False and $x$ satisfies $C$, then this path is disjoint from all other chosen paths. Since $\alpha$ is a satisfying assignment, each clause is satisfied by at least one variable, and for all clauses $C\in \mathcal{C}$, we can choose an $s_C$-$t_C$ path that is disjoint from all the other paths in $\mathcal{P}$. This finishes the construction of the solution and the forward direction of the proof.

For the converse, let $\mathcal{P}$ be a solution for the instance $(G, R, 27)$. Note that, by construction, either all the variable gadgets on the single variable cycle of a variable $x$ use the orange paths, or they all use the purple paths. We construct the  assignment $\alpha$ such that $\alpha(x)$=True if the variable gadgets on the variable cycle for $x$ use the orange paths, and $\alpha(x)$=False otherwise. We claim that this is a satisfying assignment. Let $C\in \mathcal{C}$. Clearly there is an $s_C$-$t_C$ path $P \in \mathcal{P}$ of length at most $27$ that does not intersect any other paths in $\mathcal{P}$.
From the description of the clause and connection gadgets, it is rather straightforward to verify that $P$ intersects at least one variable gadget of the connection gadget between the clause gadget and the variable cycle. Let us assume that $P$ intersects the variable gadget for variable $x$. We show that in this case $x$ satisfies clause $C$. If $x$ appears positively in $C$, then the $s_C$-$t_C$ path $P$ intersects the purple path of the gadget and the variable gadget has to choose the orange path. By the construction of the connection gadget, the orange path of this first variable gadget connected to the clause gadget intersects a purple path of a second variable gadget on the connection gadget and the second variable gadget also has to choose the orange path, and so on until the last variable gadget on the connection gadget intersects a variable gadget of a variable cycle for $x$ on a purple path. Hence, the variable gadgets on the variable cycle for $x$ all choose the orange path and $\alpha(x)=$True and $x$ satisfies $C$. The case where $x$ appears negatively in $C$ is analogous; the only difference is that the $s_C$-$t_C$ path $P$ intersects the variable gadget in the orange path in this case. It follows that $C$ is indeed satisfied by the assignment $\alpha$. Since $C$ was chosen arbitrarily, $\alpha$ is a satisfying assignment for $\varphi$, and the proof is complete.
\end{proof}

The above \NP-hardness result can be extended to the edge-disjoint paths problem, which has been studied extensively, in the case where the path length is upper bounded by a constant: 

\normalproblem{$d$-Bounded Length Edge Disjoint Paths}{
 A graph $G= (V,E)$, a set $R$ of pairs of vertices, and an integer $d\in \nat$. }{Is there a set $\mathcal{P}$ of edge-disjoint paths, each of length at most $d$, such that there is an $s$-$t$ path in $\mathcal{P}$ for every $(s,t)\in R$?}

We will now describe how to modify the constructed instance of \textsc{$d$-Bounded Length Vertex Disjoint Paths} to obtain an equivalent instance of \textsc{$d$-Bounded Length Edge Disjoint Paths}.
First, notice that we already constructed our gadgets in a way that if the non-trivial paths (i.e., paths of length at most $27$ for the unique request $(s,t)$ in the gadget where $s\neq t$) of some variable gadget and some clause gadget intersect, then they also intersect in an edge. Hence, we only need to enforce that the path for the $(s,t)$ request in each gadget does not go through any vertex with request $(v,v)$ (for some $v\in V(G)$) without blocking any of the edges that the $s$-$t$ gadget paths use in  
the solution for the instance of \textsc{$d$-Bounded Length Vertex Disjoint Paths}. This is rather straightforward. Let $(v,v)$ be a request in the constructed instance of \textsc{$d$-Bounded Length Vertex Disjoint Paths} for some vertex at the grid position $(a,b)$. Let $v_L, v_R, v_U, v_D$ be the vertices at positions $(a-1,b), (a+1,b), (a,b+1), (a,b-1)$, i.e., on the right, left, top, and bottom, respectively, of vertex $v$ in the grid. For $i\in \{L, R, U, D\}$, we add the request $(v, v_i)$ to the instance of  \textsc{$d$-Bounded Length Edge Disjoint Paths}. Note that if $(v_i, v_i)$ is also a request in the original reduction, then we are also adding the request $(v_i, v)$, however, the graph is not directed, hence we treat the request $(v, v_i)$ and $(v_i,v)$ as the same request and a single path can satisfy both of them. Note that, in order to satisfy all these requests with edge-disjoint paths, we have to use all the edges incident to $v$, as there are four edges and four requests for edge-disjoint paths. Moreover, taking the paths $(v,v_R),(v,v_L), (v,v_U), (v,v_D)$ satisfies these requests. Therefore, on one hand, no path between two vertices $s$ and $t$ such that $s\neq v$ and $t\neq v$ can pass through $v$ in a valid solution of \textsc{$d$-Bounded Length Edge Disjoint Paths}. On the other hand, we can always assume that $(v,v_R),(v,v_L), (v,v_U), (v,v_D)$ are the paths satisfying the requests for $v$ and their requests do not block any edge that is not incident to $v$. Hence, the set of valid $s$-$t$ paths of length at most $27$ that we can use to satisfy all the remaining requests remains the same and the rest of the \NP-hardness proof follows from the \NP-hardness proof for  \textsc{$d$-Bounded Length Vertex Disjoint Paths}. 

\begin{longtheorem}\label{thm:disjoint_edge_paths_hard}
  \textsc{$d$-Bounded Length Edge Disjoint Paths} is \NP-hard even when restricted to instances where $d = 27$ and $G$ is a grid-graph.
   \end{longtheorem}

\subsection{\NP-hardness of \cmpm with Constant Makespan}
 In this subsection, we show the \NP-hardness of \cmpm for makespan $26$.  The high-level ideas are the same as in the \NP-hardness proof for \textsc{$d$-Bounded Length Vertex Disjoint Paths} starting from \BPSAT: An $(s, t)$ path request corresponds to routing some robot from its starting gridpoint $s$ to its destination gridpoint $t$. However, the way we force robots to follow the prescribed paths here is completely different and presents the main difficulty when going from \textsc{$d$-Bounded Length Vertex Disjoint Paths} on grids to \cmpm. To ensure that the prescribed paths are followed, we block certain regions of the embedding by adding a large number of auxiliary non-stationary robots, and coordinating their motion so that they block the desired regions while still allowing the original robots to follow the set of paths prescribed by the encoding; this task turns out to be highly technical. 
 
 To overcome these complications, we introduce a set of new gadgets whose role is to force the main robots in the reduction to follow the paths prescribed by the embedding. Those gadgets are dynamic, as opposed to the static blocker gridpoints used in the  \textsc{$d$-Bounded Length Vertex Disjoint Paths} on grids reduction. The two main new gadgets employed are a gadget simulating a ``stream'' of robots and a gadget simulating an ``arrow'' of robots. The stream gadget consists of a relatively large number of robots, all moving along the same line, such that each needs to move precisely the makespan many steps in the same direction, and hence cannot afford to waste a single time step.
The robots in the stream gadget will be used to either push the main robots in a certain direction, or to prevent them from taking shorter paths than the prescribed ones.  

The starting point of the reduction is exactly the same as that in the previous subsection. We start from an instance $\varphi$ of \BPSAT\ with a set of clauses $\mathcal{C}$, a set of variables $X$, and incidence graph $G_{\varphi}$ that is planar. We obtain an orthogonal drawing of $G_{\varphi}$ on a grid with area $\bigoh(n^2)$, and refine it by a factor of $1000$. We call the final drawing $\Omega$, and let $\Omega(C)$, for $C\in \mathcal{C}$, be the point of the grid where the vertex of $G_\varphi$ representing clause $C$ is placed. We also assume that we have the outline of the variable cycle for every variable, and we now only need to describe the clause gadget, the variable gadget that we have on the variable cycle, and how the connection gadget connects to the variable cycle and the clause gadgets. 

\subsubsection{Clause Gadget}
We start with the description of the clause gadget. This gadget is again rather rigid, and there are only four different placements of a clause gadget, which correspond to the gadget and its three rotations by 90/180/270 degrees, depending on which direction there is no outgoing edge out of $C$. See Figure~\ref{fig:mapf_clause_gadget} for the gadget in the case where there is no edge going to the right out of $C$. 

\begin{figure}[h]
\centering
\includegraphics[page=1, width=\textwidth]{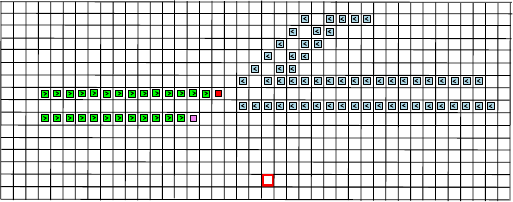}
\caption{An illustration of a clause gadget.} 
\label{fig:mapf_clause_gadget}
\end{figure}

All the blue robots in the figure want to go 26 steps to the left. All the green robots want to go 26 steps to the right. The red robot, which we call the \emph{``clause robot''}, wants to go 7 steps down and 4 steps right. The red robot starts at position $\Omega(C)$. Finally, the pink robot wants to go 25 steps to the right and 1 up. In order for the blue and green robots to reach their respective destinations in at most $26$ steps, they have to move left and right, respectively, in each step. It is rather straightforward to see that the purple robot can go up only after it passes the stream of the blue robots. 

\begin{longlemma}\label{lem:clause_robot_possitions}
If the clause robot reaches its destination in 26 steps then one of the following holds:
\begin{enumerate}
\item The red robot is at position $\Omega(C)+(11,-1)$ after 12 steps, i.e., 11 steps right and one down. 
\item The red robot is at position $\Omega(C)+(-7,-1)$ after 8 steps, i.e., one step down and 7 left. 
\item The red robot is at position $\Omega(C)+(2,7)$ or at position $\Omega(C)+(1,7)$ after 9 steps, i.e., 7 steps up, and either one or two to the right. Moreover, in the 10th step, the robot either waits or moves right.  
\end{enumerate}
\end{longlemma}

\begin{proof}
We can distinguish the following cases based on the first starting steps of the red clause robot: (1) either the clause robot goes at least two steps to the right, or (2) the robot goes one step down, or (3) the robot goes one step right followed by one step up or immediately goes one step up. Note that the robot cannot start by going left (as it would collide with the green robots), and also cannot go one step right followed by one step down (as it would collide with the blue robots).

If the first two steps of the clause robot are both to the right, then since the blue robots are moving only to the left in each step, and the green robots are moving only to the right in each step, it follows that in each position between the position after step 2 and the position after step 10:
\begin{itemize}
\item There is a green robot to the left of the clause robot forbidding the robot to move left or to wait.
\item There is a blue robot that would be above the clause robot in the next step. 
\item There is a blue robot that would be below the clause robot in the next step.
\end{itemize} 
After step 11, the robot is free to go up, down, or to continue right. However, it cannot wait or go left. Moreover, it is at position $\Omega(C) + (11,0)$, which is at distance $14$ from $\Omega(C) + (4,7)$. Going up or right would mean that, after step 12, it would be at distance $15$ from the final destination with only $14$ steps left. Since it cannot wait due to the presence of a green robot to its left, it has to go one step down and reach $\Omega(C)+(11,-1)$ at step $12$. Note that in the 13th step, the purple robot could finally do its one step up and allow the clause robot to go down. After that, the route to the destination for the clause robot is not blocked by anything, and it can reach it by step 26 (even 25) by just going towards the destination in every step. See Figure~\ref{fig:mapf_clause_going_right} for an illustration of the configurations after 11, 12, and 13 steps, respectively. 

\begin{figure}[h]
\centering
\begin{subfigure}[b]{.48\textwidth}
\includegraphics[page=2, scale=.75]{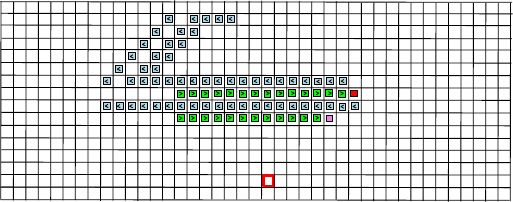}
\end{subfigure}
\begin{subfigure}[b]{.48\textwidth}
\includegraphics[page=3,scale=.75]{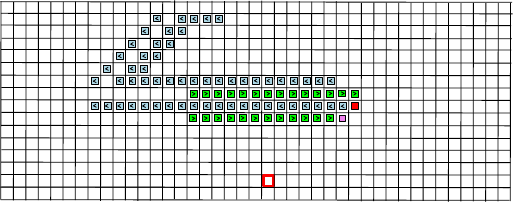}
\end{subfigure}
\begin{subfigure}[b]{.48\textwidth}
\includegraphics[page=4,scale=.75]{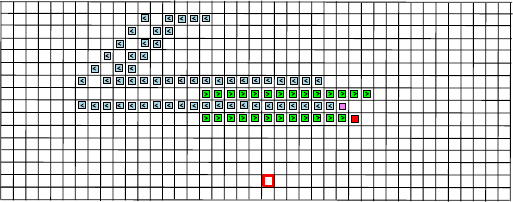}
\end{subfigure}

\caption{Clause gadget. A clause robot going right after steps 11, 12, and 13.}
\label{fig:mapf_clause_going_right}
\end{figure}

If the first step is down, then after each step between step one and step seven, there is a robot (green or purple) that would be in the next step below and above of the clause robot. Moreover, there is a robot to the right that is going left. Hence, the clause robot is forced to move left between step one and step seven, and reaches position $\Omega(C)+(-7,-1)$ after 8 steps. Moreover, if it goes down in the 8th step, then the route to the destination is again not blocked by anything. 

\begin{figure}[h]
\centering
\begin{subfigure}[b]{.48\textwidth}
\includegraphics[page=5, scale=.75]{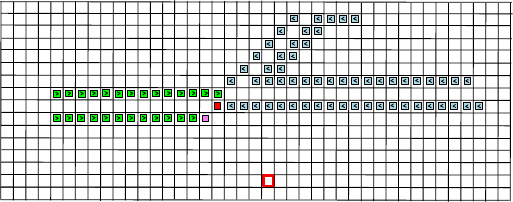}
\end{subfigure}
\begin{subfigure}[b]{.48\textwidth}
\includegraphics[page=6,scale=.75]{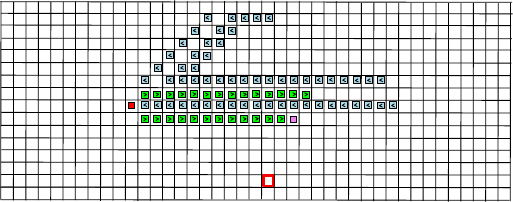}
\end{subfigure}

\caption{Clause gadget. A clause robot going left after steps 1 and 8.}
\label{fig:mapf_clause_going_left}
\end{figure}

Finally, the clause robot can either go immediately up, or go one step right and then up. After that it can either move another 6 steps up, or move left together with the stream of blue robots. If the robot finishes seven steps going up, then it needs to do an additional 14 steps going down and (at least) 4 steps going right. This sums up to 25 steps. Therefore, the robot in this case can wait only once, otherwise, it has to move either right or down (besides the 7 steps up). Hence, in this case it is indeed after 9 steps at position $\Omega(C)+(1,7)$ or $\Omega(C)+(2,7)$, since the blue stream forces the robot to go at least 2 steps (or 3 if it does not wait) right before it goes down. After that, the robot can, for example, finish all 4 steps right and continue going down. If it did not wait until then, then it has to wait for one step, after 5 steps down, to let the top blue stream pass. However, after waiting one step it just passes between the blue and green streams. See Figure~\ref{fig:mapf_clause_going_up_correct} for an illustration of this case.

\begin{figure}[h]
\centering
\begin{subfigure}[b]{.48\textwidth}
\includegraphics[page=7, scale=.75]{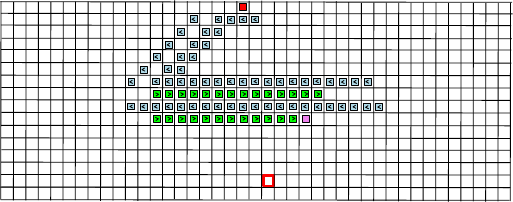}
\subcaption{After 9 steps without waiting.}
\end{subfigure}
\begin{subfigure}[b]{.48\textwidth}
\includegraphics[page=8, scale=.75]{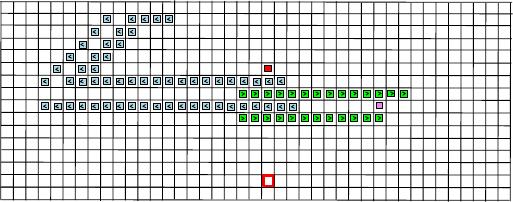}
\subcaption{After 16 steps without waiting.}
\end{subfigure}
\begin{subfigure}[b]{.48\textwidth}
\includegraphics[page=9, scale=.9]{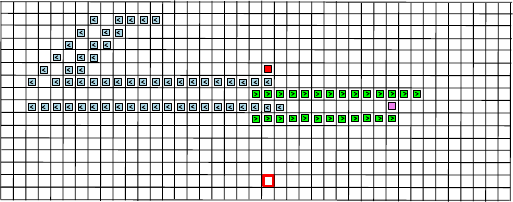}
\subcaption{After 17 steps with waiting one step.}
\end{subfigure}

\caption{Clause gadget. A clause robot going up.}
\label{fig:mapf_clause_going_up_correct}
\end{figure}   

If the robot does not complete seven steps going up, then it can follow the blue stream going left, and at some point eventually turn down. Note that it can only turn down if the robot went one step right and then only one step up, or if it went at most two times up with the first up being in the first step. The important observation here is that the robot cannot wait in any step while it is following the blue stream. Hence, it cannot wait until it returns to the $y$-coordinate of $\Omega(C)$. Moreover, it can only return to this coordinate at the left end of the top green stream. Since there are 14 green robots, it follows by a parity argument that the clause robot cannot return to the position immediately behind the last green robot in the top green stream (the grid is a bipartite graph and the last green robot and the clause robot start at vertices in the same part of the bipartition). However, at the beginning, the position two steps behind the last green is at distance 27 from the destination of the clause robot, and hence at step $i$, it is at distance at least $27-i$ from the destination. Any place to the left of this position is at an even larger distance. Hence, if the robot starts by going up, it is not possible for it to reach the destination without going 7 steps up. See Figure~\ref{fig:mapf_clause_going_up_wrong} for the two closest steps when the clause robot can return to the original $y$-coordinate.  
\end{proof}

\begin{figure}[h]
\centering
\begin{subfigure}[b]{.48\textwidth}
\includegraphics[page=10, scale=.75]{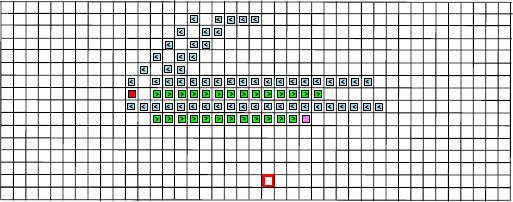}
\end{subfigure}
\begin{subfigure}[b]{.48\textwidth}
\includegraphics[page=11, scale=.75]{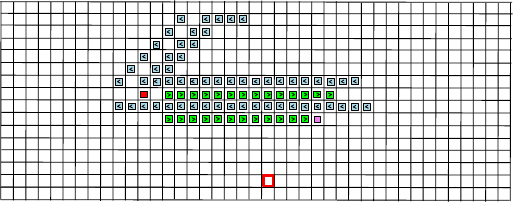}
\end{subfigure}
\caption{Clause gadget. The closest position to the destination where the robot can return to its original $y$-coordinate is after 9 steps if it went up and left (left subfigure); the closest position to the destination where the robot can return to its original $y$-coordinate is after 10 steps if it went one right, one up, and then left (right subfigure).} 
\label{fig:mapf_clause_going_up_wrong}
\end{figure}

The three different positions highlighted in Lemma~\ref{lem:clause_robot_possitions} will be precisely the points that will be intersected by the routes of the robots in a variable gadget, if the assignment represented by the choice of the robot in this gadget does not satisfy the clause (similarly as in the \textsc{$d$-Bounded Length Vertex Disjoint Paths} problem on grids).

\subsubsection{Variable Gadgets}
 
 The variable gadget is again very flexible. Each variable gadget will have one \emph{variable robot} that needs to move from its starting position $s$ to position $s+(dx,dy)$ such that $|dx|+|dy|=26$ and $1\le dx, dy \le 25$. If $dx\ge 2$ and $dy\ge 2$, the gadget will also contain a stream of robots called \emph{arrow} that forces the robot to chose one of the two directions towards the destination in the first step and then stick to this selection for a number of steps that depends on the size of the arrow. For example, in Figure~\ref{fig:mapf_var_gadget_example}, we can see an orange variable robot that wants to go 7 steps down and 19 steps left. There is a so-called ``right arrow'' of green robots that all want to go 26 steps right. The right arrow forces the variable robot to either take the first 5 steps all to the right, or the first 7 steps all down. Note that after 5 steps to the right, the arrow does not restrict the variable robot from going down. In our gadgets, the diagonal part of the arrow will usually force the robot to take all the steps in a given direction, while the straight part will often be shorter. 
\begin{figure}[h]
\centering
\includegraphics[page=12, scale=2]{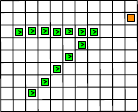}
\caption{An example of a variable gadget.} 
\label{fig:mapf_var_gadget_example}
\end{figure} 

Chaining the variable gadgets is rather simple and depicted in Figure~\ref{fig:mapf_var_chaining}. The two orange variable robots want to go seven steps to the right and 19 steps down. The purple variable robot wants to go seven steps left and 19 up. Clearly, if the orange robot on the left chooses to start by going right, it forces the purple robot to start by going up, and this subsequently forces the right orange robot to start by going right as well. On the other hand, if the orange robot on the right starts by going down, it forces the purple robot to start by going left, and this subsequently forces the left orange robot to start by going down as well. However, this is not the only way that variable gadgets can be chained. Sometimes, to avoid some collisions, it could be beneficial to chain arrows going in the same direction. See Figure~\ref{fig:mapf_var_chaining_alternative} for an example of chaining up-arrows. In this case where a variable robot of some arrow chooses to go right and the variable robot of the next arrow chooses to start down, then these two intersect after 7 steps. Here, it is important that the difference of the $x$-coordinates and the $y$-coordinates of two consecutive variable robots is the same (it is 7 in the example).

\begin{figure}[h]
\centering
\includegraphics[page=13, scale=2]{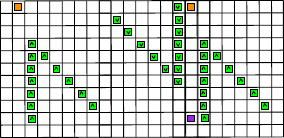}
\caption{Example of a variable gadgets connecting.}
\label{fig:mapf_var_chaining}
\end{figure} 

\begin{figure}[h]
\centering
\includegraphics[page=14, scale=2]{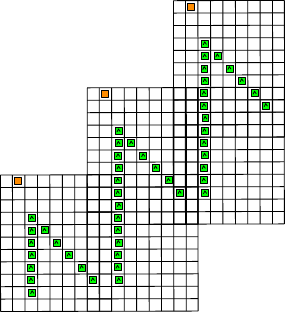}
\caption{Alternative example of a variable gadget moving diagonally.} 
\label{fig:mapf_var_chaining_alternative}
\end{figure} 

Changing the orientation from left/right arrows to up/down arrows in the corners is a little tricky, but it can be done with the help of a variable gadget with either $dx=1$ or $dy=1$ that does not need an arrow. An example of using such a variable gadget in a chain can be seen in Figure~\ref{fig:mapf_var_chaining3}. There, both purple robots want to go 12 steps down and 14 steps left. The light blue robot wants to go one step up and 25 steps right. If the bottom purple robot chooses to first go left, then the light blue robot is forced to move up in the first step, and the top purple robot is forced to choose to go left as well. On the other hand, if the top robot chooses to start by going down, then the light blue robot cannot do its one up move until it passes the top purple robot, and the bottom purple robot is forced to choose down as well.  An example of the corner case is seen in Figure~\ref{fig:mapf_var_corner}. The orange robot wants to go 12 steps right and 14 steps up, and the purple robot wants to go 12 steps down and 14 steps left. If the orange robot goes first 12 steps right and the purple robot goes 12 steps down, then all their remaining moves are forced, and in step 24, they are forced to either wait or collide. The light blue and pink robots are the special robots without an arrow as in Figure~\ref{fig:mapf_var_chaining3}. 

\begin{figure}[h]
\centering
\includegraphics[page=17]{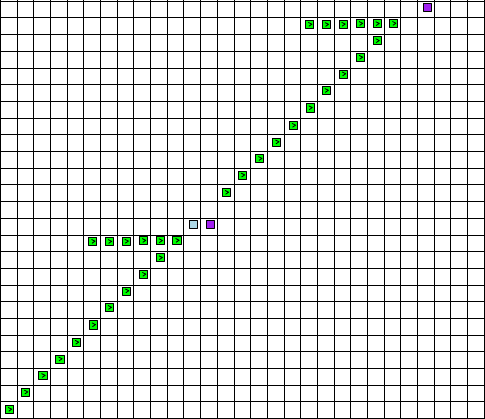}
\caption{Alternative example of a variable gadget chaining using a light blue robot that wants to go one up and 25 right.} \label{fig:mapf_var_chaining3}
\end{figure} 

\begin{figure}[h]
\centering
\includegraphics[page=16]{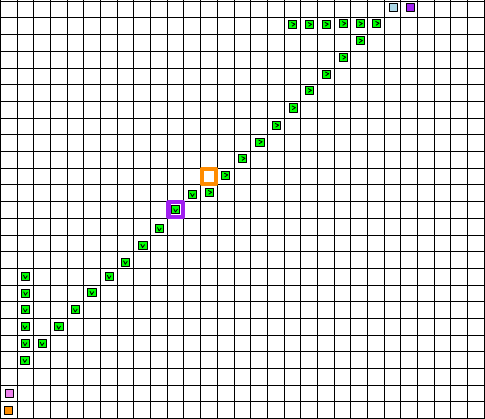}
\caption{Variable gadgets connecting at a corner. A similar construction is also used when connecting to a variable cycle.} \label{fig:mapf_var_corner}
\end{figure}

\subsubsection{Connecting Gadgets}

The connection of a clause gadget for clause $C$ and a variable gadget for the corresponding variables is depicted in Figure~\ref{fig:mapf_clause_connection}. There are three variable gadgets in the figure; we will refer to them as left, right, and top, depending on which side of the clause gadget the variable robot starts. 

\begin{figure}[h!]
\centering
\includegraphics[page=15, scale=1.5]{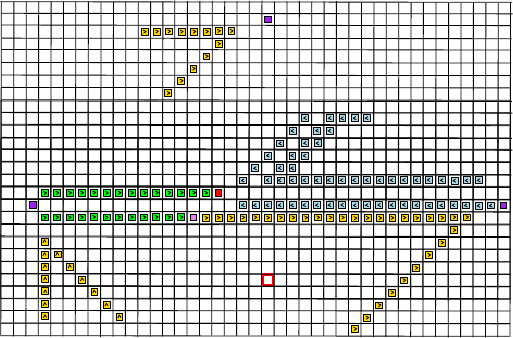}
\caption{Connecting a clause gadget with variable gadgets. We can use alternative (diagonal) chaining of variable gadgets for the right variable in order to not interfere with the variable gadgets for the left variable.} 
\label{fig:mapf_clause_connection}
\end{figure} 

\subparagraph{Left Variable Gadget.} The variable robot on the left starts at position $\Omega(C)+(-15,-1)$ and wants to go 8 steps to the right and 18 steps down. The up-arrow forces this robot to start by either moving 8 steps right or 5 steps down. If it chooses to go right, the variable robot is after 8 steps at position $\Omega(C)+(-7,-1)$, and by Lemma~\ref{lem:clause_robot_possitions}, it would intersect the clause robot of $C$ if the robot's route conforms to case (2) in the lemma. The choice will be propagated by the chaining depicted in Figure~\ref{fig:mapf_var_chaining} such that there will be a down-arrow to the left of the variable gadget and the choice ``up'' in that arrow will intersect the choice ``down'' of the connection variable gadget within the first 5 steps.  It is important to mention that it now comes into play the gap created in the blue stream of the clause gadget that allows the clause gadget to go first one step to the right and only then up since the the rightmost robot of the up-arrow will precisely fit in that gap.

\subparagraph{Right Variable Gadget.} The variable robot on the right starts at position $\Omega(C)+(23,-1)$ and wants to go 12 steps left and 14 steps down. The right arrow below it forces the robot to either first do 12 steps to the left or first do 11 steps down. If it chooses to go 12 steps left, then it will be after the 12th step at position $\Omega(C)+(11,-1)$, and by Lemma~\ref{lem:clause_robot_possitions}, it would intersect the clause robot of $C$ if the robot's route conforms to case (1) in the lemma. To propagate the information, we will start with chaining diagonally down and right by putting another variable gadget with variable robot at $\Omega(C)+(34,-12)$ with right-arrow that forces it to go at least 11 steps left or some number of steps down. This allows us to not interfere with the arrows on the left. 

\subparagraph{Top Variable Gadget.} The variable robot in the top starts at position $\Omega(C) + (4,14)$ and wants to go 7 steps down and 19 steps left. It is forced by a right-arrow to either start with 7 steps down or with 5 steps left. If it starts down, then after the 7 steps down all the steps have to be to the left. Hence, after 9 steps, it is at position $\Omega(C) + (2,7)$ and in the 10th step it goes left. Therefore, by Lemma~\ref{lem:clause_robot_possitions}, it would intersect the clause robot of $C$ if the robot's route conforms to case (3) in the lemma. The choice is again propagated by the chaining depicted in Figure~\ref{fig:mapf_var_chaining}. One thing worth noticing is that, if the variable robot starts by going left, then going straight left for 19 steps would intersect the up arrow of the left variable gadget. However, after 5 steps left, it is actually free to take the 7 steps down and then to continue the remaining 14 steps left. This does not intersect the up-arrow on the left.

\begin{figure}[h]
\centering
\includegraphics[page=18, scale=1.1]{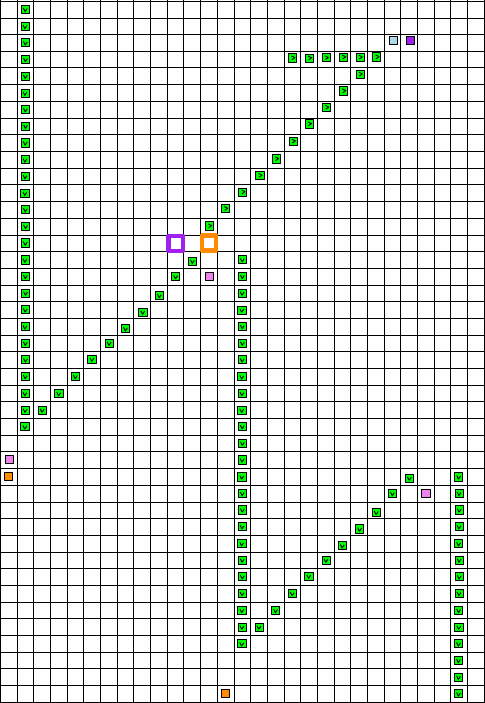}
\caption{Connecting a variable cycle with the connection gadget. The choice ``up'' on the leftmost variable gadget of the cycle intersects with the choice ``down'' on the first variable gadget of the connection gadget.} 
\label{fig:mapf_var_cycle_connection}
\end{figure} 

\subparagraph{Connecting to a Variable Cycle.} To connect to a variable cycle, we need a way to distinguish between a variable appearing positively or negatively in the cycle. Both of these connections are very similar to how we change the direction on the variable cycle depicted in Figure~\ref{fig:mapf_var_corner}. In this figure, on the left, we have a stream of up and down arrows that connect to a stream of left and right arrows. To change this corner gadget to a connection between a connection gadget and a variable cycle,  we need to add a stream of up and down arrows going right. This is done again by using the variable gadgets without arrows. How this can be done precisely is depicted in Figure~\ref{fig:mapf_var_cycle_connection}. All orange variable robots want to go 12 steps right and 14 steps up. The leftmost robot is forced to either start with 12 steps left or with 14 steps up. The purple robot wants to go 12 steps down or 14 steps left. It is forced to start by either going 14 steps down or 3 steps left. The pink robots all want to go one step right and 25 steps down. If the purple robot chooses to go down, then the leftmost orange robot cannot choose to go up, as they would intersect after 25 steps.

\begin{theorem}
\label{thm:cmpmhard}
\cmpm\ is \NP-hard even when restricted to instances where $\ell=26$. 
\end{theorem}

\begin{proof}
We summarize the proof of the theorem. We start from an instance $\varphi$ of \BPSAT\ with a set of variables $X$, a set of clauses $\mathcal{C}$, and the incidence graph $G_{\varphi}$. We follow the construction described above and construct an instance $(G, \R, 26)$ of \cmpm. Note that this construction takes polynomial time, and the number of grid-points of $G$ is $\bigoh((n+m)^2)=\bigoh(n^2)$, where $n= |X|$ and $m=|\mathcal{C}|$. We prove that $\varphi$ is satisfiable if and only if there is a schedule for $\R$ of makespan at most $26$. The proof is now basically identical to the proof of Theorem~\ref{thm:disjoint_paths_hard}. 

First, let $\alpha: X \rightarrow \{\text{True, False}\}$ be a satisfying assignment for $\varphi$. We again construct the variable cycles and the connection gadgets so that each variable gadget is associated with a variable. Moreover, we can set routes for all the variable gadgets for a variable $x$ in a way that they are all consistent among themselves and they block the route of Lemma~\ref{lem:clause_robot_possitions} for a clause $C$ if and only if the negation of $\alpha(x)$ satisfies $C$. Since every clause is satisfied by $\alpha$, setting the routes for variable gadgets this way allows us to always select one of the routes of Lemma~\ref{lem:clause_robot_possitions} for the clause robot that does not collide with a route of any other robot.

For the converse, let us assume that we have a valid schedule for all the robots. From the construction of the variable cycle, it follows that there are only two ways of setting routes for all variable gadgets on the variable cycle, since for a fixed variable gadget on the cycle, the route of the variable robot either forces a choice for the variable robot in the consecutive variable gadget in the clockwise direction around the cycle, or the one in the counter-clockwise direction. Hence, if two consecutive variable gadgets on a cycle are not consistent, the variable robots for these gadgets either collide, or they force opposite choices to be propagated on the cycle in the two directions going away from these two gadgets. Hence, there has to be a collision somewhere on the cycle. During the construction of the variable gadgets, we assigned one of these two settings to represent assigning the associated variable True and the other setting to represent False. This gives us the truth assignment of the variable and we only need to verify that this assignment satisfies all the clauses. By Lemma~\ref{lem:clause_robot_possitions} and the construction of connection gadgets, for each clause, the clause robot has to intersect one of the two possible routes of one of the three variable gadgets associated with the three variables whose literals are in the clause. If the route of the clause robot intersects one of the possible routes of the variable robot associated with variable $x$, then the variable robot has only one choice for its route and this choice is propagated for all variable robots in variable gadgets on the connection gadget between the clause gadget and the variable cycle. Finally, the connection between the last variable gadget in the connection gadget and the variable cycle is set in a way that this forced route for the variable robot at this end of the connection gadget intersects a route of some variable robot on the variable cycle that represents setting the variable $x$ in a way that it does not satisfy $C$. Hence, this variable robot on the variable cycle (and hence all variable robots on the variable cycle) had to choose the route that represents setting the truth value of the variable $x$ such that it satisfies $C$. Therefore, all clauses of $\varphi$, and hence $\varphi$, are satisfied by this assignment. 
\end{proof}
\fi

\section{Conclusion}
\label{sec:conclusion}

In this work, we settled the parameterized complexity of both \cmpm and \cmpl with respect to their two most fundamental parameters: the number of robots, and the objective target.
Along the way, we established the \NP-hardness of the classical Vertex Disjoint Paths and the Edge Disjoint Paths problem with constant path-lengths on grids, strengthening the existing lower bounds for these problems as well. Our results reveal structural insights into the properties of optimal solutions that may also prove useful in contexts that lie outside of this work. We conclude by stating two open questions that arise from our work.

\begin{enumerate}
\item What is the parameterized complexity of other variants of CMP, such as the ones where the objective is to minimize the maximum length traveled or the total arrival time?

\item Can the fixed-parameter tractability of \cmpm or \cmpl parameterized by the number $k$ of robots be lifted to grids with obstacles/holes, or more generally to planar graphs? It is worth noting that neither the structural results developed in this paper, nor other known techniques~\cite{lokshtanov}, seem to be applicable to these more general settings.
\end{enumerate}

\bibliography{ref}

\end{document}